\documentclass{article}

\usepackage[english]{babel}

\usepackage[letterpaper,top=2cm,bottom=2cm,left=3cm,right=3cm,marginparwidth=1.75cm]{geometry}

\usepackage{amsfonts,amsmath,amssymb,amsthm}
\usepackage{algorithm}
\usepackage{algpseudocode}
\usepackage{graphicx}
\usepackage{float}
\usepackage{placeins}
\usepackage[colorlinks=true, allcolors=blue]{hyperref}
\usepackage[style=nature,urldate=long]{biblatex}
\DeclareFieldFormat[online]{urldate}{\mkbibparens{%
  \iffieldundef{year}{}{\printfield{year}\addsemicolon\space}%
  accessed\space#1}}
\usepackage{cleveref}
\usepackage{parskip}
\usepackage{booktabs,array}
\usepackage{mathtools}
\usepackage{enumerate}
\usepackage{authblk}
\usepackage{abstract}
\usepackage{titlesec}
\usepackage[font=small,labelfont=bf]{caption}
\usepackage{lineno}

\titleformat{\section}{\large\bfseries}{\thesection}{1em}{}
\titleformat{\subsection}{\normalsize\bfseries}{\thesubsection}{1em}{}

\newtheoremstyle{mystyle}
  {\parskip}
  {}
  {\itshape}
  {}
  {\bfseries}
  {.}
  { }
  {\thmname{#1}\thmnumber{ #2}\thmnote{ (#3)}}

\theoremstyle{mystyle}
\newtheorem{theorem}{Theorem}

\newtheorem{corollary}{Corollary}[theorem]
\newtheorem{lemma}{Lemma}
\newtheorem{definition}{Definition}[section]
\newtheorem{proposition}{Proposition}[section]

\theoremstyle{remark}

\usepackage{xpatch}
\xpretocmd{\eqref}{Eq.~}{}{}

\newcommand{\abs}[1]{\left\lvert #1 \right\rvert}
\newcommand{\norm}[1]{\left\lVert #1 \right\rVert}
\newcommand{\KL}[2]{\operatorname{KL}(#1 \Vert #2)}

\newcommand{\figOneLegend}{%
\textbf{A quadratic framework explains the sensitivity of spatial pattern tests.}
\textbf{a,} Spatial variability tests share the quadratic form \(Q_n=\mathbf z^\mathsf{T}\mathbf K\mathbf z\), where \(\mathbf z\) is standardized expression. Tests differ primarily in their choice of kernel \(\mathbf K\), which determines their behavior and consistency.
\textbf{b,} Synthetic mean- and variance-shift patterns evaluated with Moran, Gaussian and Matérn kernels. Mean-preserving variance shifts are not consistently detected, although non-zero finite-sample power is possible. The right panel shows standardized test scores with vertical bars denoting the central 95\% Gaussian null interval, estimated from 50,000 draws with empirical variance standardization.
\textbf{c,} An observed spatial pattern can be expanded exactly in a harmonic basis, \(f=\sum_i\hat f_i\phi_i\). The coefficient plot shows the contribution of each eigenmode; selected basis patterns are shown below.
\textbf{d,} Zero-weight modes produce blind spots, whereas positive and negative spectral weights permit cancellation.
\textbf{e,} The illustrated Moran spectrum assigns opposite signs to low- and high-frequency modes. Their contributions cancel exactly, producing a nonconstant pattern with \(Q=0\).
\textbf{f,} Illustrative CAR continuum and graph-Laplacian spectra emphasize broad trends and local variation, respectively.
\textbf{g,} The bivariate \(R\)-statistic measures spatial association. Its normalized coefficient, \(r_K=R/\sqrt{Q_xQ_y}\), is shown alongside Pearson's correlation. Nonnegative spectral weights give \(R\) a kernel inner-product interpretation and remove cancellation caused by opposite-signed eigenvalues.%
}

\newcommand{\figTwoLegend}{%
\textbf{Scalable spatial pattern detection in million-spot datasets.}
\textbf{a,} Computational strategies for efficient quadratic spatial tests. Moment matching avoids eigendecomposition, sparse CAR precision matrices enable fast quadratic evaluation, and fast Fourier transform (FFT) and non-uniform FFT (NUFFT) methods evaluate stationary kernels at near-linear complexity. By contrast, sparse graph kernels or low-rank approximations risk losing power on specific patterns.
\textbf{b,} Single-core wall time and peak memory for kernel construction, null-moment calculation and a single-feature test as a function of sample size. Values show median runtime and maximum resident memory across three fresh-process repetitions. Benchmarks were run on an M5 MacBook Air under a 16\,GB memory budget. CAR switches to an implicit sparse-solve implementation above \(n=5{,}000\).
\textbf{c,} Projected runtime for screening 10,000 features using up to eight CPU cores and a 16\,GiB total process-memory budget. Dashed lines mark 1\,min and 1\,h. Crosses denote configurations unsupported because of geometry or projected memory use.
\textbf{d,} Spatial expression of \textit{Reg3b}, marking a putative sample-preparation artifact, and \textit{Coro1a}, marking an immune-cell cluster, in a mouse small-intestine Visium HD sample.
\textbf{e,} Standardized CAR-\(Q\) scores obtained with slow-decay (\(\rho=0.70\)) and fast-decay (\(\rho=0.95\)) kernels. Highlighted genes represent the artifact and immune patterns in \textbf{d}.
\textbf{f,} L1-normalized radial power profiles of the highlighted genes and theoretical CAR spectra. Faint lines show individual genes; colored solid lines show group means and neutral-gray dashed lines show the theoretical spectra.
\textbf{g,} Median transcriptome ranks of the artifact and immune gene sets across CAR autocorrelation values and two contrastively fitted spectral kernels. Rank 1 denotes the strongest spatial variability. Diamonds show results from the fitted kernels; annotations report fitting time and storage.%
}

\newcommand{\figThreeLegend}{%
\textbf{Spectral cancellation obscures lineage-associated expression.}
\textbf{a,} Graph-based pattern detection in spatial and single-cell data.
\textbf{b,} Standardized Moran-\(I\) and CAR-\(Q\) statistics for 9,757 genes in tumor 3430\_NT\_T2. Genes significant by both tests or only by CAR at FDR\,\(\leq0.05\) are colored; dashed lines denote the significance cutoffs.
\textbf{c,} Barcode-derived phylogeny of tumor 3430\_NT\_T2. The inner ring denotes cell state; the remaining rings show log-normalized expression of \textit{Rgs4}, \textit{Col4a2}, \textit{Tmem37} and \textit{Cd200}. Branches leading to AT2-like and gastric-like clones are highlighted.
\textbf{d,} Distributions of log-normalized expression by cell state for the four genes in \textbf{c}.
\textbf{e,} Eigenmode contributions, \(\lambda_i\hat f_i^2\), to Moran's \(Q\) for \textit{Rgs4}.
\textbf{f,} Moran eigenvalue versus effective eigenmode support, \(1/\sum_j\phi_{ji}^4\), where \(\phi_{ji}\) is the value of unit-length eigenmode \(i\) at cell \(j\).
\textbf{g,} Left, fraction of spectral mass cancelled for genes significant by both tests, genes significant only by CAR and minor-clone markers. Markers were identified by one-sided Welch tests against all remaining cells at FDR\,\(\leq0.05\). Boxes show the median and interquartile range (IQR); whiskers extend to 1.5\,\(\times\)\,IQR. Right, Spearman agreement between Moran and CAR gene rankings across expression-sparsity quartiles. Gray lines represent 17 tumors; the blue line and band show the median and IQR.
}

\newcommand{\figFourLegend}{%
\textbf{Frequency-domain testing separates global and localized responses in Alzheimer's disease.}
\textbf{a,} Differential-frequency (DF) analysis of four wild-type and four \(\mathrm{App}^{NL-G-F}\) mouse-brain sections. Mean-centered expression maps are transformed into radially binned 1D or polar-binned 2D power spectra, using zero spatial frequency as the origin. A log-spectral quadratic statistic compares group-mean spectra using a covariance-aware weighted chi-square null. DF captures differences in spectral amplitude and composition; DF-norm rescales profiles to unit total power to isolate changes in relative frequency composition. Optional background normalization or covariate adjustment removes unwanted variation, analogous to controls used in DE analysis.
\textbf{b,} DF versus differential-expression (DE) significance for 3,020 genes. DE was assessed by two-sided Welch tests on mean expression. Dashed lines denote nominal \(P=0.05\); panel counts report discoveries at FDR\,\(\leq0.05\). Plaque-induced genes and selected DF hits are highlighted; \(\rho\) is Spearman's correlation.
\textbf{c,} Representative expression maps and radial power profiles for \textit{Apoe}, showing an overall abundance shift. Faint curves show individual samples and bold curves show group means. Bottom plots show log2 ratios of group-mean L1-normalized profiles; shading denotes frequencies enriched in either genotype.
\textbf{d,} DF versus DF-norm significance, formatted as in \textbf{b}. Selected DF- and DF-norm-preferred genes are labeled.
\textbf{e,} GO Biological Process and KEGG enrichment among the top 100 genes ranked by DF or DF-norm. Bars show \(-\log_{10}(\mathrm{FDR})\); numbers give the intersection with each top-100 set.
\textbf{f,} Expression maps and radial power profiles for \textit{Vamp2}, formatted as in \textbf{c}, showing a shift in relative frequency composition. Unless otherwise stated, displayed \(P\) values are unadjusted. DE \(P\) values are two-sided; DF and DF-norm \(P\) values use the upper tail of their null distributions.%
}

\newcommand{\figFiveLegend}{%
\textbf{A joint spectral atlas maps spatial gene-expression dynamics across mouse and human embryogenesis.}
\textbf{a,} Principal-component embedding of 529,092 gene--section radial spectra from 30 mouse and 28 human embryo sections. Spectra were L1-normalized, log-transformed and standardized before principal-component analysis (PCA). A random subset of at most 120,000 patterns is displayed and colored by five clusters fitted in the first five principal components; white circles denote cluster centers. Cluster labels summarize biological processes enriched in both species.
\textbf{b,} The same embedding colored by local gene density, species and Carnegie stage, spatial-variability percentile and expression percentile.
\textbf{c,} Hypergeometric-corrected local mixing scores for species, developmental time, expression quintile and spatial-variability quintile. Scores near 0 indicate local segregation by the corresponding attribute, whereas scores near 1 indicate mixing at or above the null expectation. Each point is the section-level mean across genes (\(n=58\) sections per comparison).
\textbf{d,} Trajectory of mouse \textit{Col11a1} across developmental stages in the joint embedding (left) and representative spatial expression maps from one section per stage (right).
\textbf{e,} Retrieval of related mouse and human patterns using a mouse \textit{Col11a1} query. Euclidean distances were calculated in the five-dimensional embedding. The four examples show cross-species mesenchymal patterns for \textit{Sfrp2/SFRP2} and liver patterns for \textit{Apoa2/APOA2}. Expression scales are independently capped at their respective 99th percentiles.%
}

\newcommand{\edFigOneLegend}{%
\textbf{Spatially variable ANKRD30A codewords reveal distinct luminal and invasive transcript classes.}
\textbf{a,} Atera human breast-cancer tissue segmented by cell state. The 12 most abundant states are colored; remaining cells are gray.
\textbf{b,} Distribution of SplisosmNP HSIC-IR \(P\) values for 4,361 genes. Of these, 243 passed FDR\,\(\leq0.01\); seven \(P\) values underflowed to zero and are displayed at the plotting floor. Selected genes and \textit{ANKRD30A} are labeled.
\textbf{c,} Log1p expression for four \textit{ANKRD30A} codewords. Each map uses an independent scale capped at its 99.5th percentile.
\textbf{d,} Compressed view of the \textit{ANKRD30A} locus showing GENCODE v49 transcripts, Atera probe regions and selected quality-controlled PacBio isoforms from an external long-read bulk RNA-sequencing cohort. Orange marks the 45213\,+\,50762 luminal class and blue the 50762\,+\,7939 invasive class. PacBio labels report predicted open-reading-frame length, total full-length read support in ER-positive samples and the number of supporting samples.%
}

\newcommand{\edFigTwoLegend}{%
\textbf{TRPS1 controls alternative transcription-start-site usage at ANKRD30A.}
\textbf{a,} Separate windows around four PacBio-defined \textit{ANKRD30A} transcription start sites (TSSs), showing ReMap transcription-factor (TF)-binding sites. Coordinates are relative to each TSS, with absolute hg38 positions indicated above. Genomic regions may appear in more than one window. Orange denotes luminal and blue invasive transcript classes. Core sites fall within 1\,kb of a TSS.
\textbf{b,} Conditional associations between \textit{ANKRD30A} differential usage and TF expression for the top de novo and TSS-bound candidates. Associations were tested using SplisosmNP HSIC-GP; bars show \(-\log_{10}P\) for unadjusted \(P\) values. L and I denote luminal- and invasive-specific ReMap peaks, with core peaks marked separately.
\textbf{c,} TF expression normalized to 10,000 counts per cell and log1p-transformed in invasive and luminal/DCIS cells. Lines within violins denote medians.
\textbf{d--f,} TCGA validation.
\textbf{d,} Counts for total \textit{ANKRD30A}, selected PacBio transcripts and Atera probe-compatible regions in ER-positive and ER-negative tumors. Boxes show the median and IQR; whiskers extend to 1.5\,\(\times\)\,IQR.
\textbf{e,} Pairwise Spearman correlations between PacBio transcript usage and Atera probe-region usage in ER-positive tumors.
\textbf{f,} Correlations between candidate-TF expression and two \textit{ANKRD30A} usage ratios in ER-positive tumors; dark outlines indicate raw \(P<0.01\).
\textbf{g,} T47D dTAG-TRPS1 time course showing normalized PRO-seq signal within 1\,kb of each PacBio TSS and across the \textit{ANKRD30A} gene body. Values are log2 fold changes from 0\,h; downward triangles denote zero signals displayed at the lower plotting limit. Acute TRPS1 degradation increased transcription at the invasive start sites, supporting their repression by TRPS1.%
}

\newcommand{\edFigThreeLegend}{%
\textbf{DF and DF-norm distinguish amplitude and pattern-shape changes during postnatal brain development.}
Seventeen anterior-plane sections were analyzed: P5, P7, P9 and P12 (\(n=3\) each), P21 (\(n=2\)) and P60 (\(n=3\)). Primary DF and DF-norm tests were performed on 2D-binned power spectra to retain anisotropy; results from radially binned profiles are also shown for comparison. Developmental age was modeled as the continuous covariate \(\log_{10}(\mathrm{day}+1)\).
\textbf{a,} DF versus DE significance for 3,075 genes. DE was assessed by linear regression of mean expression on age; DF used the corresponding continuous-design log-spectral test. Dashed lines denote nominal \(P=0.05\), whereas counts report discoveries at FDR\,\(\leq0.05\). Curated control genes among the top 100 DF hits are colored by biological group; \(\rho\) is Spearman's correlation.
\textbf{b,} DF versus DF-norm significance, formatted as in \textbf{a}.
\textbf{c,} GO and KEGG enrichment among the top 100 DF and DF-norm genes. Bars show \(-\log_{10}(\mathrm{FDR})\); numbers give the intersection with each top-100 set.
\textbf{d,e,} Representative maps at P5, P12 and P60, raw radial power profiles and normalized profile changes for \textit{Aplnr} (\textbf{d}, amplitude shift) and \textit{Synpr} (\textbf{e}, frequency shift). Faint curves show individual sections and bold curves show stage means. Bottom heatmaps show the log2 ratio of each stage-mean L1-normalized profile to the across-section mean. The 2D and radial \(P\) values shown in the panel titles are unadjusted.%
}

\newcommand{\edFigFourLegend}{%
\textbf{DF-norm reveals IDH-mutation-associated spatial reorganization in the glioma microenvironment.}
Eighteen glioma sections were analyzed: IDH-mutant (\(n=6\)) and IDH-wild-type (\(n=12\)).
\textbf{a,} DE versus DF-norm significance for 7,000 tested genes. The top 100 DF-norm genes are assigned to five literature-guided groups using their strongest significant DF-norm-preferred GO or KEGG term; unassigned genes are gray. Circled labels mark manually reviewed examples.
\textbf{b,} Selected GO and KEGG terms preferentially enriched among the top 100 DF-norm genes. Bars show \(-\log_{10}(\mathrm{FDR})\) and are colored by functional group. Labels give intersections with the DE and DF-norm top-100 sets.
\textbf{c--f,} Representative maps, raw radial power profiles and normalized profile differences for \textit{NDRG1} (\textbf{c}, hypoxia), \textit{C3} (\textbf{d}, complement), \textit{S100A10} (\textbf{e}, coagulation) and \textit{FN1} (\textbf{f}, extracellular-matrix remodeling). Top, log1p spatial expression in two representative sections per IDH group. Bottom left, unnormalized radial power profiles; faint curves show individual sections and bold curves show group means. Bottom right, log2 ratios of the IDH-wild-type to IDH-mutant group-mean L1-normalized profiles; shading denotes frequencies enriched in either group. DE was assessed by two-sided Welch tests on mean expression. \(P\) values in \textbf{a,c--f} are unadjusted.%
}

\newcommand{\edFigFiveLegend}{%
\textbf{Shared and species-preferred temporal variation in spatial expression patterns during embryogenesis.}
Variation was calculated in the first five principal components of the joint spectral embedding, and within-stage variation was estimated using sections from the same stage (mouse \(n=4,4,6,4,7,5\); human \(n=4\) per stage).
\textbf{a,b,} Within-stage versus between-stage mean-square variation for mouse (\textbf{a}) and human (\textbf{b}) genes. Genes above the species-specific 95th percentile of within-stage variation were excluded from subsequent comparisons; dotted lines show these cutoffs.
\textbf{c,} Mouse and human between-stage-variation percentiles for 6,146 strict one-to-one orthologs. The shared dynamic set contained 257 genes ranking in the top decile in both species, 4.2-fold more than expected by chance (one-sided hypergeometric test, \(P=1.16\times10^{-112}\)). A further 358 genes ranked in the top decile only in mouse and 358 only in human, defining the mouse-preferred and human-preferred sets.
\textbf{d,} GO enrichment for the shared dynamic and human-preferred gene sets. No mouse-preferred term passed the significance criteria. Points show \(-\log_{10}(\mathrm{FDR})\); labels give fold enrichment over chance.
\textbf{e,} Spectral-pattern trajectories for \textit{Col11a1/COL11A1} (shared dynamic), \textit{Copg2/COPG2} (mouse-preferred) and \textit{Eng/ENG} (human-preferred). Mouse stages are connected by solid lines and circles; human stages by dashed lines and triangles. Color intensity denotes developmental progression.
\textbf{f,} Representative spatial maps for the same genes, with one section shown per stage. Normalized expression is capped at the species-specific 99th percentile.%
}

\newcommand{\suppFigOneTitle}{Validation of scalable \(Q\)-test implementations.}
\newcommand{\suppFigOneLegend}{%
\textbf{\suppFigOneTitle}
\textbf{a,} Simulated Mat\'ern Gaussian-process fields on regular and irregular coordinates, with agreement among matrix, FFT and NUFFT evaluations of \(Q\). Matrix and FFT implementations agree for matching kernels under periodic boundary conditions. Differences can arise from boundary mismatch. For graph-based kernels such as CAR and Moran's \(I\), the NUFFT approximation uses a fundamentally different discrete operator because oversampling changes the effective spot-to-spot connectivity. Dashed lines mark \(y=x\); relative errors summarize agreement between the displayed scores.
\textbf{b,} Null calibration, including agreement between Welch-Satterthwaite and four-moment \(P\) values (left) and false-positive rates at nominal \(P<0.05\) across 2,000 null simulations (right). Error bars show approximate 95\% binomial confidence intervals. Gaussian, Mat\'ern and CAR tests use four-moment calibration; Moran's \(I\) uses a normal approximation.
\textbf{c,} Detection power across signal-to-noise ratios, implementations and kernel hyperparameters. Target patterns were simulated from a Mat\'ern field with added Gaussian noise (200 simulations per setting). Rows compare implementations on regular coordinates, irregular coordinates and across kernel-parameter settings. M, matrix; F, FFT; N, NUFFT; CLT, normal null approximation.%
}

\newcommand{\suppFigTwoTitle}{Scalable detection of spatial isoform usage across Visium HD assays.}
\newcommand{\suppFigTwoLegend}{%
\textbf{\suppFigTwoTitle}
SPLISOSM analyses of mouse coronal brain sections compare transcript diversity measured by Oxford Nanopore (ONT) sequencing, 3$^{\prime}$-end short-read sequencing of fresh-frozen tissue, and targeted probes in formalin-fixed, paraffin-embedded (FFPE) tissue.
\textbf{a,} Runtime and peak memory for setup and spatial-variability testing at 16-\(\mu\)m resolution, together with test runtime at 16, 8 and 2\,\(\mu\)m.
\textbf{b,} Pairwise agreement of HSIC-IR \(P\) values across resolutions for each assay; \(\rho\) is Spearman's correlation. Stacked bars show significant genes detected at all three resolutions or at only a subset.
\textbf{c,} Cross-assay agreement at 16\,\(\mu\)m. Left, overlap of spatially variable proportion (SVP) gene sets with ONT discoveries. Middle, pairwise \(P\)-value correlations stratified by expression quintile. Right, pairwise \(P\)-value comparisons. SVPs were called at FDR\,\(\leq0.01\) after Benjamini--Hochberg correction.%
}

\newcommand{\suppFigThreeTitle}{Spatial \textit{Mbp} isoform usage and candidate regulators across assays.}
\newcommand{\suppFigThreeLegend}{%
\textbf{\suppFigThreeTitle}
\textbf{a,} Runtime and peak memory for conditional differential-usage (DU) analysis at 16-\(\mu\)m resolution, together with significant DU gene--regulator pairs shared across assays. Purple and green bars denote RNA-binding proteins and transcription factors, respectively. Significance was defined at FDR\,\(\leq0.01\).
\textbf{b,} Log1p-transformed counts for selected \textit{Mbp} transcripts, TREND peaks and targeted FFPE probes, alongside log1p-transformed expression of \textit{Myt1l}, \textit{Egr1} and \textit{Sox10}. Regulator labels report FDR values from conditional HSIC-GP tests.
\textbf{c,} Genomic positions of detected \textit{Mbp} transcripts, TREND peaks, targeted FFPE probes and ReMap binding sites for candidate transcription factors.%
}

\newcommand{\suppFigFourTitle}{Spatial \textit{Gabbr1} isoform usage and candidate regulators across assays.}
\newcommand{\suppFigFourLegend}{%
\textbf{\suppFigFourTitle}
\textbf{a,} Log1p-transformed counts for selected \textit{Gabbr1} transcripts, TREND peaks and targeted FFPE probes, alongside expression of the candidate regulators \textit{Rbfox3} and \textit{Tcf7l2}. Red circles mark the medial habenula.
\textbf{b,} Expression of additional candidate RNA-binding proteins (\textit{Khdrbs3}, \textit{Celf5} and \textit{Elavl4}) and transcription factors (\textit{Shox2}, \textit{Satb1} and \textit{Egr1}) in the FFPE section. Labels in \textbf{a,b} report FDR values from conditional HSIC-GP tests.
\textbf{c,} \textit{Gabbr1} transcript structures, TREND peaks, targeted FFPE probes and regulator-binding tracks from POSTAR3 and ReMap. Vertical guides mark alternative transcription start sites.%
}

\newcommand{\suppFigFiveTitle}{Spectral cancellation alters lineage-gene rankings and coexpression modules.}
\newcommand{\suppFigFiveLegend}{%
\textbf{\suppFigFiveTitle}
\textbf{a,} Spearman agreement between Moran's \(I\) and CAR gene scores, stratified by Moran significance across 17 tumors. Boxes show medians and IQRs while whiskers extend to 1.5\,\(\times\)\,IQR.
\textbf{b,} Recovery of CAR-\(Q\)-significant genes along the Moran's \(I\) ranking. Gray curves show individual tumors; the blue curve and band show the median and IQR, and red highlights tumor 3430\_NT\_T2.
\textbf{c,} Negative contributions as a fraction of total absolute Moran spectral contributions versus expression sparsity in tumor 3430\_NT\_T2. The curve and band show the median and IQR across sparsity deciles; the dashed line marks equal positive and negative contributions.
\textbf{d,} Standardized Moran-\(R\) and CAR-\(R\) scores for 208,335 gene pairs in tumor 3730\_NT\_T2. The dashed line shows a least-squares fit, not the identity line. Colored pairs have residuals exceeding three standard deviations.
\textbf{e,} Spectral coefficients and contributions to Moran's \(R\) for \textit{Gdpd2}--\textit{Spink5}, with selected positive (P1) and negative (N1) modes.
\textbf{f,} Phylogeny of tumor 3730\_NT\_T2, with branches colored by the dominant selected mode. Outer rings show P1, N1 and expression of the two genes.
\textbf{g,} Fractions of discordant gene pairs across residual thresholds, using a separate linear fit for each tumor. Boxes summarize tumors and show medians and IQRs while whiskers extend to 1.5\,\(\times\)\,IQR.
\textbf{h,} Pairwise scores for 646 genes ordered by CAR-based clustering. Lower and upper triangles show Moran-\(R\) and CAR-\(R\) scores, respectively.
\textbf{i,} Overlap of six modules inferred separately by each kernel and matched to maximize agreement. Rows show percentages within Moran-based modules. Overall, 83\% of genes retain matched assignments.%
}

\newcommand{\suppFigSixTitle}{Parameter sensitivity and null calibration of DF analysis in Alzheimer's disease.}
\newcommand{\suppFigSixLegend}{%
\textbf{\suppFigSixTitle}
\textbf{a,} \textit{Tyrobp} spatial power profiles showing raw and L1-normalized spectra with linear or logarithmic radial bins. Faint curves show individual sections and bold curves show genotype means.
\textbf{b,} Empirical distributions of analytic DF and DF-norm \(P\) values under 12 sample-label permutations, comparing raw counts or normalized log1p expression and radial or 2D spectra. Diagonals indicate the uniform null expected for calibrated \(P\) values.
\textbf{c,} Sensitivity analysis showing test \(P\) values for 53 known plaque-induced genes, defined by plaque colocalization and coexpression in the original study.
\textbf{d,} Effect of radial-bin count on, from left to right, per-section spectral participation ratio, the effective rank of the pooled residual covariance, permuted false-positive rates and agreement with the reference gene ranking. The residual covariance is computed from log-power profiles after fitting the group comparison, pooling residual cross-products across genes. Frequency bins are generally correlated. Points in the leftmost plot show individual sections and lines show means. Per-section spectral participation ratio approaches a plateau as bins are added, while the effective rank of the pooled residual covariance continues to increase. False-positive rates remained near the nominal level, whereas gene rankings were more sensitive at low bin counts.%
}

\newcommand{\suppFigSevenTitle}{Spectral directionality and variation between Alzheimer's disease brain sections.}
\newcommand{\suppFigSevenLegend}{%
\textbf{\suppFigSevenTitle}
\textbf{a,b,} Spectral anisotropy (\textbf{a}) and dominant orientation (\textbf{b}) for 53 plaque-induced genes (PIGs) and 2,967 other genes in each section. Values near zero in \textbf{a} indicate weak directionality. These diagnostics support one-dimensional radial binning for the DF analysis in this cohort.
\textbf{c,} Representative 2D power spectra for \textit{Cst3} in wild-type and \(\mathrm{App}^{NL-G-F}\) sections, with white lines indicating the dominant spectral axis.
\textbf{d,} Relative differences in section-level mean expression (DC component) and radial-bin power between sections from the same or different animals, shown separately by genotype. For each gene and bin, absolute pairwise differences were divided by the pair mean and averaged over two within-animal or four between-animal pairs. Boxes summarize 3,020 genes and show medians and IQRs, while whiskers span the 5th--95th percentiles. Given the limited number of animals, each section was analyzed as a separate sample, without accounting for possible correlations between sections from the same animal.%
}

\newcommand{\suppFigEightTitle}{DF identifies additional candidate genes with AD-associated spatial changes.}
\newcommand{\suppFigEightLegend}{%
\textbf{\suppFigEightTitle}
Each row shows one representative gene: \textit{Lrrc17}, \textit{Ift46}, \textit{Cpe}, \textit{Cpne6}, \textit{Nefh} or \textit{Penk}. Left, log1p-transformed expression in one wild-type and one \(\mathrm{App}^{NL-G-F}\) section. Middle, radial power profiles; faint curves show individual sections and bold curves show genotype means. Right, log2 ratios of group-mean L1-normalized profiles (\(\mathrm{App}^{NL-G-F}\)/wild type). Yellow and gray shading mark frequencies enriched in mutant and wild-type sections, respectively. DE, DF and DF-norm \(P\) values are unadjusted.%
}

\newcommand{\suppFigNineTitle}{Distinct spectral responses of immune and synaptic genes in Alzheimer's disease.}
\newcommand{\suppFigNineLegend}{%
\textbf{\suppFigNineTitle}
\textbf{a,} \textit{Tyrobp}, showing increased spectral power and a relative shift toward lower frequencies.
\textbf{b,} Synaptic genes \textit{Syt11} and \textit{Calm2}, prioritized by DF-norm.
\textbf{c,} Immune and lysosomal genes \textit{Gusb}, \textit{Trem2} and \textit{Fcer1g}, prioritized by DF.
Row layouts are as in Supplementary Fig.~\ref{fig:s8}. Displayed \(P\) values are unadjusted.%
}

\newcommand{\suppFigTenTitle}{Validation of continuous-design DF analysis for postnatal brain development.}
\newcommand{\suppFigTenLegend}{%
\textbf{\suppFigTenTitle}
\textbf{a,} Spectral anisotropy of curated developmental markers and other genes across 17 anterior brain sections from P5 to P60. Boxes show medians and IQRs while whiskers span the 5th--95th percentiles. These diagnostics motivate two-dimensional binning for this cohort.
\textbf{b,} Empirical DF and DF-norm \(P\)-value distributions under 50 permutations of age labels, using raw counts or log-normalized expression and radial or two-dimensional spectra. Diagonals show the uniform null expected for calibrated \(P\) values.
\textbf{c,d,} DF (\textbf{c}) and DF-norm (\textbf{d}) significance for curated immature-neuronal, mature-synaptic, astrocyte, oligodendrocyte and vascular/mural markers. Gene lists were obtained from the original study. Age was modeled as \(\log_{10}(\mathrm{day}+1)\). Displayed \(P\) values are unadjusted.%
}

\newcommand{\suppFigElevenTitle}{Gene-specific timing of spatial expression changes during postnatal development.}
\newcommand{\suppFigElevenLegend}{%
\textbf{\suppFigElevenTitle}
Each row shows one representative gene: \textit{Sparc}, \textit{Zeb2} or \textit{Dpysl3}. Left, representative log1p-transformed expression maps at P5, P7, P9, P12, P21 and P60. Right, raw radial power profiles and log2 fold-change ratios of stage-mean L1-normalized profiles relative to the across-section mean profile. Faint curves show individual sections and bold curves show stage means. Displayed \(P\) values from the 2D and radial analyses are unadjusted.%
}

\newcommand{\suppFigTwelveTitle}{Spectral heterogeneity and test calibration in the glioma cohort.}
\newcommand{\suppFigTwelveLegend}{%
\textbf{\suppFigTwelveTitle}
\textbf{a,} Gene-wise anisotropy across six IDH-mutant and 12 IDH-wild-type sections. These diagnostics support one-dimensional radial binning for DF analysis in this cohort.
\textbf{b,} Effective covariance rank (x-axis) versus effective radial-power rank (y-axis), calculated for each section.
\textbf{c,} Relative spectral distances between IDH-wild-type sections from the same or different patients, before and after L1 normalization. Boxes summarize mean distances per gene across 14 within-patient and 52 between-patient pairs. Given the limited number of patients, each section was analyzed as a separate sample in DF and DF-norm tests, without accounting for possible correlations between sections from the same patient. Boxes in \textbf{a,c} show medians and IQRs while whiskers span the 5th--95th percentiles.
\textbf{d,e,} Analytic DF (\textbf{d}) and DF-norm (\textbf{e}) \(P\)-value distributions under 50 permutations of IDH labels. Diagonals show the uniform null expected for calibrated \(P\) values.
\textbf{f,} Overlap of DE, DF and DF-norm discoveries at FDR\,\(\leq0.05\).
\textbf{g,h,} Representative maps and radial profiles for \textit{BCAT1} (\textbf{g}) and \textit{LDHA} (\textbf{h}). Faint curves show individual sections and bold curves show group means. Shaded plots show log2 fold-change ratios of group-mean L1-normalized profiles (IDH-wild-type/IDH-mutant); orange and blue indicate frequencies enriched in wild-type and mutant tumors, respectively. Displayed gene-level \(P\) values are unadjusted.%
}

\newcommand{\suppFigThirteenTitle}{DF-norm prioritizes IDH-associated hypoxia and vascular programs in glioma.}
\newcommand{\suppFigThirteenLegend}{%
\textbf{\suppFigThirteenTitle}
Results for \textit{VEGFA}, \textit{HILPDA}, \textit{SLC2A1} and \textit{DCN}. Each gene block contains representative log-normalized expression maps from two sections per IDH group, raw radial power profiles and log2 fold-change ratios of group-mean L1-normalized profiles (IDH-wild-type/IDH-mutant). Faint curves show individual sections and bold curves show group means (six mutant and 12 wild-type sections). Orange and blue shading indicate frequencies enriched in wild-type and mutant tumors, respectively. Displayed DE, DF and DF-norm \(P\) values are unadjusted.%
}

\newcommand{\suppFigFourteenTitle}{DF-norm prioritizes IDH-associated immune and complement programs in glioma.}
\newcommand{\suppFigFourteenLegend}{%
\textbf{\suppFigFourteenTitle}
Results for \textit{C1QA}, \textit{C1QB}, \textit{LGALS3} and \textit{C1QC}. Layout is as in Supplementary Fig.~\ref{fig:s13}. Displayed \(P\) values are unadjusted.%
}

\newcommand{\suppFigFifteenTitle}{DF-norm prioritizes IDH-associated plasminogen and coagulation programs in glioma.}
\newcommand{\suppFigFifteenLegend}{%
\textbf{\suppFigFifteenTitle}
Results for \textit{ANXA2}, \textit{ENO1}, \textit{PGK1} and \textit{A2M}. Layout is as in Supplementary Fig.~\ref{fig:s13}. Displayed \(P\) values are unadjusted.%
}

\newcommand{\suppFigSixteenTitle}{DF-norm prioritizes IDH-associated extracellular-matrix programs in glioma.}
\newcommand{\suppFigSixteenLegend}{%
\textbf{\suppFigSixteenTitle}
Results for \textit{COL1A1}, \textit{COL6A1}, \textit{COL1A2} and \textit{TNC}. Layout is as in Supplementary Fig.~\ref{fig:s13}. Profiles include sections where the gene was detected in at least 2\% of spots; group sizes are shown in each panel. Displayed \(P\) values are unadjusted.%
}

\newcommand{\suppFigSeventeenTitle}{Functional annotation of gene clusters in the joint embryonic spectral atlas.}
\newcommand{\suppFigSeventeenLegend}{%
\textbf{\suppFigSeventeenTitle}
Selected GO Biological Process terms for five k-means clusters, with enrichment analyzed separately in mouse and human. Colors match the clusters in Figure~\ref{fig:5}a. Each gene was assigned at each stage using its median embedding coordinate across sections; enrichment used the unique genes assigned to each cluster across stages. FDR values were obtained by Benjamini--Hochberg correction across all cluster--term tests within each species. Fold enrichment is relative to the species-specific background of embedded genes.%
}

\newcommand{\suppFigEighteenTitle}{Stage-specific expression and spatial-pattern changes during embryogenesis.}
\newcommand{\suppFigEighteenLegend}{%
\textbf{\suppFigEighteenTitle}
\textbf{a,} Absolute changes in stage-mean expression between adjacent stages, divided by the sum of the two means.
\textbf{b,} Distances between stage-mean coordinates in the five-dimensional spectral embedding after removing section-wide shifts and correcting squared distances for sampling variance. White points and thick bars in \textbf{a,b} show medians and IQRs. Transitions are placed at the later stage. Mouse ages are mapped to an approximate Carnegie-stage axis.
\textbf{c,} Fractions of genes significant by DE, DF or DF-norm at FDR\,\(\leq0.05\).
\textbf{d,} Recurrent and stage-selective GO programs. Genes in the top decile of adjusted embedding movement were selected if they also passed DF-norm. Color shows \(-\log_{10}(\mathrm{FDR})\). Dot size gives the gene-set intersection, while x-axis labels report the query-set size.
\textbf{e,} Representative stage series for \textit{Ankrd1/ANKRD1}, \textit{Tnnt2/TNNT2} and \textit{Cxcl12/CXCL12}. Each gene pair uses a shared log-normalized expression scale capped at the pooled 99th percentile of positive values across all displayed sections.%
}

\title{
\textbf{\Large Consistent and scalable detection and comparison of spatial patterns}
}
\author[1,2,3,7\#]{Jiayu Su}
\author[1,2]{Jun Hou Fung}
\author[2,4]{Haoyu Wang}
\author[2,4]{Dian Yang}
\author[8]{Xiao Wang}
\author[2,3,5\#]{David A. Knowles}
\author[1,2,6\#]{Raul Rabadan}

\affil[1]{Program for Mathematical Genomics, Columbia University, New York, NY, USA}
\affil[2]{Department of Systems Biology, Columbia University, New York, NY, USA}
\affil[3]{New York Genome Center, New York, NY, USA}
\affil[4]{Department of Molecular Pharmacology and Therapeutics, Columbia University, New York, NY, USA}
\affil[5]{Department of Computer Science, Columbia University, New York, NY, USA}
\affil[6]{Department of Biomedical Informatics, Columbia University, New York, NY, USA}
\affil[7]{Eric and Wendy Schmidt Center, Broad Institute of MIT and Harvard, Cambridge, MA, USA}
\affil[8]{Broad Institute of MIT and Harvard, Cambridge, MA, USA}

\date{}

\begin{document}
\begin{refsection}

\maketitle

\begingroup
\renewcommand\thefootnote{\#}
\footnotetext{Correspondence: sujiayu@broadinstitute.org; dak2173@columbia.edu; rr2579@cumc.columbia.edu}
\endgroup

\begin{abstract}
Spatial omics maps molecular features within tissues to study their organization across biological contexts, including disease and development\autocite{bressan2023dawn}. Interpreting these maps requires distinguishing patterns from noise and comparing them across samples. Yet pattern detection remains unresolved\autocite{charitakis2023disparities,luecken2025defining}, and comparisons based on regional summaries lose spatial detail. Here we show that spatial frequency provides a common basis for detecting and comparing these patterns. Our framework, SONIC, unifies spatial detection tests through a common quadratic form, establishes conditions for reliable detection and explains how kernel choice determines sensitivity to different spatial patterns. It also separates changes in spatial organization from changes in overall abundance without aligning tissues. Fast, memory-efficient algorithms make detection and comparison practical across large datasets. Applying SONIC, we identify plaque-associated responses in Alzheimer's disease\autocite{chen2020spatial} and differences in the spatial expression of blood-clotting genes between glioma subtypes\autocite{greenwald2024integrative}. We further construct a gene-centric atlas that traces changes in spatial expression through mouse and human embryogenesis and supports retrieval of related patterns across species\autocite{chen2022spatiotemporal,pan2026spatiotemporal}. These results extend spatial analysis from identifying individual patterns to following their changes across anatomically different tissues. Together, SONIC establishes a shared language for integrating molecular maps across studies and technologies.
\end{abstract}

\section{Main}

Imagine an oncologist examining biopsies to understand why some patients respond to immunotherapy and others do not. Modern spatial assays offer many potential clues, from immune-cell infiltration to protein and RNA abundance and metabolite distributions\autocite{carstens2024spatial}. Because these measurements contain noise and artifacts, the oncologist must first determine which spatial patterns reflect biology, then ask how they relate to treatment response. This practice of detecting and comparing patterns has guided biomedical discovery since Hooke first described cells in 1665. Yet the statistical foundations of spatial analysis remain unsettled, a problem now intensified by multiplexed assays that produce far more molecular maps than can be reliably assessed by eye.

For spatial transcriptomics, pattern detection is usually framed as identifying spatially variable genes (SVGs), a recognized open problem\autocite{luecken2025defining}. Many methods have been proposed under different assumptions and formulations\autocite{svensson2018spatialde,edsgard2018identification,vandenbon2020clustering,sun2020statistical,detomaso2021hotspot,miller2021characterizing,dries2021giotto,bintayyash2021non,zhu2021spark,chang2024graph,su2026mapping}, yet independent evaluations report conflicting results without a general explanation\autocite{charitakis2023disparities,chen2024evaluating,yan2025categorization,chen2025benchmarking}. Meanwhile, Moran's \(I\)\autocite{moran1948interpretation}, a longstanding measure of spatial autocorrelation, remains the standard choice in spatial analysis workflows and downstream applications\autocite{detomaso2021hotspot,schiffman2024defining}. A unified account of these methods is needed to explain their differing sensitivity and establish which spatial patterns they can reliably detect.

Cross-sample comparison presents a different challenge. Tissue geometry and sampling positions rarely match, so measurements cannot be paired location by location. Common approaches summarize cell-type composition or test cell-type-specific expression in defined spatial niches\autocite{mason2024niche}. Although useful, such summaries discard much of the underlying spatial detail. Alternatively, alignment-based methods register sections to a common coordinate system for location-matched comparison\autocite{zeira2022alignment,tan2026integrative,clifton2026stcompare}. These methods work only when sections share anatomical landmarks, which may be absent between heterogeneous samples such as tumors or tissues undergoing extensive developmental remodeling.

Here we present SONIC (Spatial Organization through Nonrandom-pattern Inference and Comparison), a framework that addresses both problems. SONIC decomposes each pattern into spatial harmonics spanning scales from broad gradients to fine local variation. For detection, SONIC explains how tests weight spatial components and establishes conditions for statistical consistency, while supporting shared algorithms for scalable computation. For comparison, it represents patterns on a common physical-frequency axis, enabling differential-frequency analysis across unaligned tissues. We apply SONIC to neurodegeneration, postnatal brain development and glioma, and build a gene-centric embryonic atlas to trace how individual genes acquire tissue-specific distributions in mouse and human.

\subsection{A unified framework for consistent spatial pattern detection}\label{sec:results_detection}

Spatial pattern detection asks whether expression depends on location. We show that virtually all SVG tests, including Moran's \(I\)\autocite{moran1948interpretation}, other graph-based tests\autocite{detomaso2021hotspot,miller2021characterizing,dries2021giotto,chang2024graph}, regression models\autocite{svensson2018spatialde,sun2020statistical,bintayyash2021non}, and dependence-based tests\autocite{zhu2021spark,edsgard2018identification,vandenbon2020clustering,su2026mapping}, can be represented by quadratic forms \(Q=\mathbf{z}^{\mathsf T}\mathbf{K}\mathbf{z}\), either exactly or in the large-sample limit (Figure~\ref{fig:1}a; Supplementary Note~\ref{supnote:2}). We call any test with this form a \textit{\(Q\)-test}. Here, \(\mathbf{z}\) denotes standardized expression; in parametric models, the standardization reflects the assumed mean-variance relationship. The kernel \(\mathbf{K}\) encodes spatial structure by weighting pairs of locations (Supplementary Note~\ref{supnote:1}), and its construction varies across method families: regression models use spatial basis functions or covariance models, graph methods use adjacency or Laplacian matrices, and dependence tests use spatial similarity kernels. Cell-type-specific SVG methods additionally weight the kernel by estimated cell-type proportions\autocite{cable2022cell,su2025stance,shang2025celina}.

\begin{figure}[t!]
    \centering
    \includegraphics[width=\textwidth]{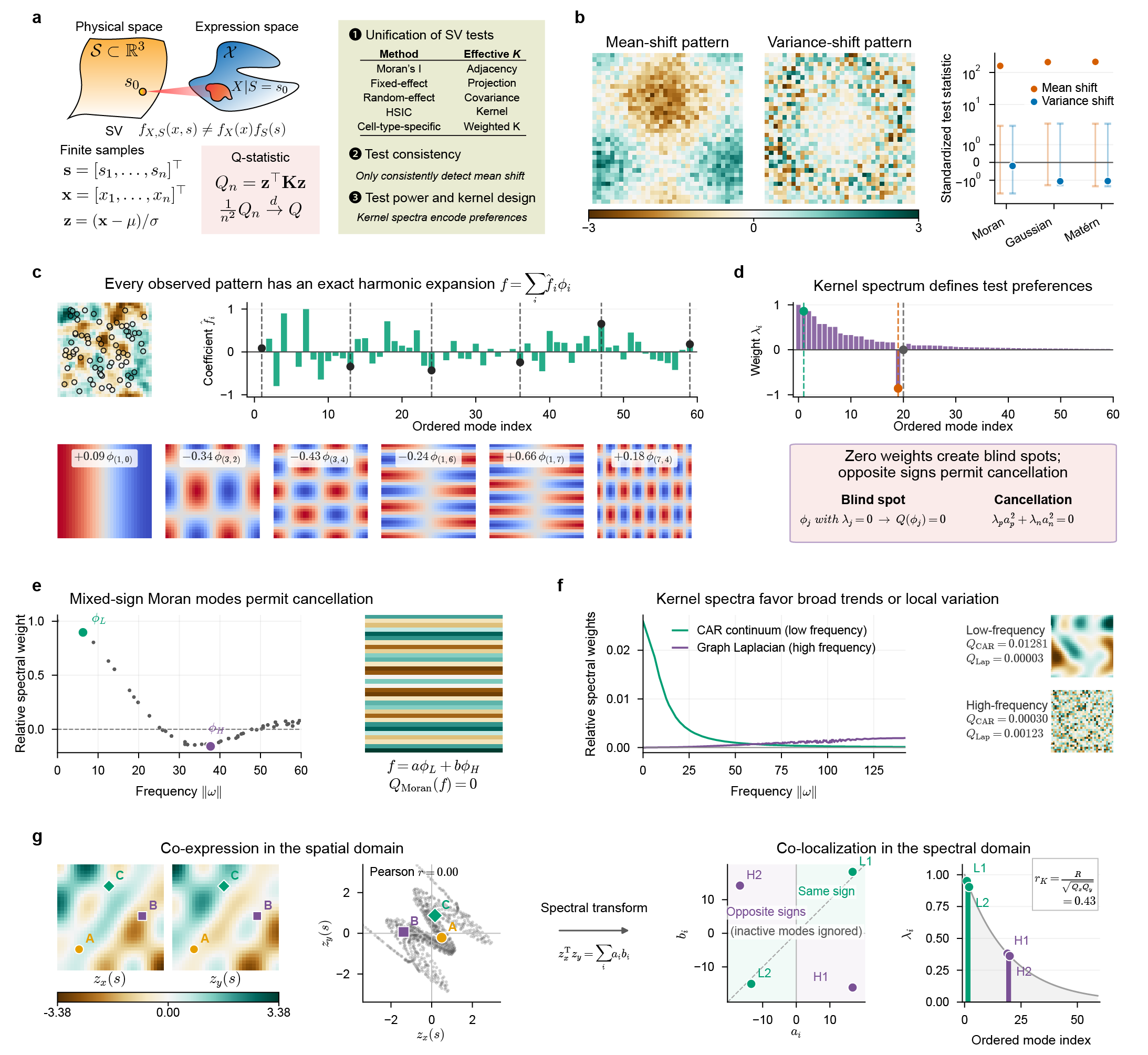}
    \caption{\figOneLegend}
    \label{fig:1}
\end{figure}

The quadratic form captures spatial variation in average expression across locations (Theorem~\ref{theorem:q_mean_dependence}). Consequently, \(Q\)-tests cannot reliably detect changes that preserve the mean, including spatial differences in variance, even as sample size grows (Corollary~\ref{corollary:q_test_mean_independent}; Figure~\ref{fig:1}b). Because destructive spatial assays typically measure each location only once, the observations cannot empirically distinguish a spatially varying mean from changes confined to higher moments without additional assumptions. We therefore model the observed pattern as a deterministic mean function, which can be expanded in harmonics describing broad trends and local variation (Figure~\ref{fig:1}c; Supplementary Note~\ref{supnote:3}).

In this view, a \(Q\)-test consistently detects every nonconstant mean pattern if and only if its kernel assigns nonzero spectral weights of the same sign to every nonconstant mode (Theorem~\ref{theorem:q_functional_consistency}; Figure~\ref{fig:1}d). Violating either condition compromises detection: zero weights leave blind spots, whereas contributions from modes with opposite-signed weights can cancel. Moran's \(I\), for instance, has a mixed-sign spectrum and may therefore assign a score of zero to a strong pattern (Figure~\ref{fig:1}e). Beyond consistency, kernel spectra also explain differences in pattern preference: the conditional autoregressive (CAR)\autocite{su2023smoother} kernel targets low-frequency, global trends, whereas the graph-Laplacian\autocite{govek2019clustering} kernel emphasizes high-frequency, local variation (Figure~\ref{fig:1}f; Supplementary Note~\ref{supnote:3}). The same spectral problem affects the bivariate \(R\)-statistic $\mathbf{z}^\mathsf{T}\mathbf{K}\mathbf{y}$, which measures spatial association between two features (Figure~\ref{fig:1}g; \nameref{sec:methods}). For example, spatial cross-correlation\autocite{wartenberg1985multivariate}, widely used for identifying coexpression modules\autocite{detomaso2021hotspot,li2023spatialdm} and tumor heritability\autocite{schiffman2024defining}, relies on an adjacency kernel that can likewise cancel co-localized signals.

\subsection{Scalable pattern detection for atlas-scale discovery}\label{sec:results_scaling}
SONIC accelerates \(Q\)-tests through three computational strategies that exploit their quadratic form structure (Figure~\ref{fig:2}a; Supplementary Note~\ref{supnote:eng}): moment matching with Hutchinson trace estimates\autocite{liu2009new,hutchinson1989stochastic} approximates the null distribution without eigendecomposition; fast Fourier transforms (FFT) and non-uniform FFT (NUFFT)\autocite{barnett2019parallel} evaluate stationary-kernel quadratic forms in near-linear time; sparse linear solves evaluate CAR statistics on general graphs. These implementations preserve the accuracy of the original tests (Supplementary Fig.~\ref{fig:s1}), reduce runtime by orders of magnitude and enable million-spot analyses on a personal laptop (Figure~\ref{fig:2}b--c).

\begin{figure}[t!]
    \centering
    \includegraphics[width=\textwidth]{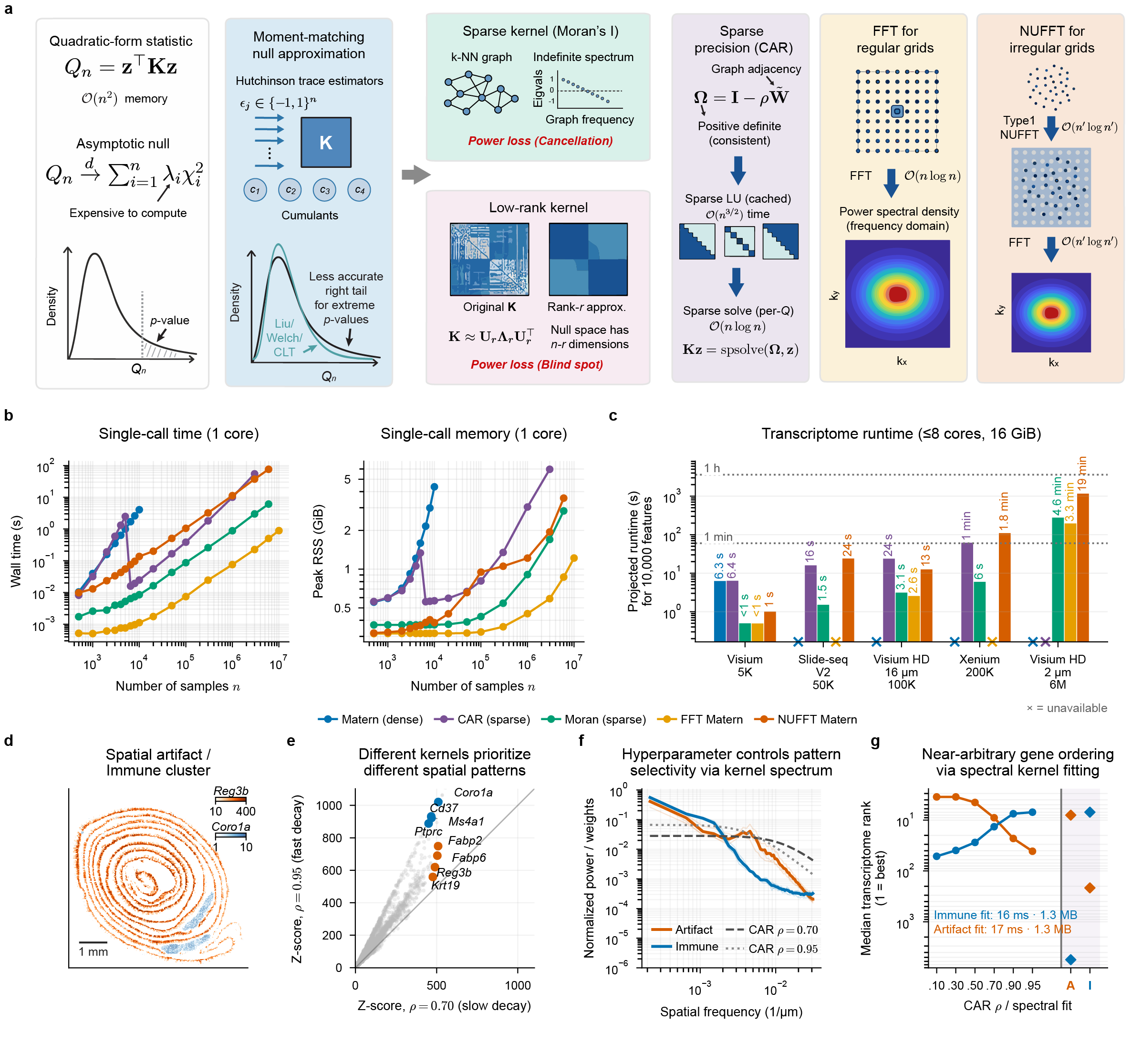}
    \vspace{-7mm}
    \caption{\figTwoLegend}
    \label{fig:2}
\end{figure}

We applied these accelerations to SPLISOSM\autocite{su2026mapping}, a multivariate dependence test for isoform usage, analyzing millions of 2-$\mu$m Visium HD bins from adult mouse brain\autocite{10x_visium_hd_mouse_brain_3prime,10x_visium_hd_mouse_brain_ffpe,ont_visium_hd_mouse_brain} (Supplementary Fig.~\ref{fig:s2}; \nameref{sec:methods}). Frequency-domain calculations reduce the cost of Gaussian-process fitting, enabling transcriptome-wide spatial conditional differential usage (DU) tests. Combining long-read, 3$'$ short-read and probe-based assays, we identified cortex-specific usage of an upstream \textit{Mbp} transcript encoding a myelin basic protein isoform with an N-terminal Golli extension. Conditional DU tests and binding data nominated MYT1L and EGR1 as candidate regulators of this regional switch (Supplementary Fig.~\ref{fig:s3}). Because Golli isoforms regulate calcium signaling in oligodendrocyte precursor cells\autocite{paez2011modulation}, this switch may reflect region-specific oligodendrocyte states in the adult brain.

Scalable analysis at high resolution revealed preferential use of a downstream \textit{Gabbr1} transcription start site (TSS) in the small medial habenula (MHb), a brain region involved in mood and fear-related behavior. The favored transcript encodes a receptor subunit lacking N-terminal sushi domains implicated in protein interactions and axonal targeting\autocite{biermann2010sushi}, and transcripts encoding these domains were depleted from the MHb (Supplementary Fig.~\ref{fig:s4}). This isoform shift may reduce dependence on sushi-mediated extracellular interactions and alter GABA-B receptor signaling in ventral MHb terminals, where receptor activation enhances transmitter release\autocite{zhang2016presynaptic}. DU tests and binding data implicate EGR1 and RBFOX3 in regulating \textit{Gabbr1} TSS and alternative last-exon usage, respectively (Supplementary Fig.~\ref{fig:s4}).

Our implementation also accommodates irregularly spaced cells from imaging-based platforms such as 10x Atera (\nameref{sec:methods}). In a human breast-cancer sample\autocite{10x_atera_breast_cancer}, \textit{ANKRD30A}, a breast-cancer antigen recognized by cytotoxic T cells\autocite{jaeger2001identification,wang2006recognition}, showed spatially distinct usage of transcript classes in luminal and invasive tumor cells (Extended Data Fig.~\ref{fig:ed1}a--c). Probe signals and independent long-read transcript data supported two classes with distinct transcription start sites and coding sequences, suggesting candidate isoform-specific epitopes for targeted T-cell therapies (Extended Data Fig.~\ref{fig:ed1}d). Conditional DU tests and binding data implicated TRPS1 in this transcript choice, consistent with its role in regulating estrogen-receptor binding at enhancers\autocite{serandour2018trps1} (Extended Data Fig.~\ref{fig:ed2}a--c). In estrogen receptor-positive tumors from The Cancer Genome Atlas, higher \textit{TRPS1} expression was associated with greater use of luminal-class transcripts (Extended Data Fig.~\ref{fig:ed2}d--f). Finally, acute TRPS1 degradation increased transcription at invasive-class start sites, supporting their repression by TRPS1 (Extended Data Fig.~\ref{fig:ed2}g).

\subsection{Kernel spectra explain pattern preference and missed lineage signals}\label{sec:results_lineage}
The Q-test framework predicts how kernel choice changes sensitivity to different spatial patterns. No single kernel is universally optimal, because emphasizing some spatial components reduces the relative weight of others. Detection is therefore Bayesian in spirit, with the kernel and its hyperparameters encoding prior beliefs about which signals matter. We illustrate this in a mouse small-intestine sample\autocite{10x_visium_hd_intestine}, where \textit{Reg3b} followed a broad ``Swiss roll'' pattern introduced by sample preparation, whereas \textit{Coro1a} was concentrated in the follicle-associated epithelium (Figure~\ref{fig:2}d). Their expression power profiles occupied different frequency ranges, so changing the CAR autocorrelation parameter shifted their relative rankings as predicted (Figure~\ref{fig:2}e--f). When target and nuisance signals are known in advance, kernels can be fitted to maximize their ranking contrast at little computational cost (Figure~\ref{fig:2}g). By contrast, fitting a separate kernel for each gene risks compromising statistical validity (Supplementary Note~\ref{supnote:3.data_driven_kernel}).

\begin{figure}[t!]
    \centering
    \includegraphics[width=\textwidth]{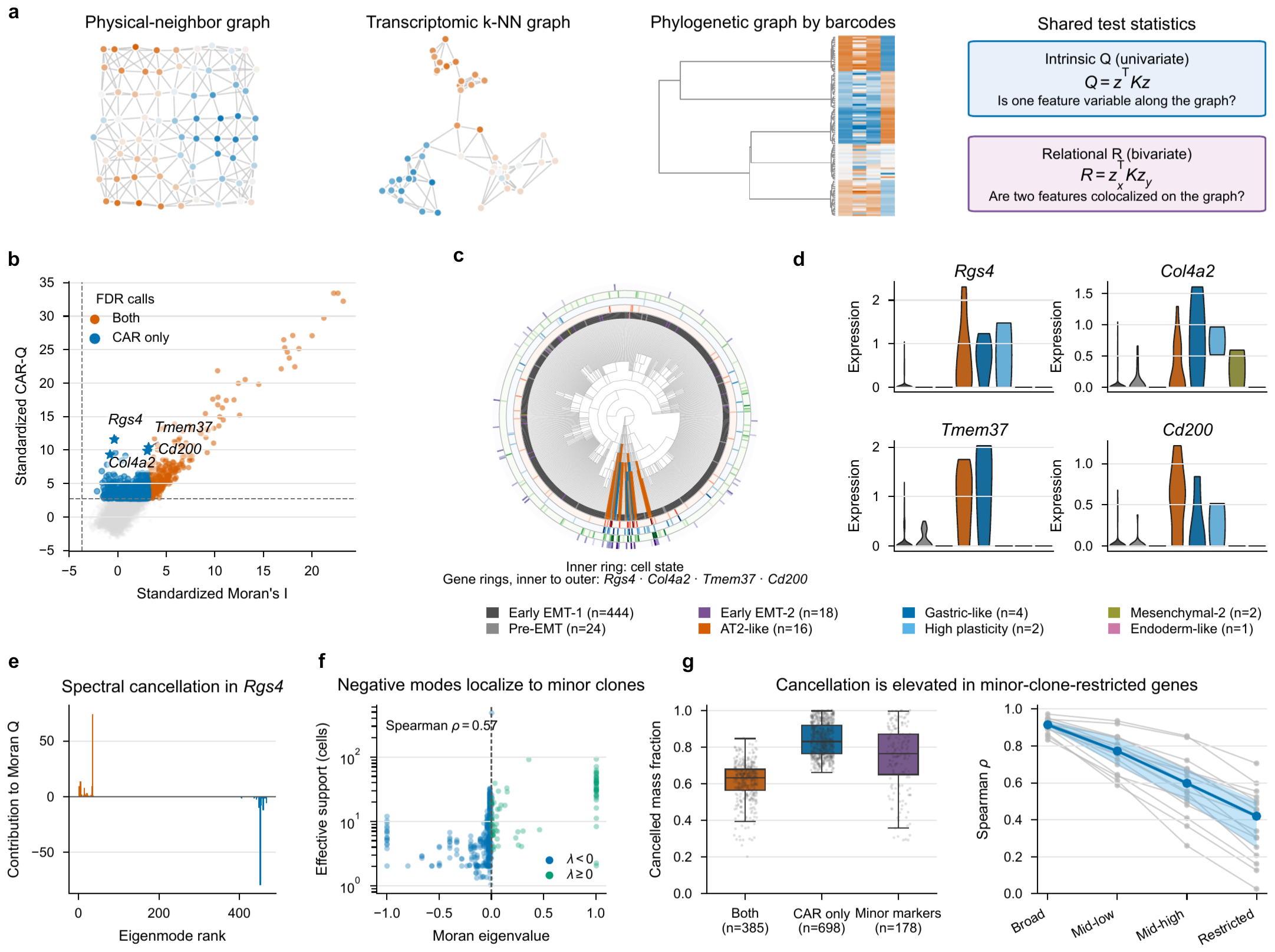}
    \vspace{-5mm}
    \caption{\figThreeLegend}
    \label{fig:3}
\end{figure}

A second prediction concerns kernels with positive and negative spectral weights, whose contributions can cancel. We examined this behavior in Moran's \(I\) and its bivariate counterpart\autocite{detomaso2021hotspot,schiffman2024defining} in a barcode-traced lung-cancer dataset\autocite{yang2022lineage}, using \(Q\)-tests to detect heritable expression programs shared across lineages (Figure~\ref{fig:3}a; \nameref{sec:methods}). Although Moran's \(I\) and CAR agreed on many of the strongest hits, Moran's \(I\) missed numerous genes with heritable expression patterns (Figure~\ref{fig:3}b; Supplementary Fig.~\ref{fig:s5}a--b). These included \textit{Rgs4}, \textit{Col4a2}, \textit{Tmem37} and \textit{Cd200}, whose expression was enriched in the minor alveolar type 2 (AT2)-like clone (Figure~\ref{fig:3}c--d). For \textit{Rgs4}, opposing mode contributions nearly cancelled one another (Figure~\ref{fig:3}e). Generally, modes with negative Moran weights concentrated in small cell groups, capturing variation within minor clones and making sparse clone markers especially vulnerable to cancellation (Figure~\ref{fig:3}f--g; Supplementary Fig.~\ref{fig:s5}c). Bivariate coexpression analysis was similarly affected (Supplementary Fig.~\ref{fig:s5}d--i). For \textit{Spink5} and \textit{Gdpd2}, opposing spectral contributions attenuated the association measured by Moran's \(R\) despite their co-occurrence across much of the tumor phylogeny (Supplementary Fig.~\ref{fig:s5}e--f). 17\% of genes were assigned to different inferred modules by the two tests (Supplementary Fig.~\ref{fig:s5}h--i).

\subsection{Pattern comparison across unmatched tissue sections}
To compare patterns across samples, SONIC extends the spectral view of single-sample \(Q\)-tests by describing how each pattern's power is distributed along a shared spatial-frequency axis measured in physical units. This allows differential-frequency (DF) analysis across groups or along a covariate, even when tissue sections cannot be anatomically aligned (Figure~\ref{fig:4}a; \nameref{sec:methods}). DF tests a multivariate frequency profile in place of the single bulk value used in differential-expression (DE) analysis. Each slide contributes one observation, so correlated cells or spots are not treated as independent replicates, avoiding a common source of false discoveries\autocite{squair2021confronting}. A normalized variant, DF-norm, further rescales each spectrum to unit total power, isolating relative frequency composition from overall pattern strength.

\begin{figure}[t!]
    \centering
    \includegraphics[width=\textwidth]{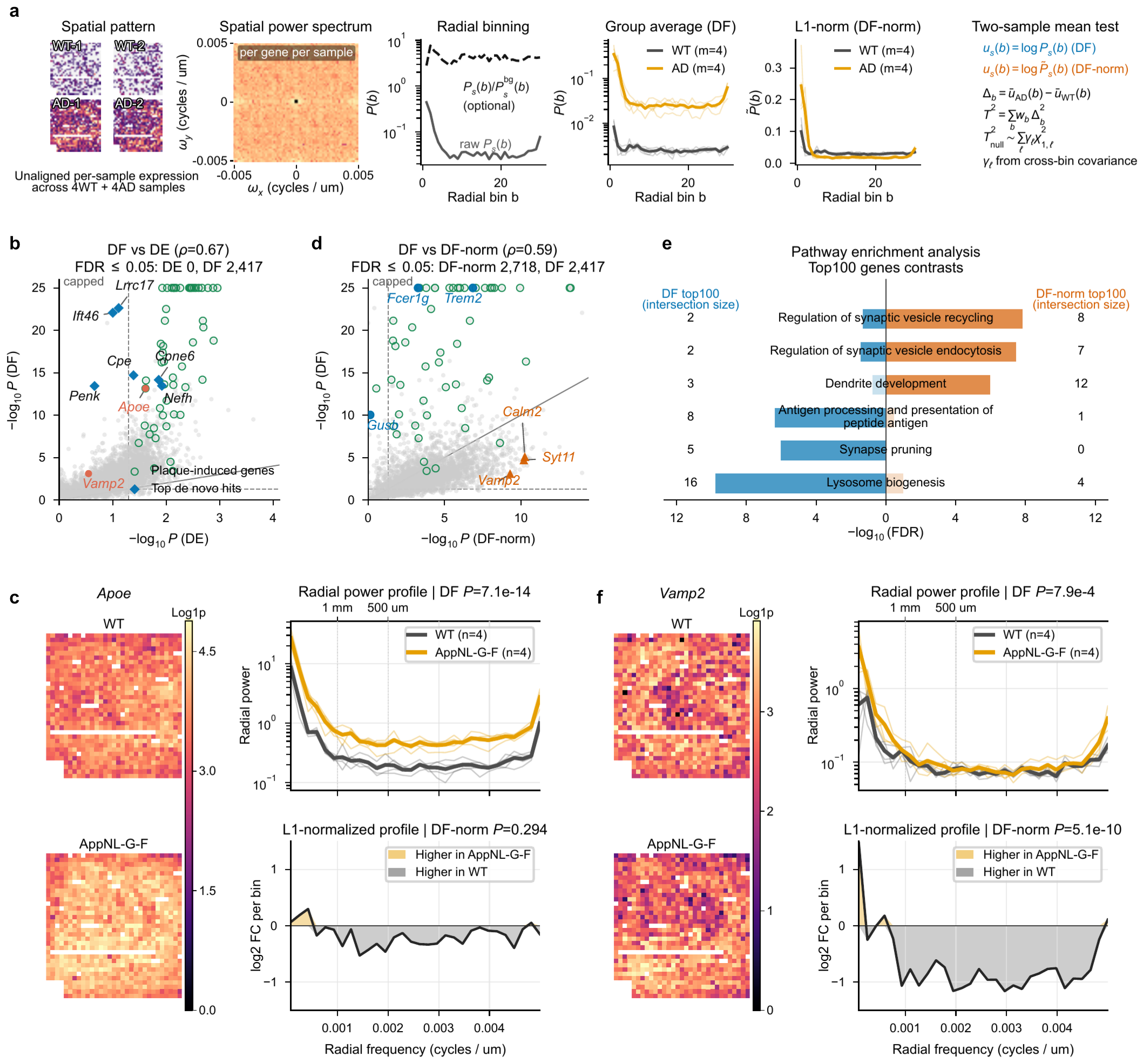}
    \caption{\figFourLegend}
    \label{fig:4}
\end{figure}

We first applied this approach to a mouse Alzheimer's disease (AD) dataset containing coronal brain sections from wild-type and amyloid-bearing $\mathrm{App}^{NL-G-F}$ mice \autocite{chen2020spatial} (\nameref{sec:methods}). Section-label permutations supported null calibration, and the tests recovered known plaque-induced genes (Supplementary Figs.~\ref{fig:s6}--\ref{fig:s7}). DE of section-mean expression yielded no significant hits, whereas DF detected 2,417 genes with AD-associated spectral differences at a false discovery rate (FDR)\,\(\leq0.05\), including additional candidates such as \textit{Lrrc17}, \textit{Nefh} and \textit{Cpe} (Figure~\ref{fig:4}b; Supplementary Fig.~\ref{fig:s8}). Established AD-related genes responded differently to amyloid pathology. \textit{Apoe}, expressed by astrocytes and microglia and involved in immune signaling, increased across brain regions with little change in frequency composition (Figure~\ref{fig:4}c). In contrast, \textit{Tyrobp}, a microglial signaling gene induced around amyloid plaques\autocite{keren2017unique}, increased with a shift toward lower frequencies and stronger regional clustering (Supplementary Fig.~\ref{fig:s9}a). This contrast motivated a transcriptome-wide comparison of DF and DF-norm. The two tests returned similar numbers of hits but prioritized different genes (Figure~\ref{fig:4}d). DF signals were enriched for lysosomal, antigen-presentation and synapse-pruning programs. DF-norm instead favored synaptic vesicle genes including \textit{Syt11}, \textit{Vamp2} and \textit{Calm2}, whose expression was more locally clustered in AD (Figure~\ref{fig:4}e--f; Supplementary Fig.~\ref{fig:s9}). Together, these results suggest that amyloid pathology is associated with both a broad immune response and more localized neuronal and synaptic disruption.

We next extended DF analysis to a continuous design, testing age-associated changes in postnatal mouse brain at six time points from postnatal day 5 (P5) to P60\autocite{de2026spatiotemporal} (\nameref{sec:methods}). DF and DF-norm identified more age-associated genes than DE, with oligodendrocyte and neuronal maturation among the strongest programs (Extended Data Fig.~\ref{fig:ed3}a--b; Supplementary Fig.~\ref{fig:s10}). DF emphasized gliogenesis and axon ensheathment, whereas DF-norm highlighted synapse organization and activity (Extended Data Fig.~\ref{fig:ed3}c). Individual genes followed distinct trajectories. Expression of the endothelial marker \textit{Aplnr} declined broadly, consistent with a transition from early vascular growth to maturation (Extended Data Fig.~\ref{fig:ed3}d), whereas expression of the synaptic vesicle gene \textit{Synpr}, enriched in the caudate-putamen at P5, became more diffuse as striatal circuits matured\autocite{krajeski2019dynamic} (Extended Data Fig.~\ref{fig:ed3}e). The timing of these changes also differed: the transition was abrupt for \textit{Sparc}, a regulator of synapse formation that became enriched in the ventral forebrain by P21, while \textit{Zeb2} and \textit{Dpysl3} instead underwent gradual frequency shifts alongside declining expression (Supplementary Fig.~\ref{fig:s11}).

Finally, we asked how expression patterns differed between tumors with sharply varying architecture. Gliomas contain many disordered regions, and existing analyses often summarize them through recurrent cell states and their composition, leaving finer gene-level organization unexplored\autocite{greenwald2024integrative,hoefflin2026spatial}. Mutations in isocitrate dehydrogenase (IDH) genes define a major glioma subtype with metabolic reprogramming, reduced immune infiltration and more favorable prognosis\autocite{miller2023isocitrate}. Using DF and DF-norm, we assessed microenvironmental differences between IDH-mutant and IDH-wild-type sections (\nameref{sec:methods}). Section-label permutations supported calibration of both spectral tests, but DF returned fewer significant hits than DE or DF-norm, possibly because strong inter-tumor variation reduced its power (Supplementary Fig.~\ref{fig:s12}a--f). Consistent with the reported effects of IDH mutation, \textit{BCAT1}, an IDH-linked metabolic gene\autocite{tonjes2013bcat1}, and \textit{LDHA}, a glycolytic gene subject to IDH-associated promoter hypermethylation\autocite{chesnelong2014lactate}, showed lower overall expression in IDH-mutant gliomas (Supplementary Fig.~\ref{fig:s12}g--h). We then ranked DF-norm hits and annotated the top 100 by functional pathway. These programs spanned metabolism, hypoxia, immune and complement response, plasminogen and coagulation, and extracellular-matrix remodeling (Extended Data Fig.~\ref{fig:ed4}a--b). Across these programs, shifts toward lower frequencies accompanied regional expression clusters in IDH-wild-type tumors, aligned with the layered glioblastoma architecture\autocite{greenwald2024integrative} (Extended Data Fig.~\ref{fig:ed4}c--f; Supplementary Figs.~\ref{fig:s13}--\ref{fig:s16}). In this cohort, coagulation-related genes, including \textit{S100A10}, \textit{ANXA2}, \textit{ENO1}, \textit{PGK1} and \textit{A2M}, showed more regionally concentrated expression in IDH-wild-type tumors despite modest differences in bulk abundance (Extended Data Fig.~\ref{fig:ed4}e; Supplementary Fig.~\ref{fig:s15}). This regional organization may contribute to elevated thromboembolic risk in IDH-wild-type glioma\autocite{unruh2016mutant}, motivating larger cohort studies to link spatial patterns, particularly frequency shifts, with patient outcomes.

\begin{figure}[t!]
    \centering
    \includegraphics[width=\textwidth]{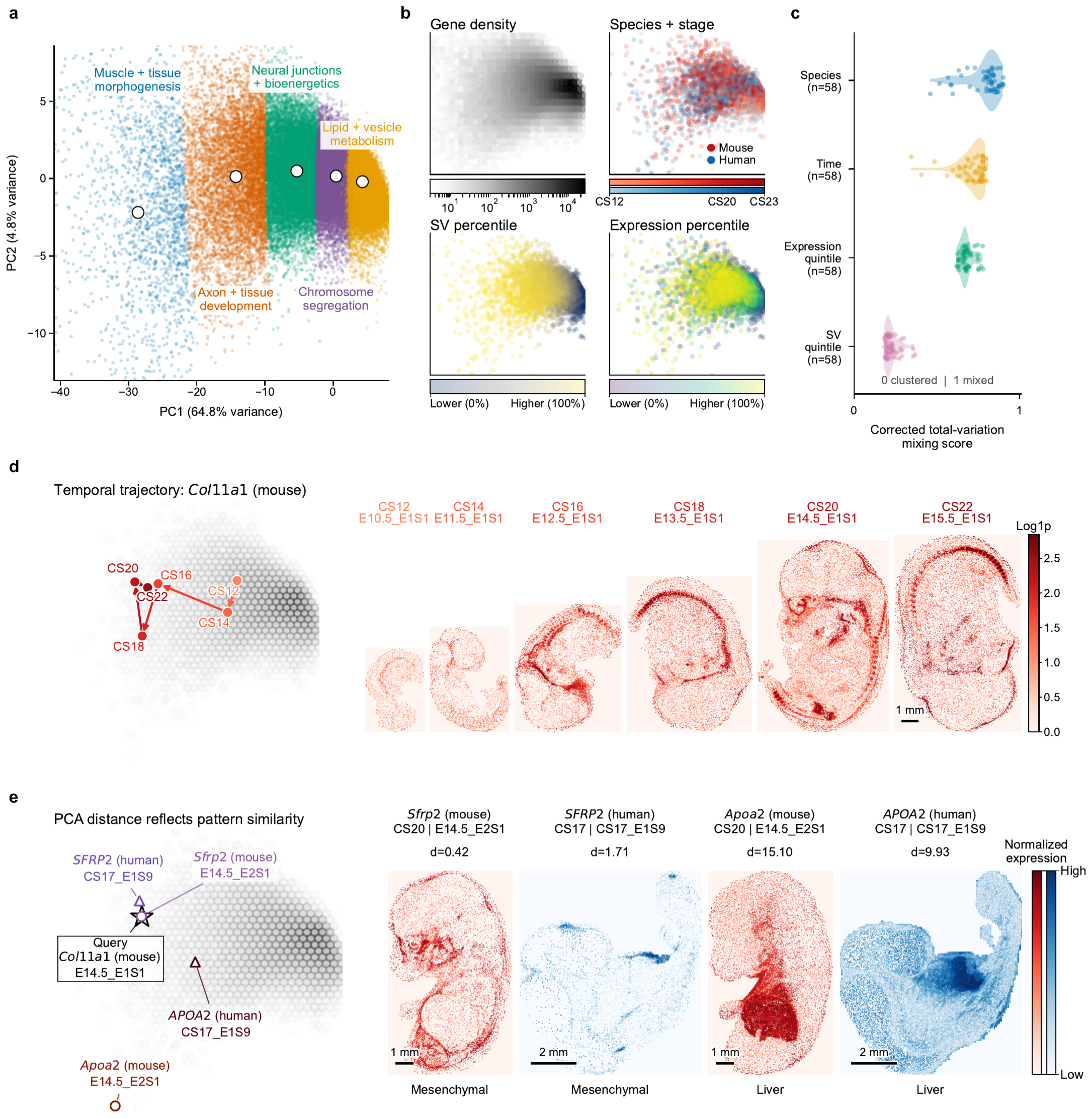}
    \caption{\figFiveLegend}
    \label{fig:5}
\end{figure}

\subsection{A spatiotemporal atlas of gene expression dynamics during embryogenesis}
\label{sec:results_atlas}
Beyond statistical testing, SONIC's harmonic representation creates a shared frequency space for tracing genes across development. We assembled a gene-centric atlas of embryogenesis from 30 mouse and 28 human Stereo-seq sections, extending existing atlases organized around cell types and organs\autocite{chen2022spatiotemporal,pan2026spatiotemporal} (\nameref{sec:methods}). Each gene-section pair contributed a normalized radial power spectrum, yielding 529,092 profiles that we embedded jointly using principal component analysis (Figure~\ref{fig:5}a).

Functional programs spanned a continuum of spatial variability in the embedding, with stronger structure in muscle and neural development than in lipid and vesicle metabolism (Figure~\ref{fig:5}a; Supplementary Fig.~\ref{fig:s17}). Species, developmental stages and expression levels were well mixed in the embedding, allowing us to track individual genes across samples (Figure~\ref{fig:5}b--c). For example, the cartilage-matrix gene \textit{Col11a1} underwent its largest spatial shift at embryonic day 12.5 (E12.5), becoming concentrated in the developing vertebral column, skull base and jaw (Figure~\ref{fig:5}d). To examine whether temporal dynamics differed between species, we focused on 6,146 one-to-one orthologs and used replicate sections to separate between-stage from within-stage variation (Extended Data Fig.~\ref{fig:ed5}a--b; \nameref{sec:methods}). Among them, 257 genes ranked in the top decile for between-stage variation in both species, a 4.2-fold enrichment over chance (Extended Data Fig.~\ref{fig:ed5}c). These genes were enriched for axon guidance and extracellular-matrix organization. Genes ranking in the top decile only in human were enriched for blood-vessel morphogenesis and epithelial proliferation (Extended Data Fig.~\ref{fig:ed5}d). At the gene level, \textit{Col11a1/COL11A1} showed conserved skeletal localization, while \textit{Copg2/COPG2} changed more strongly in the mouse brain and spinal cord. \textit{Eng/ENG} broadened from an early cardiac focus to include lung and kidney in human, but followed a less dynamic trajectory in mouse (Extended Data Fig.~\ref{fig:ed5}e--f).

We next used embedding distances to retrieve related expression maps. An E14.5 mouse \textit{Col11a1} query returned mouse \textit{Sfrp2} and human \textit{SFRP2} patterns, both distributed through craniofacial and skeletal mesenchyme\autocite{leimeister1998developmental} (Figure~\ref{fig:5}e). By contrast, liver-enriched \textit{Apoa2/APOA2} profiles lay farther away along the second principal component, which separated localized from broadly distributed expression.
Distances between patterns from adjacent stages also allowed us to quantify developmental transitions (\nameref{sec:methods}). In the human atlas, Carnegie stages 18--20 showed the largest shifts in expression and spatial distribution (Supplementary Fig.~\ref{fig:s18}a--c). Broad muscle, cartilage and neural development programs recurred across transitions, whereas cardiovascular and heart morphogenesis programs changed more selectively with stage (Supplementary Fig.~\ref{fig:s18}d). The cardiac genes \textit{ANKRD1} and \textit{TNNT2} localized these changes to developing myocardium, including ventricular walls, during ventricular septation and outflow-tract remodeling\autocite{hikspoors2022pictorial} (Supplementary Fig.~\ref{fig:s18}e).

\FloatBarrier
\subsection{Discussion}
SONIC connects spatial pattern detection and comparison through a spectral representation of molecular organization. The Q-test framework explains how kernel choice determines which patterns a test favors or misses, while shared computational strategies make these tests practical at high resolution. Comparing frequency profiles then extends the analysis across tissues without requiring anatomical correspondence. In applications, SONIC distinguishes changes in pattern strength from changes in relative frequency composition and uses these profiles to follow and retrieve expression patterns across development and species.

The \(Q\)-test theory explains persistent disagreements over spatial pattern detection performance. SVG methods have proliferated with competing claims of superiority, while independent benchmarks have documented rather than resolved their disagreement\autocite{chen2024evaluating,chen2025benchmarking}. Because no kernel is optimal for every pattern, the search for a universal winner was bound to fail. In its place, the theory offers three practical recommendations. First, treat scalability as an orthogonal design goal. Our \(Q\)-test implementations reduce runtime and memory without sacrificing power. Second, choose kernels according to prior biological expectations, using positive and negative controls where available. Without such controls, polynomially decaying Matérn and CAR kernels are useful defaults because they favor broad trends while retaining sensitivity to local variation\autocite{lindgren2011explicit}. Gene-specific tuning must account for selection. Third, examine the kernel spectrum for blind spots and cancellation. Moran's \(I\) and its bivariate extension\autocite{detomaso2021hotspot,li2023spatialdm,schiffman2024defining} illustrate how opposite-signed weights can conceal structured signals, including expression associated with minor clones.

Differential-frequency analysis extends spatial comparisons to tissues that cannot be anatomically aligned, broadening the biological questions we can address. Comparing tumor microenvironments across patients is a prime example. This flexibility does not, however, rescue a weak experimental design. In the heterogeneous glioma cohort, unrelated inter-sample variation may have reduced DF power. To mitigate such variation, we implement two optional adjustments inspired by batch correction in expression analysis: per-slide normalization removes shared spectral backgrounds estimated from the geometric mean across genes, and covariate regression adjusts for components associated with spatial covariates, such as cell-type abundance or tissue-region annotations (\hyperref[sec:methods_df]{Methods}). Nevertheless, neither adjustment substitutes for biological replication. Reliable inference requires replication across individuals, which additional sections from one individual cannot replace.

To make patterns comparable across samples, SONIC deliberately discards some information (Supplementary Note~\ref{supnote:5}). In return, it provides a shared language for studying spatial organization. Our embryonic atlas shows how this representation enables retrieval of related patterns and tracing of individual genes across developmental stages and species. The same physical frequency axis can connect gene expression with chromatin accessibility, protein and metabolite maps across assays and resolutions. Applied to spatial perturbation studies, it could reveal how interventions reshape tissue organization and help identify the regulators involved. Collectively, SONIC establishes a common framework for testing how molecular organization changes across biological conditions, developmental stages and species.

\clearpage
\begingroup
\captionsetup[figure]{name=Extended Data Fig.}
\renewcommand{\thefigure}{\arabic{figure}}
\renewcommand{\theHfigure}{ED\arabic{figure}}
\setcounter{figure}{0}

\begin{figure}[p]
    \centering
    \includegraphics[width=\textwidth]{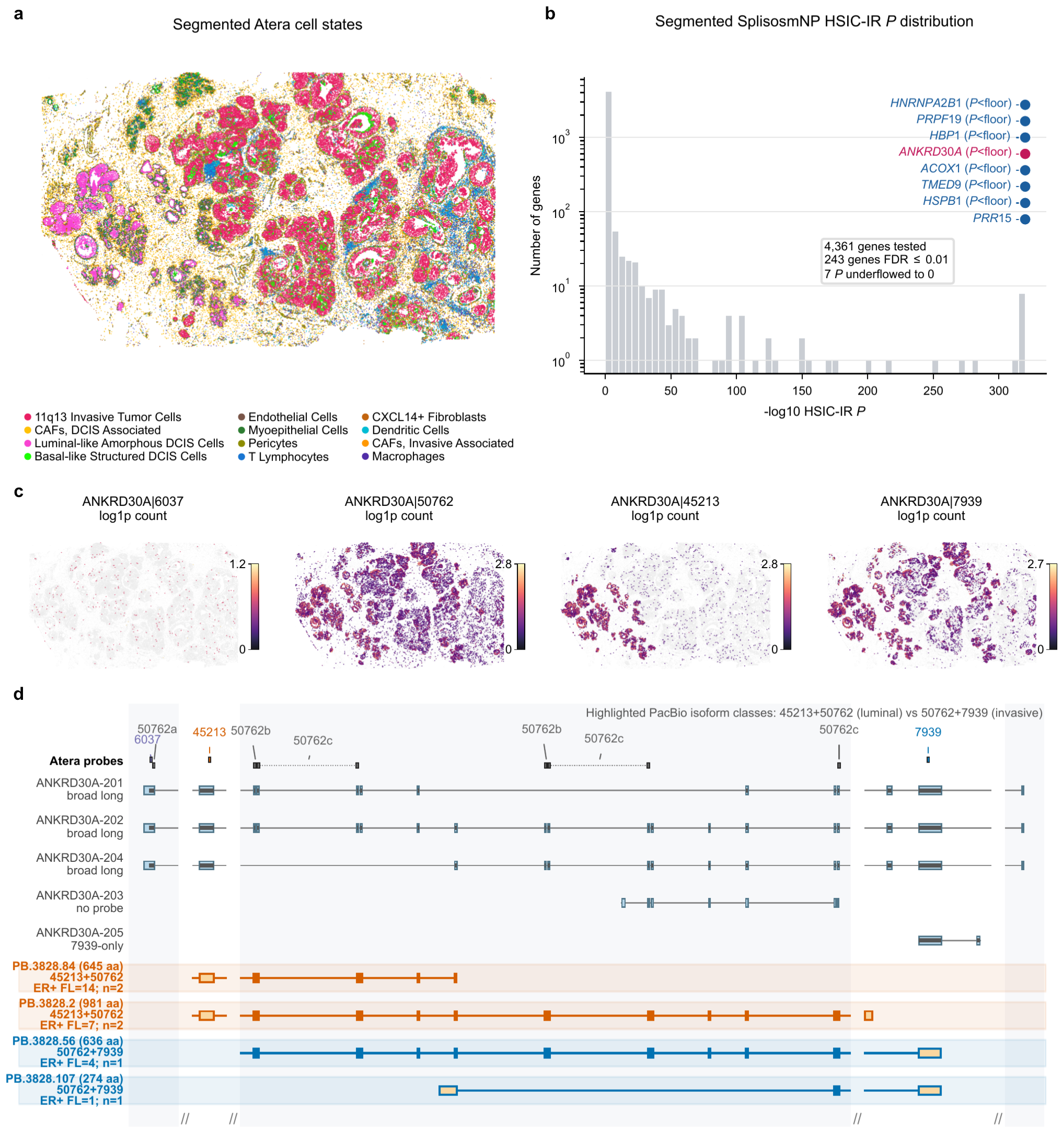}
    \caption{\edFigOneLegend}
    \label{fig:ed1}
\end{figure}

\begin{figure}[p]
    \centering
    \includegraphics[width=\textwidth]{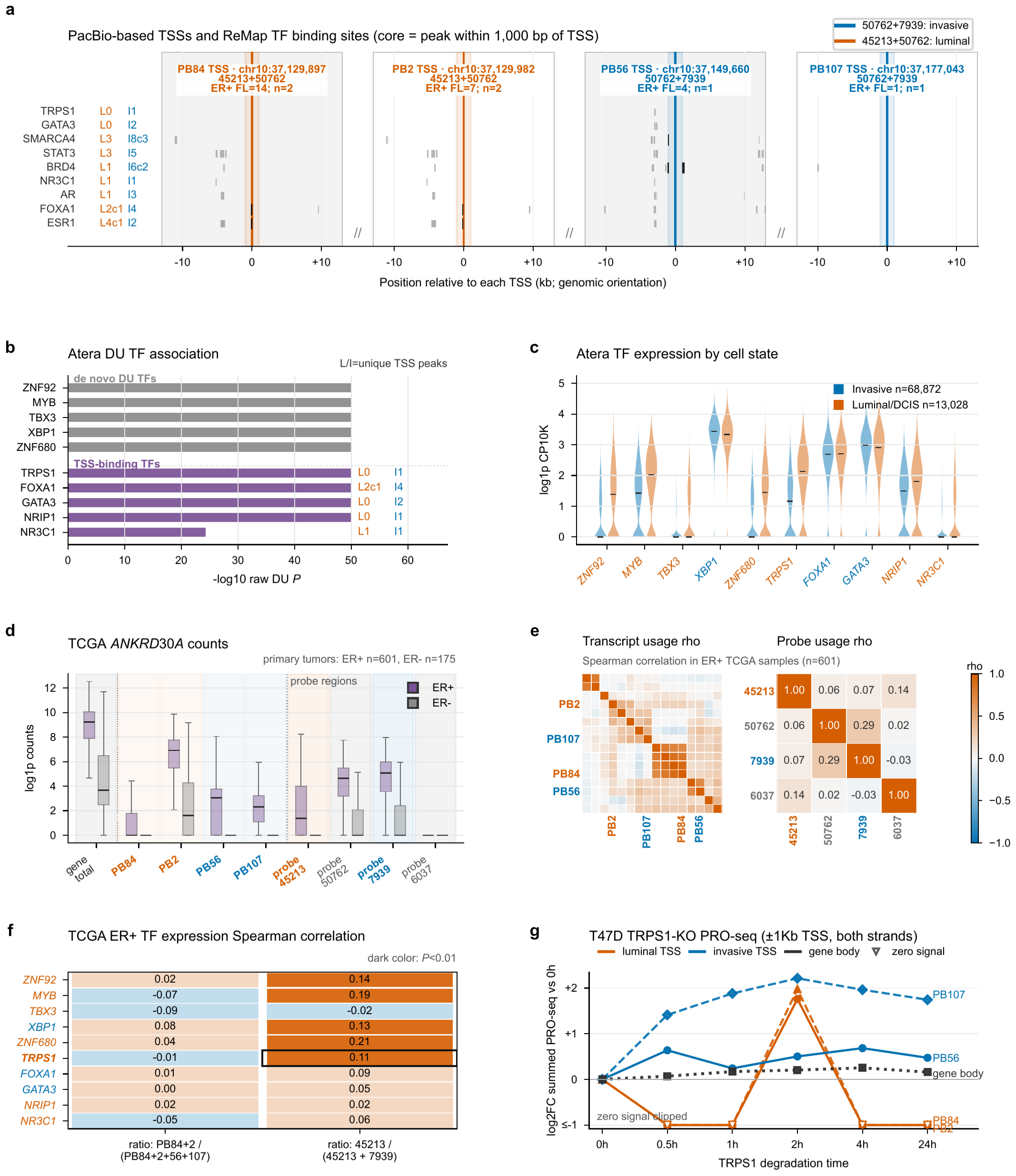}
    \caption{\edFigTwoLegend}
    \label{fig:ed2}
\end{figure}

\begin{figure}[p]
    \centering
    \includegraphics[width=\textwidth]{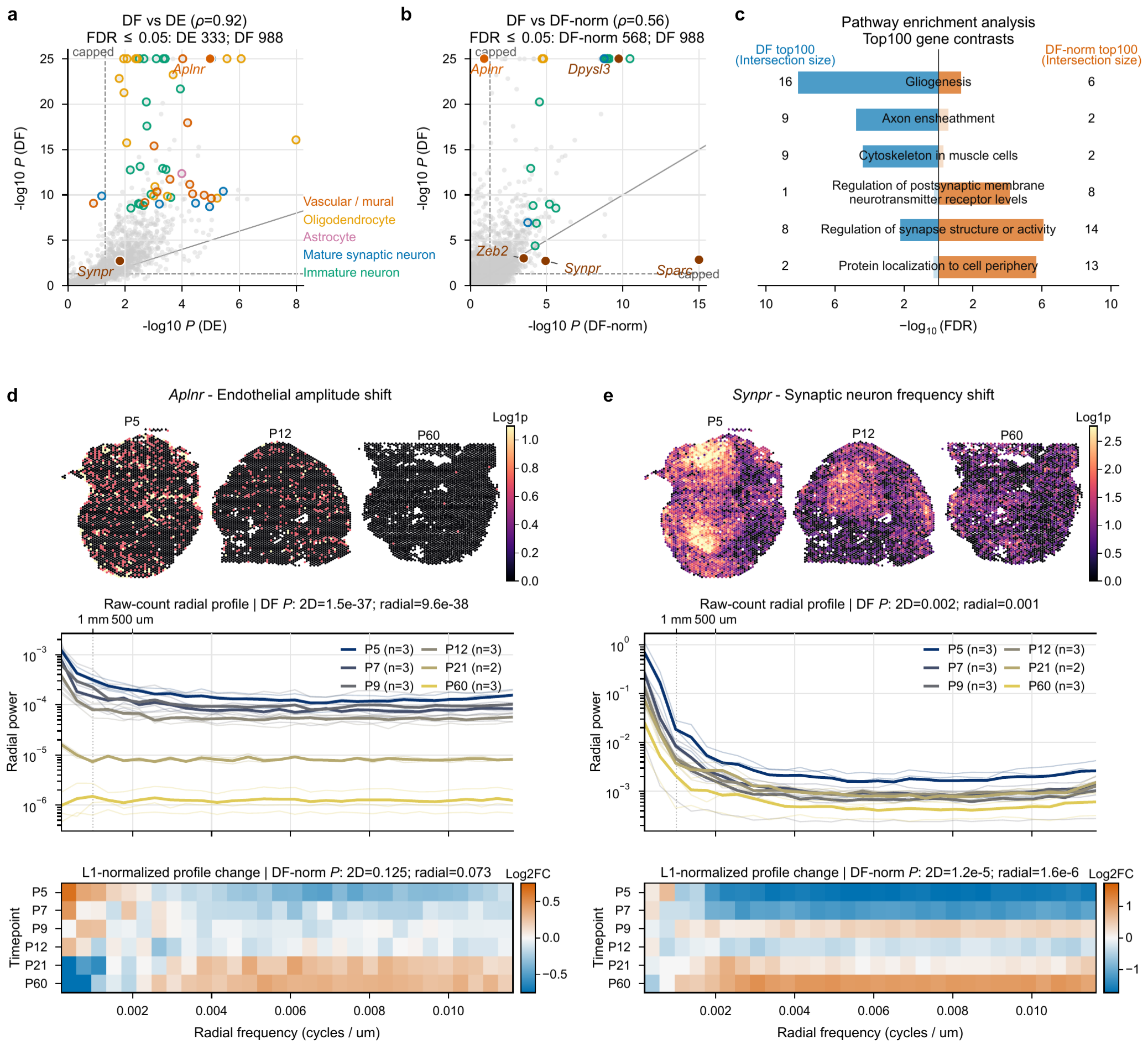}
    \caption{\edFigThreeLegend}
    \label{fig:ed3}
\end{figure}

\begin{figure}[p]
    \centering
    \includegraphics[width=\textwidth]{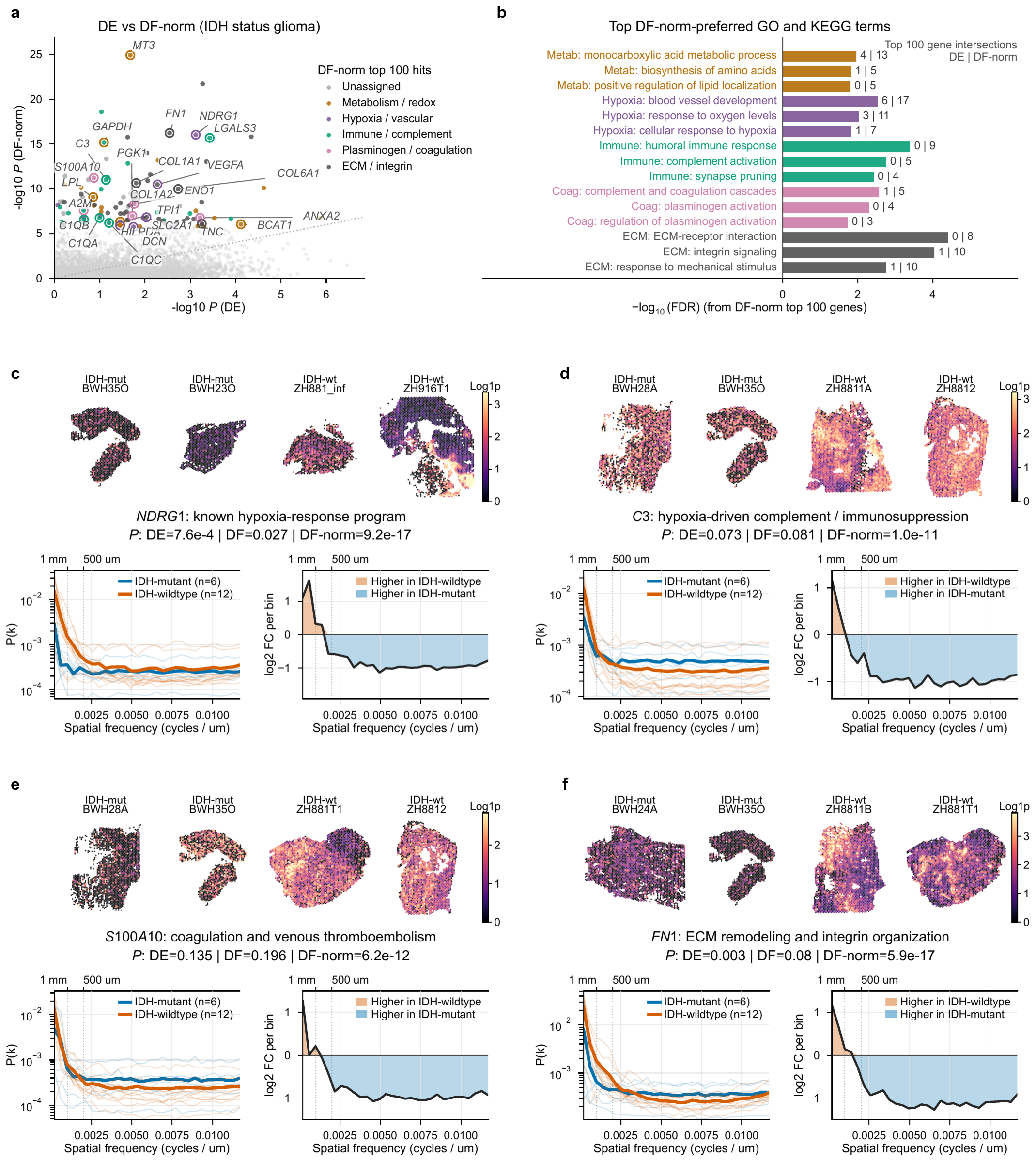}
    \caption{\edFigFourLegend}
    \label{fig:ed4}
\end{figure}

\begin{figure}[p]
    \centering
    \includegraphics[width=\textwidth]{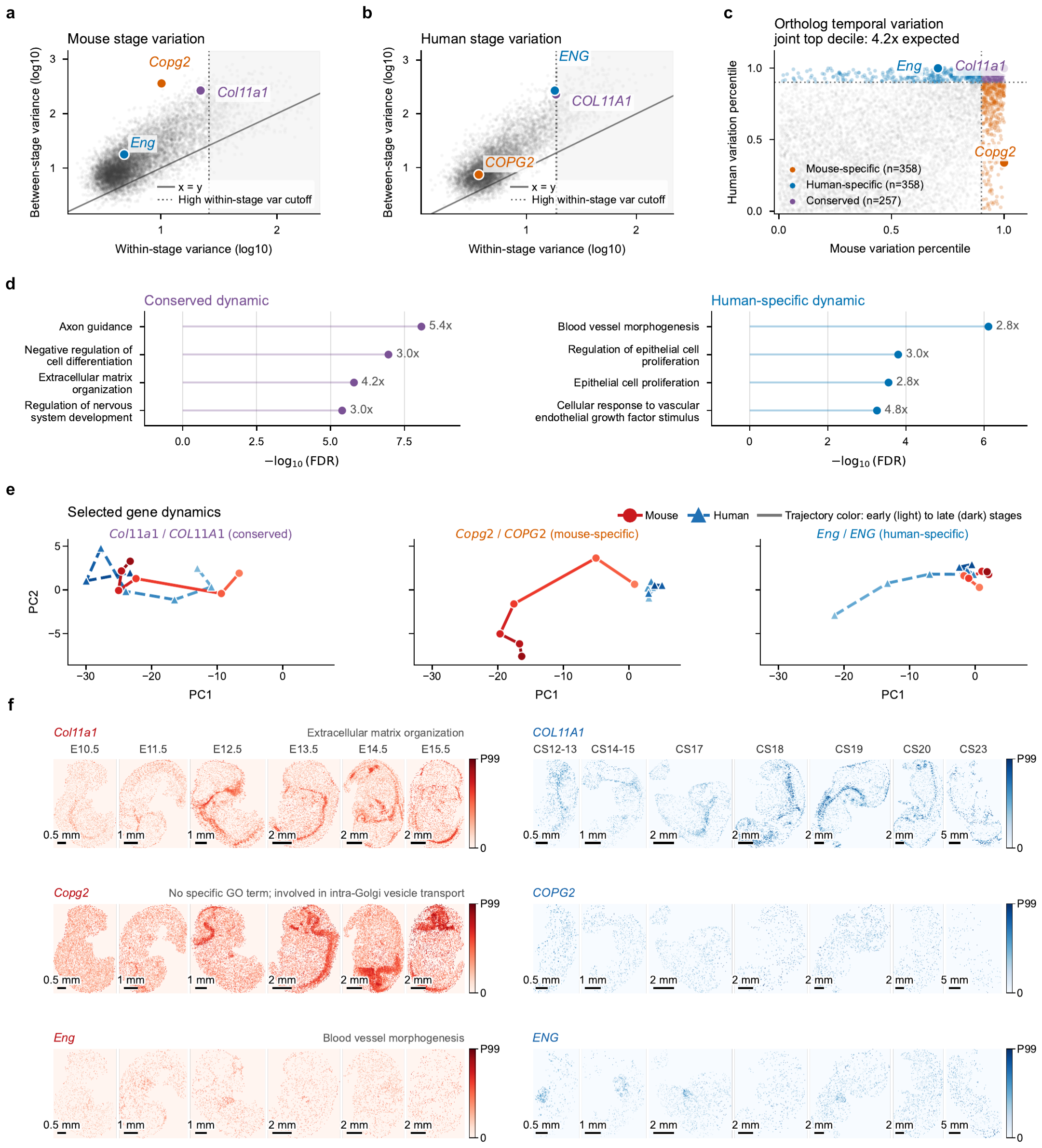}
    \caption{\edFigFiveLegend}
    \label{fig:ed5}
\end{figure}

\clearpage
\endgroup

\section{Methods}
\label{sec:methods}

\subsection{Unification of spatially variable gene tests}
Let $\mathbf{x}=(x_1,\ldots,x_n)^{\mathsf{T}}$ denote gene expression at spatial locations $s_1,\ldots,s_n$. We standardize expression as $\mathbf{z}=(\mathbf{x}-\bar{x}\mathbf{1})/\widehat{\sigma}$, where $\bar{x}$ is the sample mean and $\widehat{\sigma}^2=(n-1)^{-1}\sum_{j=1}^n(x_j-\bar{x})^2$ is the sample variance, and define
\[
    Q_n=\mathbf{z}^{\mathsf{T}}\mathbf{K}\mathbf{z},
    \qquad K_{ij}=k(s_i,s_j).
\]
The symmetric spatial kernel $k$ specifies how locations are related. Under spatial independence, the kernel spectrum gives a weighted chi-square approximation to the null distribution. We obtain $P$ values from this spectral representation or moment matching, accounting for empirical standardization. Mean centering uses $\widetilde{\mathbf{K}}=\mathbf{H}\mathbf{K}\mathbf{H}$, where $\mathbf{H}=\mathbf{I}-\mathbf{1}\mathbf{1}^{\mathsf{T}}/n$; variance estimation requires an additional correction (Supplementary Notes~\ref{supnote:2} and \ref{supnote:eng}).

We connect parametric models to this form through score tests, which measure departures from the fitted null without fitting the alternative. Fixed-effect models yield projection kernels onto spatial features, whereas variance-component tests yield spatial covariance kernels. Generalized linear models scale residuals according to their assumed mean--variance relationship. For generalized linear mixed models, a Gaussian working approximation gives a similar quadratic form in the working residuals. Finally, cell-type-specific models weight spatial kernels using cell-type contributions and observation variance. By contrast, many nonparametric statistics can be rewritten directly in quadratic form. Graph-based tests typically use adjacency, Laplacian or related spatial weights, and HSIC reduces to a quadratic statistic when the expression kernel is linear. For KL divergence, a local quadratic expansion connects it to kernels derived from density estimation. Supplementary Note~\ref{supnote:2} provides the derivations and summarizes the effective kernels in Table~\ref{tab:sv_tests}.

This common quadratic form lets us examine the properties and shared limitations of these tests. Their population signal reflects changes in mean expression across locations (Theorem~\ref{theorem:q_mean_dependence}). Suitable spatial kernels detect any mean-shift pattern with increasing reliability as more locations are measured (Corollary~\ref{corollary:q_test_mean_shift}). However, \(Q\)-tests cannot reliably detect changes that preserve the mean, such as spatial differences in expression variance (Corollary~\ref{corollary:q_test_mean_independent}).

\subsection{Kernel spectral analysis and consistency under the functional view}
Because spatial assays typically measure each location only once, we represent the underlying mean pattern as a function $f(s)$. After removing its constant component, we expand $f$ in the eigenfunctions $\phi_\ell$ of the centered kernel operator. Writing $\hat f_\ell=\langle f,\phi_\ell\rangle$, which may be complex, the functional statistic becomes
\[
    Q(f)=\sum_\ell\lambda_\ell|\hat f_\ell|^2,
\]
where $\lambda_\ell$ is the weight assigned to mode $\ell$. This representation shows which components of a pattern contribute to detection.

Consistent detection of every nonconstant mean pattern requires nonzero weights of the same sign across all nonconstant modes (Theorem~\ref{theorem:q_functional_consistency}). Zero weights create blind spots, while opposite signs allow different components to cancel. The latter explains how Moran's \(I\) can assign a near-zero score to a strongly patterned gene. Spectral cancellation can conceal arbitrarily strong patterns, with the probability of detection approaching zero as signal strength increases (Corollary~\ref{corollary:finite_sample_zero_power}).

The functional view provides general kernel design principles by revealing how spectral weights favor different spatial patterns. No single kernel is optimal for every pattern. Gaussian kernels emphasize smooth, broad trends, whereas graph-Laplacian kernels favor local contrasts. The polynomially decaying spectra of CAR and Mat\'ern kernels retain weight on finer variation while prioritizing smooth patterns. Computational shortcuts can also change detection power. Fixed low-rank approximations omit modes, while hard distance thresholds can introduce negative weights and permit cancellation. Kernels learned from the tested expression can favor dominant patterns or fit noise. This ``double dipping'' requires calibration for kernel selection in addition to multiple-testing correction across genes (Supplementary Note~\ref{supnote:3}).

\subsection{Co-expression analysis using the bivariate \(R\)-statistic}\label{sec:methods_r}
For two standardized expression vectors \(\mathbf{z}_x\) and \(\mathbf{z}_y\), the bivariate statistic \(R=\mathbf{z}_x^{\mathsf T}\mathbf{K}\mathbf{z}_y\) measures their spatial co-localization\autocite{detomaso2021hotspot,li2023spatialdm}. With a positive definite kernel, \(R\) is a kernel-weighted inner product. Dividing by \(\sqrt{Q_xQ_y}\) gives a correlation coefficient that reduces to Pearson correlation when \(\mathbf{K}=\mathbf{I}\). Under the null that the features are independent and spatially unstructured, \(R\) has mean zero and variance \(\operatorname{tr}(\widetilde{\mathbf{K}}^2)\). When no single kernel mode dominates, we use a normal approximation for the null.

\subsection{Scalable implementation of \(Q\)- and \(R\)-tests}\label{sec:methods_q_scaling}
The shared quadratic form allows us to accelerate both kernel evaluation and null calibration (Supplementary Note~\ref{supnote:eng}). A dense implementation requires $O(n^2)$ memory and work per gene, with an additional $O(n^3)$ eigendecomposition for spectral calibration. Moment-based approximations, including Welch--Satterthwaite and four-moment fits\autocite{welch1938significance,satterthwaite1941synthesis,liu2009new}, replace numerical integration with a few spectral summaries. Hutchinson's estimator\autocite{hutchinson1989stochastic} obtains these summaries through a fixed number of kernel--vector products, reducing dense-kernel calibration from cubic to quadratic cost. We account for empirical variance standardization when fitting the null approximation. Given a fixed kernel, the null parameters can be reused across genes. Approximate $P$-value evaluation therefore takes constant time once the statistic is computed.

Kernel structure provides further savings. For CAR kernels, sparse linear solves apply the dense covariance through its sparse precision matrix, avoiding storage of the covariance itself. On suitable two-dimensional lattice-like graphs, factorization costs $O(n^{3/2})$, followed by $O(n\log n)$ work per gene and $O(n\log n)$ memory. The same strategy extends to other graph-based kernels with sparse precision matrices, although costs on general graphs depend on their structure.

For stationary kernels on complete periodic grids, FFTs reduce evaluation to $O(n\log n)$ time and $O(n)$ memory while recovering the full spectrum directly. For masked rasters, null calibration accounts for the observed locations and does not treat filled entries as additional observations. NUFFTs extend this approach to irregular coordinates using an auxiliary grid of $n'$ Fourier modes, requiring $O(n+n'\log n')$ time and $O(n+n')$ memory at fixed numerical accuracy and probe count\autocite{barnett2019parallel}. When the grid size is comparable to the number of observations, this gives near-linear computation and linear memory use. Periodic FFTs preserve the chosen kernel, whereas NUFFT accuracy depends on the retained Fourier representation. Supplementary Note~\ref{supnote:eng} details these algorithms and their calibration limits.

The same operations accelerate \(R\)-statistics. Although evaluating all gene pairs still requires quadratically many pairwise scores, we limit working memory by applying the kernel to blocks of standardized expression and multiplying by a second block. FFT and NUFFT implementations instead use weighted products of Fourier coefficients cached per block. The null variance is the same second spectral power sum already computed for \(Q\)-tests, so calibration is shared across pairs.

\subsection{Spectral representation and differential-frequency analysis of spatial patterns}
\label{sec:methods_spectra}
\label{sec:methods_df}

Following the functional view, we represent spatial patterns on a shared physical-frequency scale. This allows us to compare samples without matching locations or registering anatomical structures. For a given gene in sample $i$, let $x_j^{(i)}$ be its expression at location $s_j^{(i)}$, with $n_i$ locations and sample mean $\bar{x}_i$. Its Fourier coefficients are
\begin{equation}\label{eq:physical_spectrum}
    \widehat f_i(\boldsymbol{\omega})
    =\sum_{j=1}^{n_i}(x_j^{(i)}-\bar{x}_i)
      \exp(-2\pi\mathrm{i}\boldsymbol{\omega}^{\mathsf T}s_j^{(i)}),
\end{equation}
where $\boldsymbol{\omega}\in\mathbb R^2$ is frequency in cycles per micrometer. Moving a pattern changes the phase of its Fourier coefficients. We therefore retain only their power,
$P_i(\boldsymbol{\omega})=n_i^{-2}|\widehat f_i(\boldsymbol{\omega})|^2$.
We normalize power by the squared number of observed locations to express Fourier amplitudes on an average-per-location scale. Unobserved raster entries are excluded from this count.
The transform yields a dense frequency array, even when expression is sparse, so we compress it by binning:
\[
    P_i(b)=\sum_k a_{ibk}P_i(\boldsymbol{\omega}_{ik}),
    \qquad a_{ibk}\geq0,\qquad \sum_k a_{ibk}=1.
\]
Here, $\boldsymbol{\omega}_{ik}$ are the evaluated frequencies, and the weights average power within bins or interpolate it onto a shared grid. Bins can group frequencies by radius alone, producing a one-dimensional radial profile, or by both radius and angle, producing a two-dimensional profile that retains directional information. The two-dimensional profiles are not rotation-invariant and require additional orientation alignment, using either known landmarks or the transcriptome-wide background. The resulting vector $\mathbf p_i=(P_i(1),\ldots,P_i(B))^{\mathsf T}$ contains $B$ frequency features. The number and placement of bins are both tunable. Comparisons use only frequencies supported by all samples. FFT or NUFFT computation has a cost comparable to evaluating the \(Q\)-statistic (Supplementary Note~\ref{supidx:n4:comparison}).

For statistical testing, DF compares log-power profiles, while DF-norm compares profiles scaled to unit total power:
\[
    \widetilde P_i(b)=\frac{P_i(b)}{\sum_{b'=1}^{B}P_i(b')},
    \qquad
    u_i(b)=
    \begin{cases}
        \log\max\{P_i(b),\epsilon\}, & \text{DF},\\
        \log\max\{\widetilde P_i(b),\epsilon\}, & \text{DF-norm},
    \end{cases}
\]
where $\epsilon=10^{-12}$ is a numerical floor. DF-norm thus isolates changes in the relative distribution of power across frequencies, whereas DF captures changes in both total power and its distribution. Both tests fit a linear model across samples for each frequency feature,
\[
    \mathbf U=\mathbf A\boldsymbol{\beta}+\mathbf E,
\]
where $\mathbf U=(u_i(b))\in\mathbb R^{m\times B}$ contains profiles from $m$ samples, $\mathbf A\in\mathbb R^{m\times p}$ is the design matrix with an intercept and predictors, $\boldsymbol{\beta}\in\mathbb R^{p\times B}$ contains frequency-specific coefficients, and $\mathbf E$ contains residuals. For a contrast $\mathbf c\in\mathbb R^p$, we test
\[
    H_0:\boldsymbol{\beta}^{\mathsf T}\mathbf c=\mathbf0,
    \qquad
    \boldsymbol{\Delta}=\widehat{\boldsymbol{\beta}}^{\mathsf T}\mathbf c,
    \qquad
    T^2=\boldsymbol{\Delta}^{\mathsf T}\mathbf W\boldsymbol{\Delta}
       =\sum_{b=1}^{B}w_b\Delta_b^2,
\]
where $\widehat{\boldsymbol{\beta}}$ is estimated by ordinary least squares and $\mathbf W=\operatorname{diag}(w_1,\ldots,w_B)$. By default, we use equal weights, $w_b=1/B$. Alternatively, weights can emphasize selected frequency bins, such as lower frequencies, much like the kernel spectral weighting in \(Q\)-tests. With a binary design, this is a multivariate extension of a two-sample mean comparison; with a continuous predictor, it tests the regression slope. Under a Gaussian null approximation with contrast covariance $\boldsymbol{\Sigma}_{\Delta}$, $T^2$ follows a weighted chi-square distribution, with weights given by the eigenvalues of $\boldsymbol{\Sigma}_{\Delta}^{1/2}\mathbf W\boldsymbol{\Sigma}_{\Delta}^{1/2}$. We estimate covariance between frequency bins from residuals after removing fitted group means or covariate trends. Pooling these residuals across genes stabilizes estimation without incorporating the modeled effects into the noise, assuming the residual covariance is shared across genes. We then obtain the contrast covariance $\boldsymbol{\Sigma}_{\Delta}$ from the design and retained samples. The tail probability is evaluated using Liu's approximation\autocite{liu2009new} (Supplementary Note~\ref{supidx:n4:df}).

As in differential-expression models, sample-level batch effects and confounders can be included in $\mathbf A$. We also provide two optional adjustments to remove unwanted variation from the power profiles before testing. Background normalization divides each gene's power by the across-gene geometric-mean spectrum within its sample:
\[
    P_i^{\mathrm{bg}}(b)
    =\exp\!\left[\frac1G\sum_{g=1}^{G}\log(P_{ig}(b)+\epsilon)\right],
    \qquad
    P_i^{\mathrm{adj}}(b)=\frac{P_i(b)+\epsilon}{P_i^{\mathrm{bg}}(b)}.
\]
This assumes that the spectrum shared across genes within a sample primarily reflects technical variation rather than biological signal. Alternatively, when spot-level covariates (e.g., spot quality or cell-type proportions) may affect the pattern but are not of primary interest, spatial-covariate adjustment regresses each sample's log-power profile on the covariates' log-power profiles:
\[
    \log(P_i(b)+\epsilon)
    =\alpha_i+\sum_{c=1}^{C}\beta_{ic}\log(C_{ic}(b)+\epsilon)+e_i(b),
    \qquad
    P_i^{\mathrm{adj}}(b)=\exp(\widehat e_i(b)),
\]
where $C_{ic}(b)$ is the power profile of covariate $c$ in sample $i$. Unlike a spot-wise multiplicative model, which scales expression at the same location, this frequency-domain adjustment reflects a convolutional relationship that can extend covariate-associated effects across nearby spots (Supplementary Note~\ref{supnote:5.covariates}).

\subsection{General analysis settings shared across datasets}
All analyses, including runtime and memory measurements, were performed on an M5 MacBook Air under a 16\,GB memory budget. Runtime benchmarks covered kernel construction, null calibration and statistic evaluation, excluding input generation, loading, preprocessing and result assembly. Unless noted otherwise, SONIC \(Q\)-tests used a two-sided normal approximation for Moran's \(I\), Welch--Satterthwaite calibration for positive-semidefinite matrix kernels, and four-moment calibration for positive-semidefinite FFT and NUFFT kernels. We applied Benjamini--Hochberg correction\autocite{benjamini1995controlling} separately across genes for each test and report the adjusted $P$ values as FDR values throughout. Analysis-specific significance thresholds are given below.

For DF analysis and joint spectral embedding, we assessed anisotropy across genes and sections within each cohort to choose between radial and two-dimensional binning. For a given gene in section $i$, we measured directionality using the magnitude of the normalized, power-weighted second angular moment,
\begin{equation}\label{eq:spectral_anisotropy}
    A_i=\frac{\left|\sum_{k:\,\boldsymbol\omega_{ik}\neq\mathbf0}
        P_i(\boldsymbol\omega_{ik})\exp(2\mathrm{i}\theta_{ik})\right|}
        {\sum_{k:\,\boldsymbol\omega_{ik}\neq\mathbf0}P_i(\boldsymbol\omega_{ik})},
\end{equation}
where $\theta_{ik}$ is the angle of frequency $\boldsymbol\omega_{ik}$. The factor of two treats opposite directions as the same axis. For spectra with nonzero power, $A_i$ lies between 0 and 1. Values near zero indicate weak net directionality and thus support radial binning (Supplementary Note~\ref{supidx:n4:comparison}). We also examined the spectral participation ratio of the cross-gene log-power covariance within each section and the effective rank of the pooled residual covariance. These diagnostics, together with permutation calibration and gene-ranking stability, guided bin selection.

We excluded zero frequency and limited the bins to the smallest Nyquist frequency shared across sections. For radial spectra, we merged adjacent empty bins until every retained bin had support in every section. Each section contributed one observation to DF, DF-norm, and DE. For the AD, postnatal-brain and glioma comparisons, the spectral tests included an intercept and genotype or age, without background normalization or covariate adjustment. DE tested section-level mean expression using two-sided Welch tests for group comparisons and linear regression for continuous-age analyses. We used section-label permutations to assess $P$-value calibration; reported $P$ values came from the analytic tests.

For pathway enrichment, we used \texttt{clusterProfiler}\autocite{xu2024using} and mapped gene symbols using the species-specific \texttt{org.Mm.eg.db} or \texttt{org.Hs.eg.db} annotation database. Unless noted otherwise, pathway FDR values were obtained by Benjamini--Hochberg correction across pathways.

\subsection{Pattern detection in spatial and lineage-tracing data}
To illustrate how kernel choice reflects the patterns of interest, we analyzed Visium HD mouse small-intestine data\autocite{10x_visium_hd_intestine}, processed with \texttt{SpatialData}\autocite{marconato2025spatialdata} and available at \url{https://s3.embl.de/spatialdata/spatialdata-sandbox/visium_hd_3.0.0_io.zip}. The ``Swiss roll'' geometry reflects the sample-preparation procedure. We applied CAR \(Q\)-tests to 164,800 bins at 16\,$\mu$m resolution and excluded genes with fewer than 10 total counts. CAR kernels used four first-order physical neighbors and a range of autocorrelation values, $\rho$. To fit a kernel that favors one pattern over another, we averaged standardized expression maps within the artifact and immune gene sets, then calculated the positive part of the difference between their smoothed, unit-sum power spectra.

To examine spectral cancellation in lineage-tracing analysis, we used a mouse lung-cancer dataset with Cas9 barcode tracing\autocite{yang2022lineage} and phylogenies reconstructed using \texttt{Cassiopeia}\autocite{jones2020inference} (Zenodo: \href{https://zenodo.org/records/5847462}{5847462}). We used the processed, log-normalized expression supplied by the original study. Gene and co-expression analyses included tumors with 300--3,000 input cells. Within each tumor, we assigned graph weights of 1.0, 0.5, and 0.25 to cell pairs sharing a parent, grandparent, and great-grandparent, respectively. We removed isolated graph nodes before applying symmetric degree normalization, and kept genes detected in at least 5\% of the remaining cells. The CAR kernel used $\rho=0.9$. Significance was assessed using a normal approximation for Moran's \(I\) and Welch's approximation for CAR, with FDR\,\(\leq0.05\). For co-expression analysis, we retained genes significant by both tests and tumors with at least 10 tested gene pairs. Self-pairs were excluded, and each unordered gene pair was counted once.

\subsection{Scalable SPLISOSM analysis of Visium HD and Atera data}
We applied the accelerated SPLISOSM SV and DU tests\autocite{su2026mapping} to adult mouse-brain Visium HD data from 3$'$ short-read sequencing\autocite{10x_visium_hd_mouse_brain_3prime}, targeted FFPE probes\autocite{10x_visium_hd_mouse_brain_ffpe}, and Oxford Nanopore transcript sequencing\autocite{ont_visium_hd_mouse_brain}. For each dataset, we retained peaks, probes or transcripts with at least 10 total counts and detection in at least 0.1\% of 16\,$\mu$m bins. The same feature sets were retained at 8 and 2\,$\mu$m for comparison. We tested within-gene feature proportions using HSIC-IR with CAR kernels ($\rho=0.99$), adjusting neighborhood sizes across resolutions.

For genes with spatially variable usage at FDR\,\(\leq0.01\), we tested associations with transcription-factor and RNA-binding-protein expression using conditional HSIC-GP at 16\,$\mu$m. Gaussian-process regression removed spatial trends from regulator expression, with FFT-based fitting reducing computation (Supplementary Note~\ref{supidx:n4:gp}). Candidate regulators of \textit{Mbp} and \textit{Gabbr1} usage were examined alongside transcription-factor binding data from ReMap 2022\autocite{hammal2022remap} and RNA-binding-protein data from POSTAR3\autocite{zhao2022postar3}.

We also analyzed segmented cells from an Atera human breast-cancer sample\autocite{10x_atera_breast_cancer}, retaining codewords with at least 100 counts and detection in at least 5\% of cells. We excluded spatial-graph components containing fewer than 100 cells. Spatial usage tests used four-neighbor CAR kernels ($\rho=0.99$), while conditional regulator tests used NUFFT-based Gaussian-process fitting at irregular cell coordinates (Supplementary Note~\ref{supidx:n4:gp}). Significance was assessed at FDR\,\(\leq0.01\).

For \textit{ANKRD30A}, transcript classes were defined by the compatibility of full-length, quality-passing PacBio transcripts with the probe regions\autocite{veiga2022comprehensive}, requiring all segments of a probe to overlap exons of the same transcript. These transcripts also served as the reference for StringTie\autocite{pertea2015stringtie} quantification of TCGA-BRCA bulk RNA-seq data. The resulting estimates were highly correlated, likely due to ambiguous short-read assignment among overlapping isoforms. We therefore counted paired-end fragments overlapping the probe regions using featureCounts\autocite{liao2014featurecounts} and normalized them by fragment counts across the \textit{ANKRD30A} genomic region to measure regional usage. Regulator associations were assessed by Spearman correlation with usage ratios in ER-positive tumors. For T47D PRO-seq following dTAG-mediated TRPS1 degradation\autocite{scott2024trps1}, we summed the absolute processed signal from both strands within 1\,kb of each PacBio transcription start site and across the gene body. Log2 fold changes compared each time point with 0\,h, adding $10^{-3}$ to both signals before taking their ratio.

\subsection{Differential-frequency analysis across mouse and human spatial data}
\label{sec:methods_ad}
\label{sec:methods_brain_dev}
\label{sec:methods_glioma}

We reanalyzed spatial transcriptomics data from Chen et al.\autocite{chen2020spatial} (GEO: \href{https://www.ncbi.nlm.nih.gov/geo/query/acc.cgi?acc=GSE152506}{GSE152506}), comprising sections from 18-month-old wild-type and $\mathrm{App}^{NL-G-F}$ mice. Each genotype included two sections from each of two mice. With limited biological replication, we analyzed sections separately without modeling within-individual correlation. Comparable spectral similarity within and between animals suggested that animal identity was not the main source of spectral variation. Following the original study, we excluded \textit{H2-D1}, \textit{Cd63-ps}, and \textit{RP23-269H21.1}. We combined \textit{C4a} and \textit{C4b} as \textit{C4} and mapped \textit{Serpina3n} and \textit{Lyz2} to the control-list names \textit{Serpina3} and \textit{Lyz}. Among genes detected in at least 5\% of spots in every section, we selected the 3,000 with the highest mean expression across sections, excluding mitochondrial and predicted genes. We added curated control genes, yielding 3,020 genes for testing, including 53 plaque-induced genes (PIGs). The primary DF and DF-norm analysis used raw counts and 30 linearly spaced radial bins. We selected these settings using the PIGs as positive controls, comparing raw counts with library-size-normalized log1p counts, linear with logarithmic radial spacing, 8, 12, 16, 24, 30, and 45 bins, and radial with two-dimensional representations.

We reanalyzed 17 anterior coronal Visium sections from de Launoit et al.\autocite{de2026spatiotemporal} (GEO: \href{https://www.ncbi.nlm.nih.gov/geo/query/acc.cgi?acc=GSE313896}{GSE313896}): three sections each at P5, P7, P9, P12, and P60, and two at P21. Following the original study, we excluded spots with fewer than 3,000 or more than 50,000 total counts. Gene selection followed the AD analysis procedure, using annotated gene programs and curated cell-type markers from the original study as controls, for a total of 3,075 genes. The primary analysis used raw counts and $\log_{10}(\mathrm{day}+1)$ as the predictor for DF, DF-norm, and DE. Because these sections showed stronger directional structure, we used two-dimensional spectra sampled at 30 radial-bin centers and 36 angles spanning $[0,\pi)$. We estimated one rotation per section by circular cross-correlation of the transcriptome-wide geometric-mean spectra, then applied it to all genes in that section before resampling to the shared frequency bins.

We reanalyzed six IDH-mutant and 12 IDH-wild-type Visium sections from Greenwald et al.\autocite{greenwald2024integrative} (Zenodo: \href{https://zenodo.org/records/12624860}{12624860}). The wild-type sections came from five patients, with multiple regions sampled from some tumors. With limited biological replication, we analyzed sections separately without modeling within-individual correlation. Comparable normalized spectral similarity within and between patients suggested that patient identity was not the main source of spectral variation. Following the original study, we excluded spots with fewer than 1,000 total counts or more than 20\% mitochondrial counts. Among genes detected in at least 2\% of spots in at least half of the sections, we selected the 7,000 with the highest mean log-normalized expression across sections. Because sequencing depth varied substantially across samples and spots, we normalized counts to 10,000 per spot and applied a log1p transformation. Given the low measured directionality, we computed spectra using 30 linearly spaced radial bins.

To restrict comparisons to sections with sufficient detected expression, DF and DF-norm included a gene only in sections where it was detected in at least 2\% of spots, with at least three sections retained per genotype; DE used all sections. We did not adjust for patient, grade, histology, or tissue composition, so the results describe section-level associations with IDH status in this cohort. For GO Biological Process and KEGG enrichment, we used the top 100 DE and DF-norm genes and the 7,000 tested genes as background. We grouped terms with FDR\,\(\leq0.05\) and lower FDR for DF-norm than for DE into literature-guided functional programs after removing redundant GO terms. Genes were assigned to supporting terms based on DF-norm preference and term size; unmatched genes remained unassigned.

\subsection{Spectral analysis of mouse and human embryogenesis}
\label{sec:methods_atlas}
We reanalyzed Stereo-seq bin50 data from MOSTA\autocite{chen2022spatiotemporal} and HESTA\autocite{pan2026spatiotemporal}, comprising 30 mouse sections spanning E10.5--E15.5 and 28 human sections spanning CS12--23 (four representative sections at each of seven sampled Carnegie stages), with 25-$\mu$m spatial bins. Within each species, we retained genes detected in more than 1\% of observed bins in every section, excluding unnamed, mitochondrial and common long noncoding RNA gene classes. This yielded 10,516 mouse and 7,629 human genes. Because sequencing depth varied substantially across sections and species, counts were normalized to 10,000 per bin and log1p-transformed.

To emphasize lower frequencies that capture broad, smooth trends, we summarized mean-centered FFT spectra using 80 logarithmically spaced radial intervals, retaining 64 shared bins after merging empty intervals. We removed the per-section spectral background by dividing each gene's spectrum by the geometric-mean spectrum across genes, then scaled spectra to unit total power to remove overall amplitude differences, as in DF-norm. We embedded the log-transformed profiles using PCA after standardizing each frequency feature across gene--section profiles. The first five components were used for Euclidean-distance retrieval, temporal comparisons and clustering into five groups with mini-batch $k$-means.

To annotate each pattern cluster, we assigned each gene--stage pair to the nearest cluster using its median embedding coordinates across replicate sections. GO enrichment used deduplicated genes per cluster, separately for mouse and human, with all embedded genes of the corresponding species as background. BH correction was applied globally across all cluster--term pairs within each species. Local mixing of species, species-specific stage, and within-section expression and spatial-variability quintiles was assessed among 50 nearest neighbors. Spatial-variability scores used CAR kernels with $\rho=0.9$ and first-order grid neighbors. We compared local with global category proportions using total-variation distance, corrected for finite-neighbor sampling under a hypergeometric null. Scores ranged from 0 for maximal local segregation to 1 for mixing at or above the random expectation.

To compare developmental dynamics, we decomposed variation across the five components into within- and between-stage mean squares. Within-stage variation measured replicate differences around each stage mean. Between-stage variation measured differences between stage means and the overall mean, weighted by section counts. Both were scaled by their corresponding degrees of freedom. To reduce sensitivity to within-stage heterogeneity, we excluded genes above the species-specific 95th percentile of within-stage variation. We then ranked between-stage variation among 6,146 one-to-one ortholog pairs from \href{https://www.informatics.jax.org/homology.shtml}{Mouse Genome Informatics}. Genes in the top decile in both species formed the shared dynamic set, while those in only one species formed the mouse- or human-preferred sets. Overlap enrichment used a one-sided hypergeometric test. GO analysis used protein-coding orthologs, with the eligible ortholog panel as background and BH correction within each set.

To determine transition strength, we measured distances between adjacent-stage mean embedding coordinates after removing section-wide shifts and correcting squared distances for sampling variance. We also compared adjacent stages using DE, DF and DF-norm. GO enrichment examined genes in the top decile of adjusted movement that also passed DF-norm at FDR\,\(\leq0.05\).

\section{Data availability}
This study reanalyzed existing datasets and generated no new primary data. The mouse Alzheimer's disease sections\autocite{chen2020spatial} (\href{https://www.ncbi.nlm.nih.gov/geo/query/acc.cgi?acc=GSE152506}{GSE152506}), postnatal mouse-brain sections\autocite{de2026spatiotemporal} (\href{https://www.ncbi.nlm.nih.gov/geo/query/acc.cgi?acc=GSE313896}{GSE313896}), and human embryo Stereo-seq data (HESTA\autocite{pan2026spatiotemporal}, \href{https://www.ncbi.nlm.nih.gov/geo/query/acc.cgi?acc=GSE326326}{GSE326326}) are available from GEO. The human glioma Visium sections\autocite{greenwald2024integrative} (record \href{https://zenodo.org/records/12624860}{12624860}) and the mouse lung-cancer lineage-tracing data and inferred phylogenetic trees\autocite{yang2022lineage} (record \href{https://zenodo.org/records/5847462}{5847462}) are available from Zenodo. The mouse embryo Stereo-seq data (MOSTA\autocite{chen2022spatiotemporal}) are available from CNGB under the accession number \href{https://db.cngb.org/stomics/datasets/STDS0000058/}{STDS0000058}. The mouse small-intestine Visium HD data were obtained as a SpatialData\autocite{marconato2025spatialdata} archive from \url{https://s3.embl.de/spatialdata/spatialdata-sandbox/visium_hd_3.0.0_io.zip}. 

For SPLISOSM analyses, adult mouse brain Visium HD data (3$'$ fresh-frozen\autocite{10x_visium_hd_mouse_brain_3prime}, FFPE\autocite{10x_visium_hd_mouse_brain_ffpe}) and Atera human breast-cancer data\autocite{10x_atera_breast_cancer} are available from the 10x Genomics datasets at \url{https://www.10xgenomics.com/datasets/}. Adult mouse brain Visium HD ONT data\autocite{ont_visium_hd_mouse_brain} are available from the Oxford Nanopore Visium HD resource at \url{https://epi2me.nanoporetech.com/visium_hd_2025.06/}. Bulk long-read breast-cancer transcript data are available from the BRCA-LRseq\autocite{veiga2022comprehensive} repository at \url{https://github.com/TheJacksonLaboratory/BRCA-LRseq-pipeline}. TRPS1-degradation PRO-seq\autocite{scott2024trps1} data are available from GEO under accession number \href{https://www.ncbi.nlm.nih.gov/geo/query/acc.cgi?acc=GSE236172}{GSE236172}. TCGA-BRCA data are available through the Genomic Data Commons although access to sequencing BAM files is controlled by the GDC.

\section{Code availability}
Spatial univariate \(Q\)-tests, bivariate \(R\)-tests, and differential frequency tests DF and DF-norm are implemented in the \texttt{sonic-spatial} Python package, available under the BSD-3-Clause license from the GitHub repository \url{https://github.com/JiayuSuPKU/sonic}. Accelerated SPLISOSM tests are available in \texttt{splisosm} v1.2.0 from the GitHub repository \url{https://github.com/JiayuSuPKU/SPLISOSM/}. Analysis scripts and materials for reproducing this paper are available in the accompanying paper repository \url{https://github.com/JiayuSuPKU/sonic-paper}.

\section{Acknowledgements}
This work was funded by the Eric and Wendy Schmidt Center at the Broad Institute of MIT and Harvard (to J.S.), the Edward P. Evans Center for MDS at Columbia University (to J.S. and R.R.), the National Institutes of Health, National Cancer Institute (grant nos. R35CA253126, U01CA243073, P01CA285250 to R.R.), and the National Science Foundation (grant no. CAREER DBI2146398 to D.A.K.). D.Y. is supported by a Damon Runyon Dale Frey Award, an NCI Transition Career Development Award 1K22CA289207, an NIH Director's New Innovator Award 1DP2OD037078, and a Lung Cancer Research Foundation Leading Edge Award. Any opinions, findings and conclusions or recommendations expressed in this material are those of the authors and do not necessarily reflect the views of the NSF or the NIH.

\section{Author contributions}
J.S. conceived the project, implemented SONIC, and conducted the theoretical and data analyses. J.H.F. contributed to the theoretical analysis of pattern detection. H.W. and D.Y. contributed to the lineage-tracing analysis. X.W. contributed to the embryonic atlas analysis and interpretation. D.A.K. and R.R. supervised the project. J.S. took the lead in writing the article with input from all authors. All authors read and approved the final paper.

\section{Competing interests}
R.R. is a founder of Genotwin and a member of the Scientific Advisory Board of DiaTech and Flahy. D.Y. is a co-founder of DEM Biopharma. None of these activities are related to the work described in this paper. The other authors declare no competing interests.

\clearpage

\printbibliography[title={References}]
\end{refsection}

\clearpage
\begingroup
\nolinenumbers
\setlength{\parskip}{0pt}
\pdfbookmark[0]{Supplementary material}{supplementary.index}
\section*{Supplementary material}
\label{sec:supplementary_index}
\small

\newcommand{\suppindexfigure}[3]{%
  \par\noindent\hangindent=2.4em\hangafter=1
  \makebox[2.4em][l]{\hyperref[#1]{\textbf{#2}}}%
  \hyperref[#1]{#3}\dotfill\nobreak\hspace{0.3em}\makebox[1.2em][r]{\pageref*{#1}}\par}
\newcommand{\suppindexnote}[2]{%
  \par\noindent\hangindent=1.8em\hangafter=1
  \makebox[1.8em][l]{\hyperref[#1]{\textbf{#2}}}%
  \hyperref[#1]{\textbf{\nameref*{#1}}}\dotfill\nobreak\hspace{0.3em}\makebox[1.2em][r]{\pageref*{#1}}\par}
\newcommand{\suppindexsection}[1]{%
  \par\begingroup\leftskip=1.4em
  \noindent\hyperref[#1]{\nameref*{#1}}\dotfill\nobreak\hspace{0.3em}\makebox[1.2em][r]{\pageref*{#1}}\par\endgroup}

\subsection*{Supplementary figures}
\small
\setlength{\parskip}{0pt}
\suppindexfigure{fig:s1}{S1}{\suppFigOneTitle}
\suppindexfigure{fig:s2}{S2}{\suppFigTwoTitle}
\suppindexfigure{fig:s3}{S3}{\suppFigThreeTitle}
\suppindexfigure{fig:s4}{S4}{\suppFigFourTitle}
\suppindexfigure{fig:s5}{S5}{\suppFigFiveTitle}
\suppindexfigure{fig:s6}{S6}{\suppFigSixTitle}
\suppindexfigure{fig:s7}{S7}{\suppFigSevenTitle}
\suppindexfigure{fig:s8}{S8}{\suppFigEightTitle}
\suppindexfigure{fig:s9}{S9}{\suppFigNineTitle}
\suppindexfigure{fig:s10}{S10}{\suppFigTenTitle}
\suppindexfigure{fig:s11}{S11}{\suppFigElevenTitle}
\suppindexfigure{fig:s12}{S12}{\suppFigTwelveTitle}
\suppindexfigure{fig:s13}{S13}{\suppFigThirteenTitle}
\suppindexfigure{fig:s14}{S14}{\suppFigFourteenTitle}
\suppindexfigure{fig:s15}{S15}{\suppFigFifteenTitle}
\suppindexfigure{fig:s16}{S16}{\suppFigSixteenTitle}
\suppindexfigure{fig:s17}{S17}{\suppFigSeventeenTitle}
\suppindexfigure{fig:s18}{S18}{\suppFigEighteenTitle}
\vspace{0.7em}
\subsection*{Supplementary notes}
\small
\setlength{\parskip}{0pt}
\suppindexnote{supnote:1}{1}
\suppindexsection{supidx:n1:kernel-function}
\suppindexsection{supidx:n1:kernel-operator}
\suppindexsection{supidx:n1:rkhs}
\suppindexsection{supnote:1.cnn}
\addvspace{0.2em}
\suppindexnote{supnote:2}{2}
\suppindexsection{supnote:2.formulation}
\suppindexsection{supnote:2.asymptotic}
\suppindexsection{supidx:n2:moran}
\suppindexsection{supidx:n2:linear-models}
\suppindexsection{supidx:n2:glm-tests}
\suppindexsection{supidx:n2:hsic}
\suppindexsection{supidx:n2:other-tests}
\suppindexsection{supnote:2.ctsv}
\suppindexsection{supidx:n2:multivariate}
\addvspace{0.2em}
\suppindexnote{supnote:3}{3}
\suppindexsection{supidx:n3:functional}
\suppindexsection{supidx:n3:spectra}
\suppindexsection{supidx:n3:consistency}
\suppindexsection{supnote:3.finite_power}
\suppindexsection{supnote:3.examples}
\suppindexsection{supnote:3.kernel_design}
\addvspace{0.2em}
\suppindexnote{supnote:eng}{4}
\suppindexsection{supidx:n4:null}
\suppindexsection{supidx:n4:operators}
\suppindexsection{supidx:n4:fourier}
\suppindexsection{supidx:n4:complexity}
\suppindexsection{supidx:n4:bivariate}
\suppindexsection{supidx:n4:gp}
\suppindexsection{supidx:n4:comparison}
\suppindexsection{supidx:n4:df}
\addvspace{0.2em}
\suppindexnote{supnote:5}{5}
\suppindexsection{supnote:5.similarity}
\suppindexsection{supnote:5.covariates}
\par

\renewcommand{\thefigure}{S\arabic{figure}}
\renewcommand{\theHfigure}{supplementary.\arabic{figure}}
\setcounter{figure}{0}
\newcommand{\suppfigurepage}[3]{%
    \clearpage
    \refstepcounter{figure}\label{fig:s#1}%
    \pdfbookmark[2]{Supplementary Figure S#1: #2}{suppfig.S#1}%
    {\centering
    \includegraphics[width=\textwidth,height=0.9\textheight,keepaspectratio]{figures/supp/S#1.png}\par}%
    \vspace{4pt}%
    {\footnotesize\emergencystretch=1em\clubpenalty=10000\widowpenalty=10000
    \noindent\textbf{Supplementary Figure \thefigure.} #3\par}%
}
\clearpage
\pdfbookmark[1]{Supplementary figures}{supplementary.figures}
\suppfigurepage{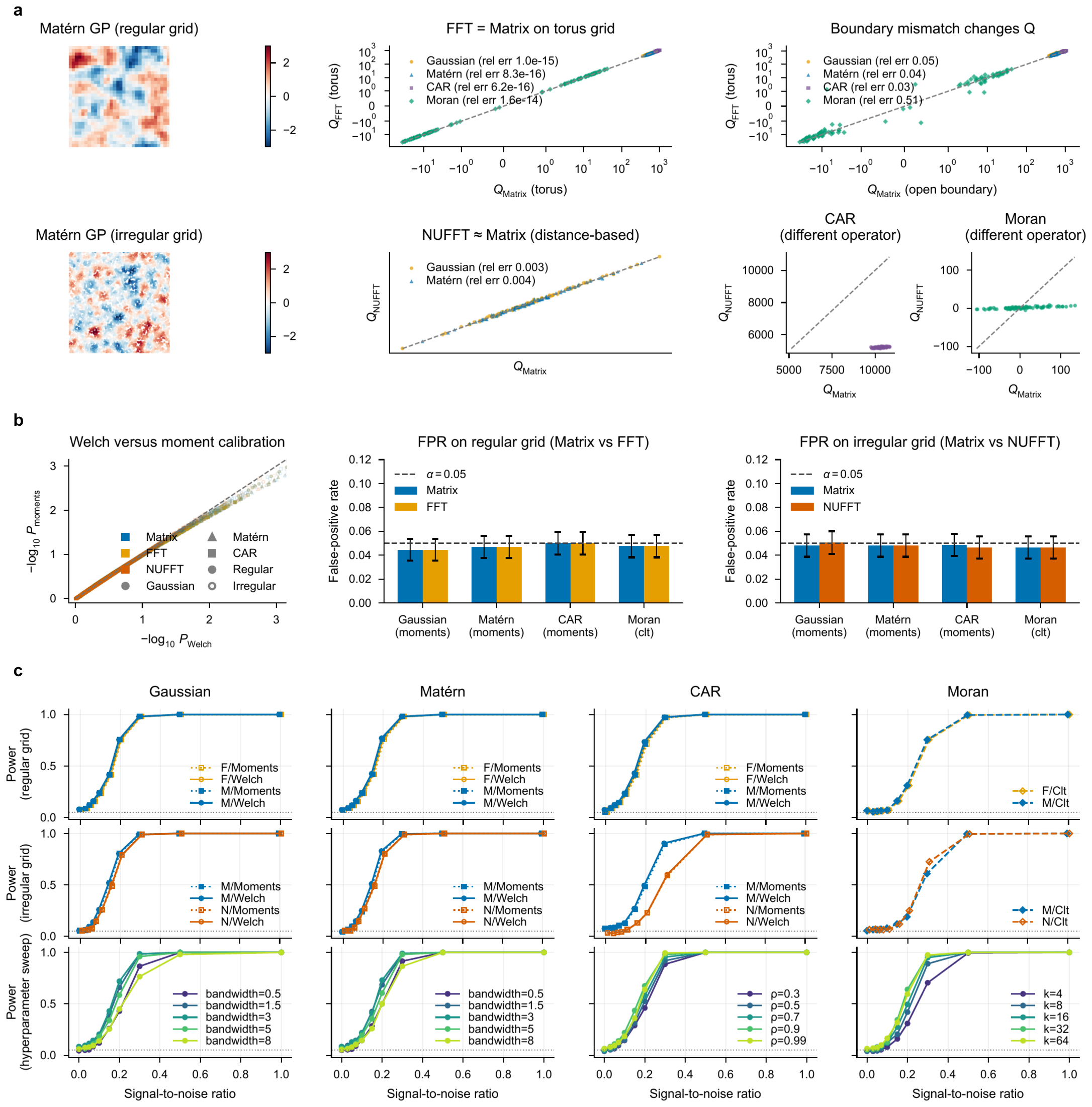}{\suppFigOneTitle}{\suppFigOneLegend}
\suppfigurepage{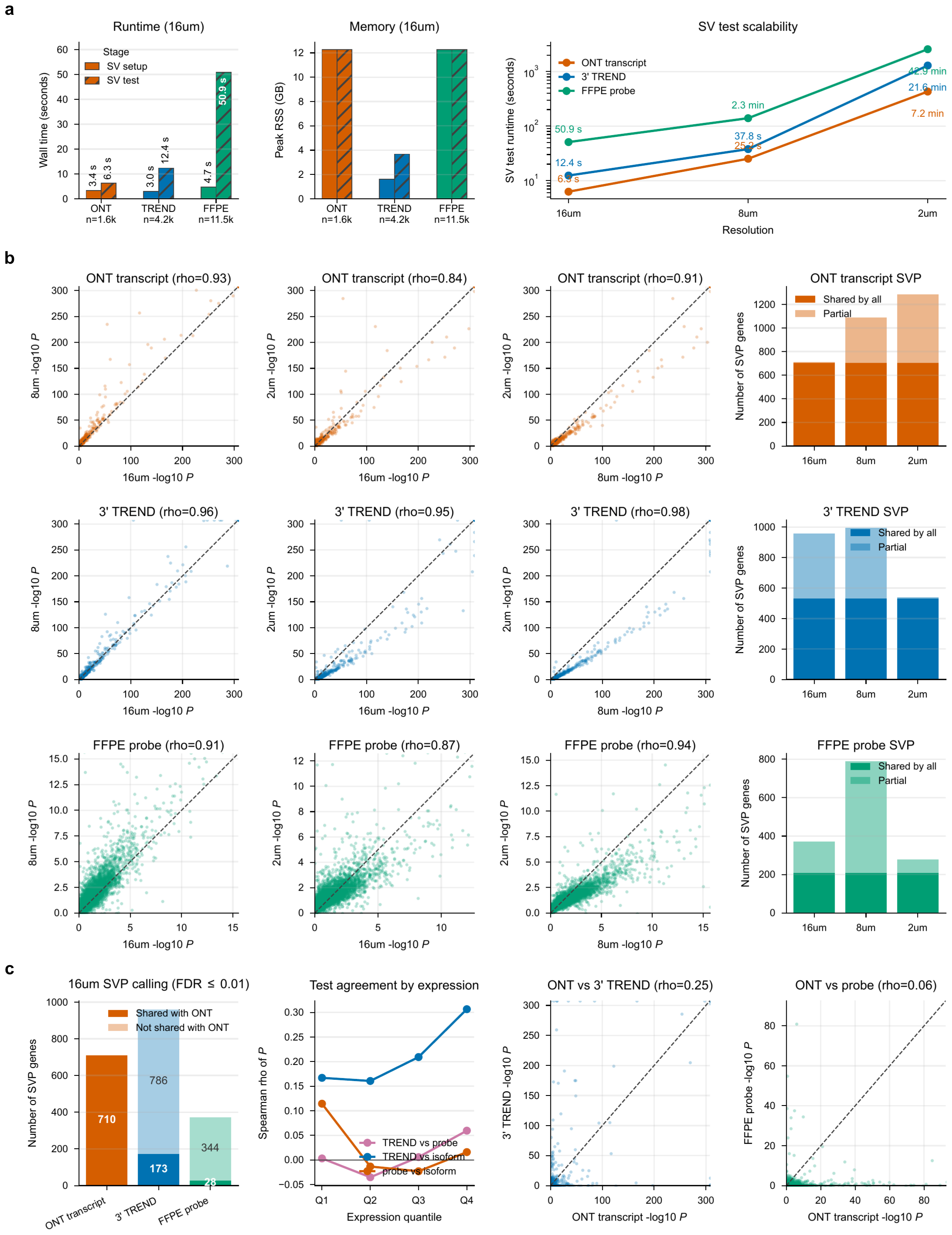}{\suppFigTwoTitle}{\suppFigTwoLegend}
\suppfigurepage{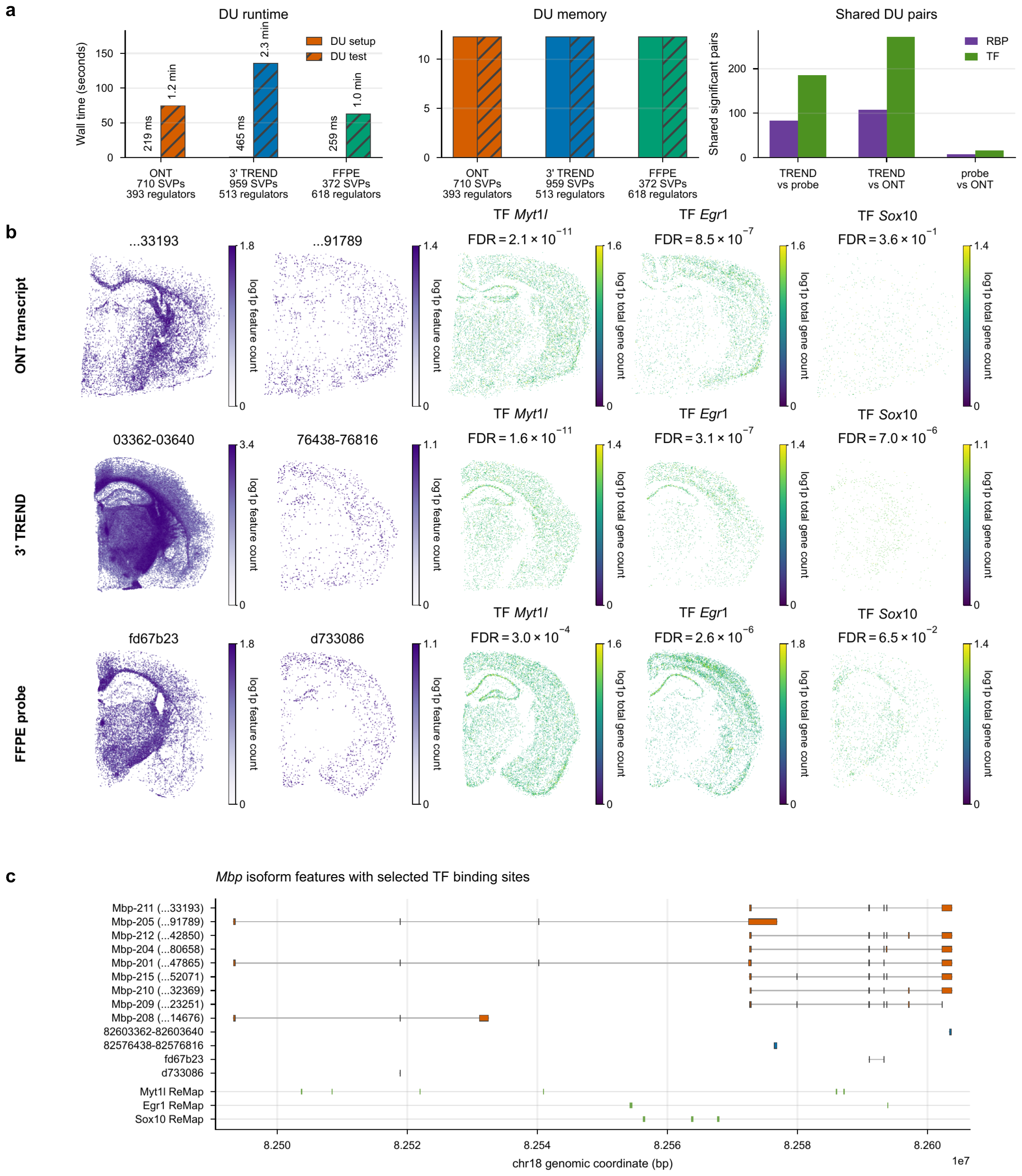}{\suppFigThreeTitle}{\suppFigThreeLegend}
\suppfigurepage{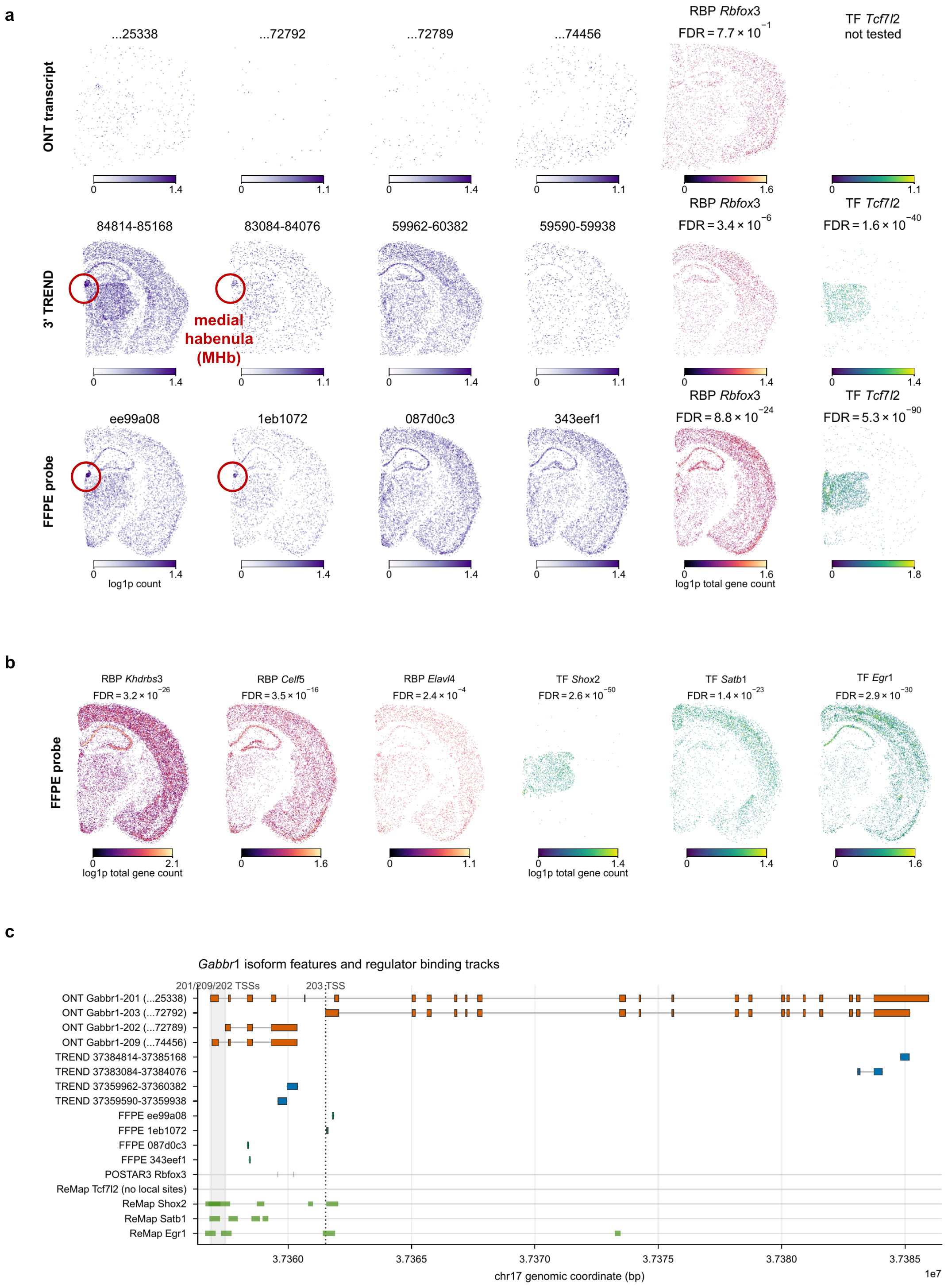}{\suppFigFourTitle}{\suppFigFourLegend}
\suppfigurepage{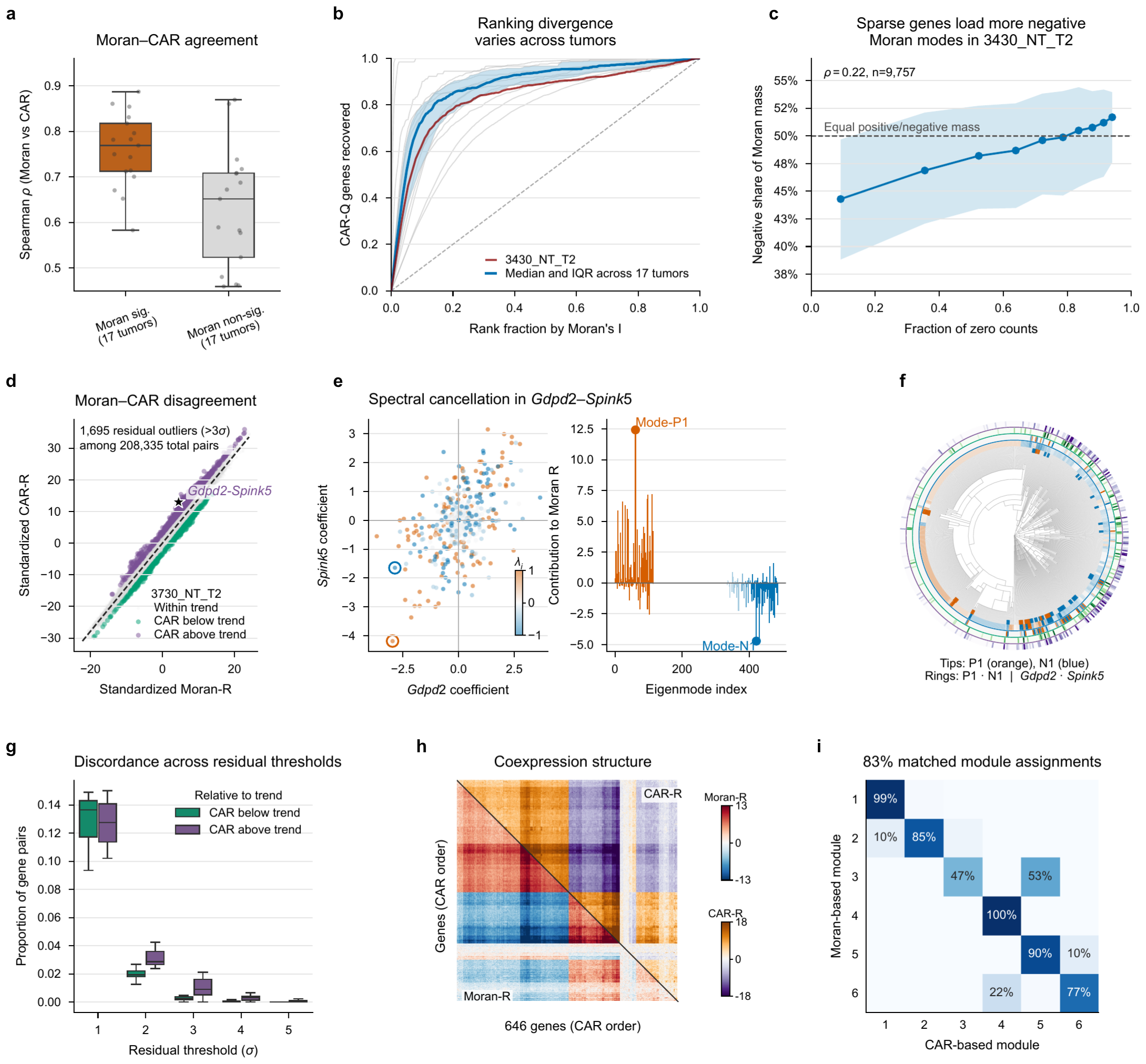}{\suppFigFiveTitle}{\suppFigFiveLegend}
\suppfigurepage{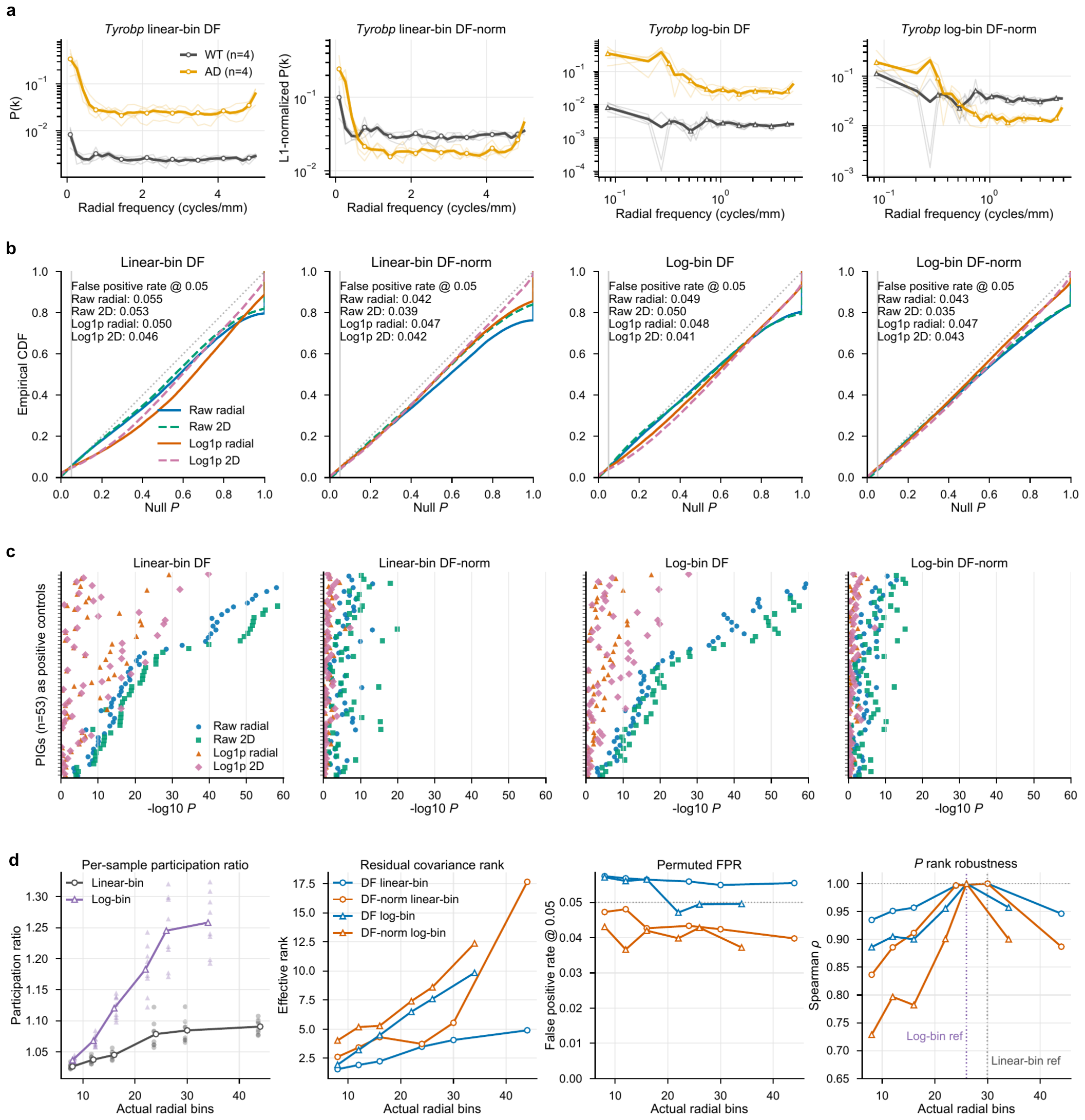}{\suppFigSixTitle}{\suppFigSixLegend}
\suppfigurepage{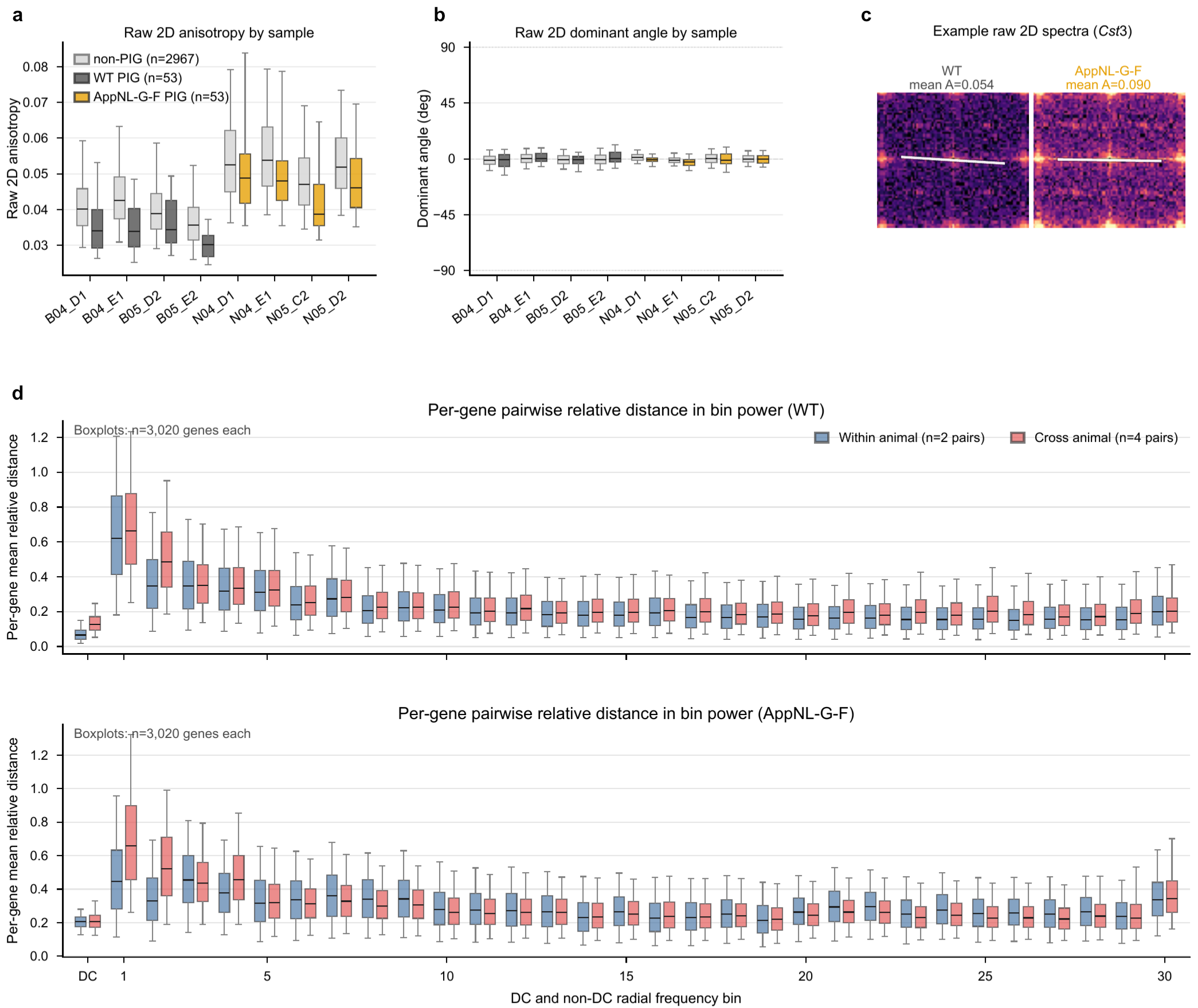}{\suppFigSevenTitle}{\suppFigSevenLegend}
\suppfigurepage{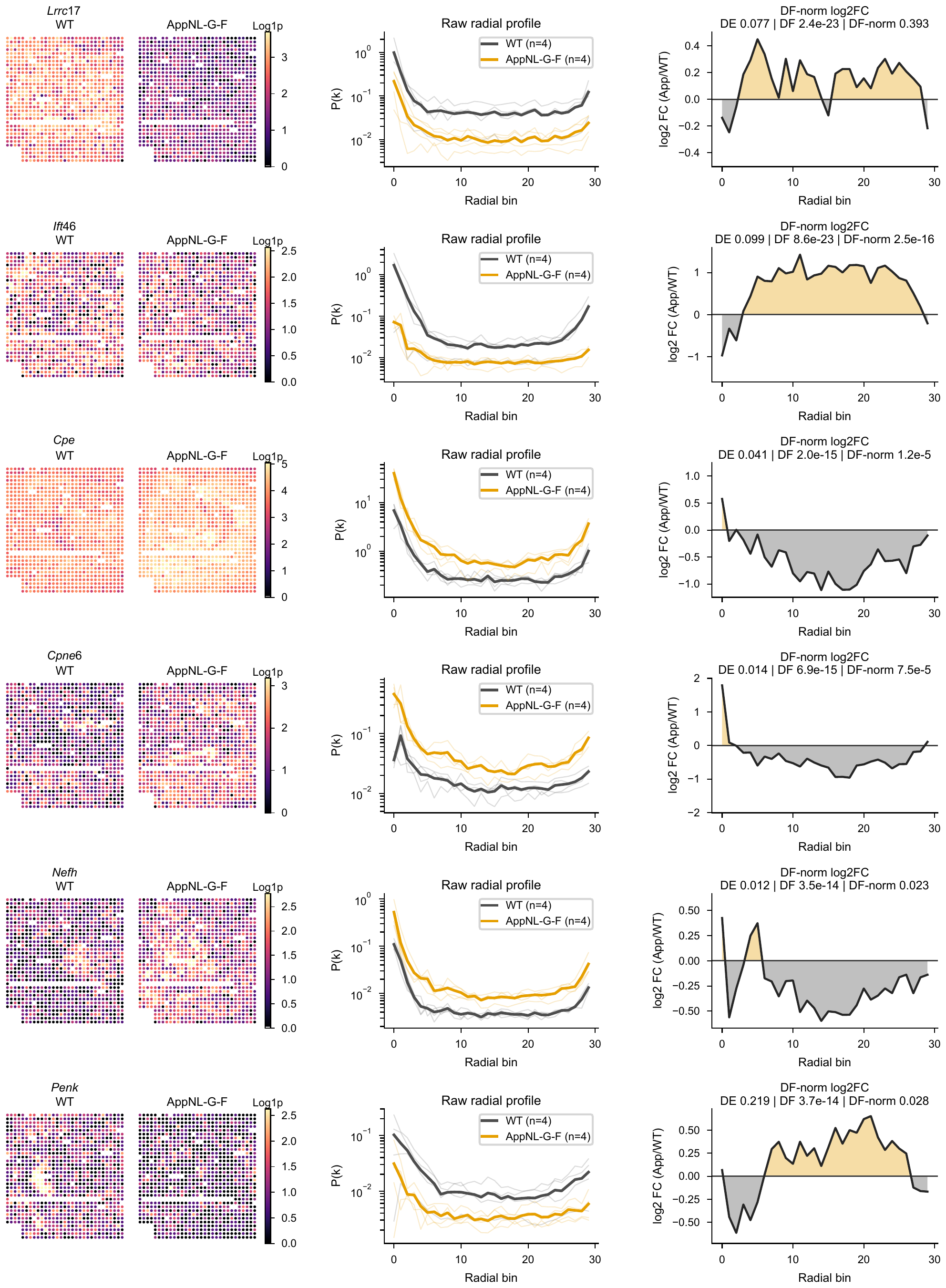}{\suppFigEightTitle}{\suppFigEightLegend}
\suppfigurepage{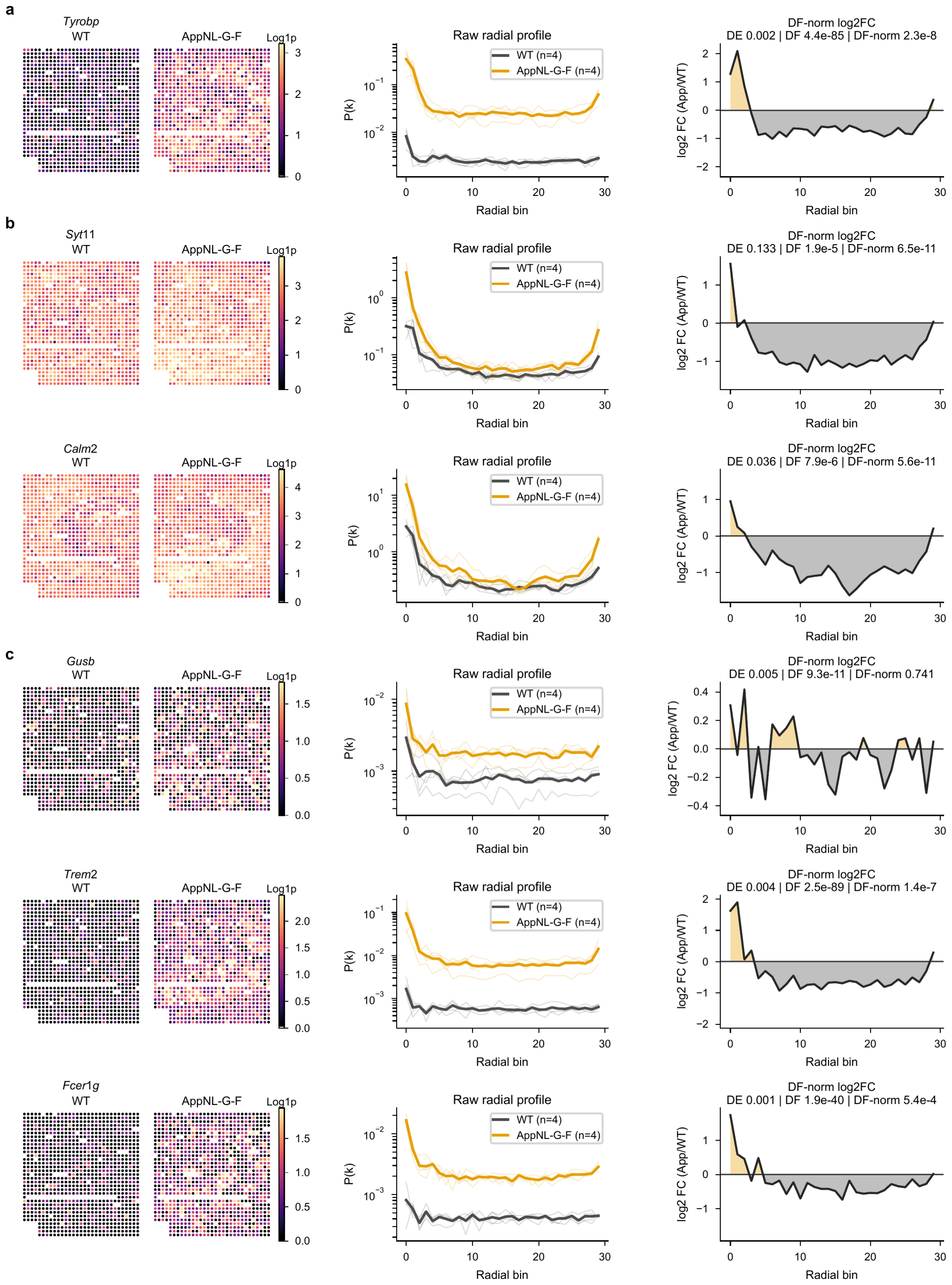}{\suppFigNineTitle}{\suppFigNineLegend}
\suppfigurepage{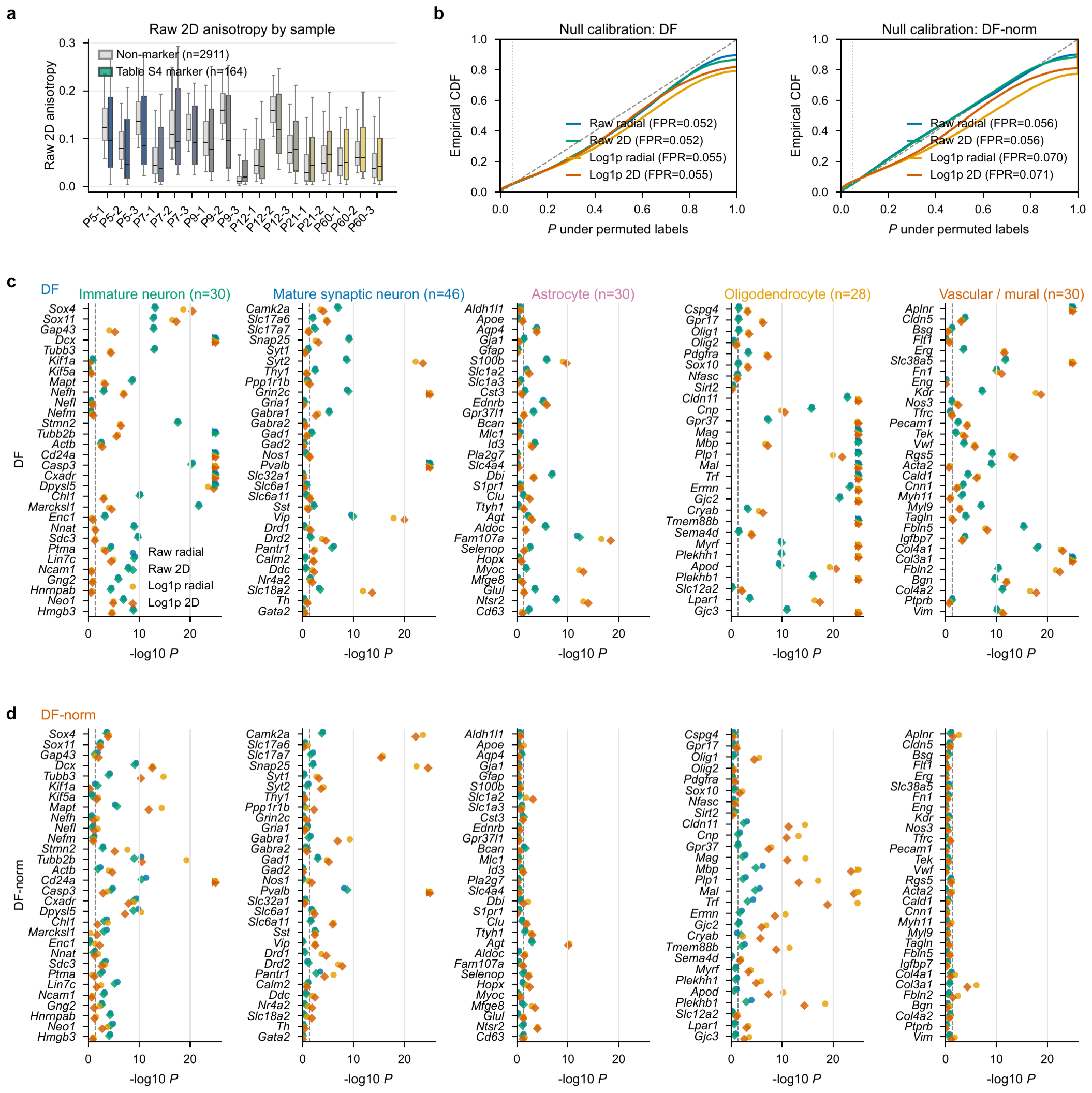}{\suppFigTenTitle}{\suppFigTenLegend}
\suppfigurepage{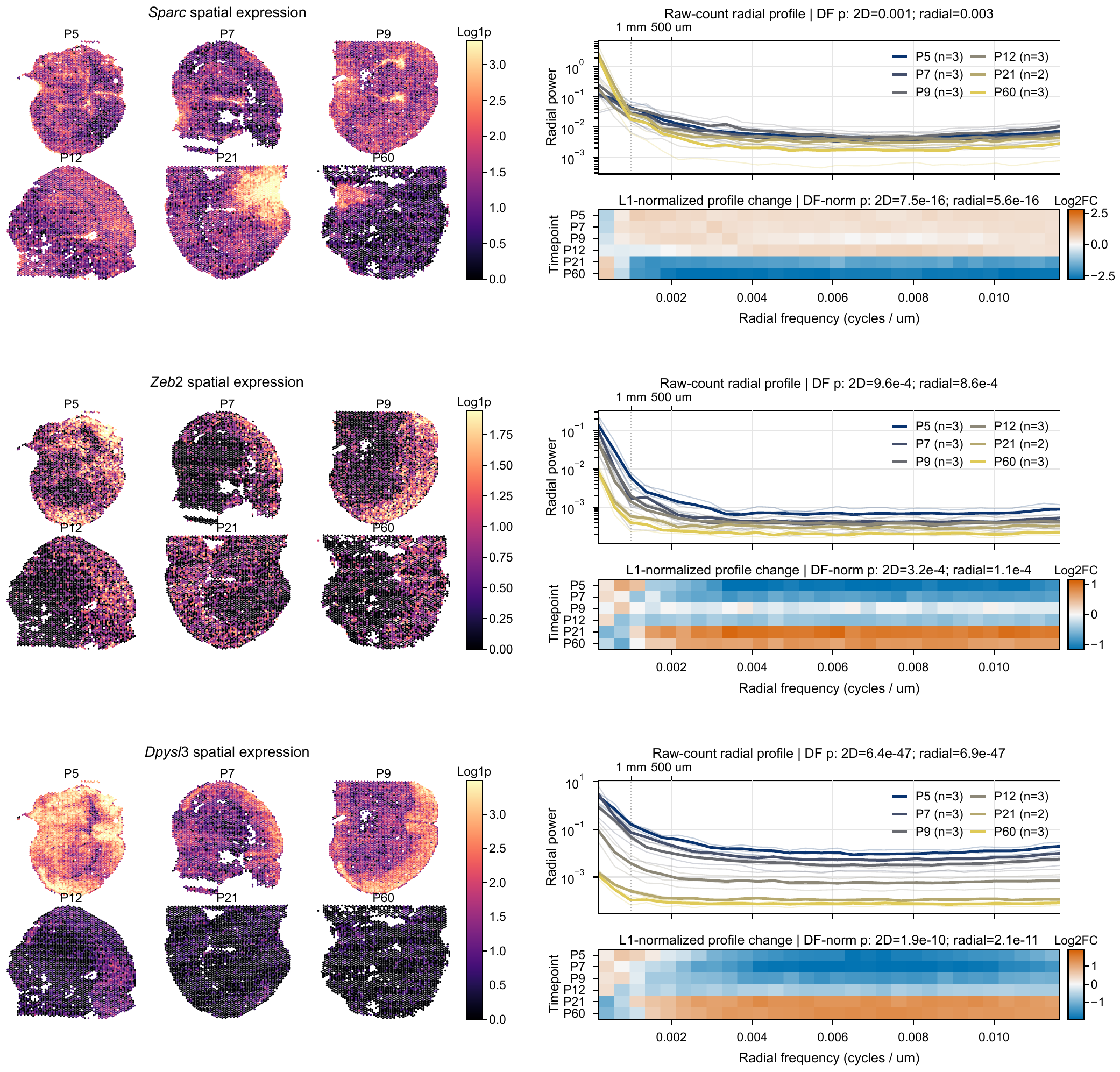}{\suppFigElevenTitle}{\suppFigElevenLegend}
\suppfigurepage{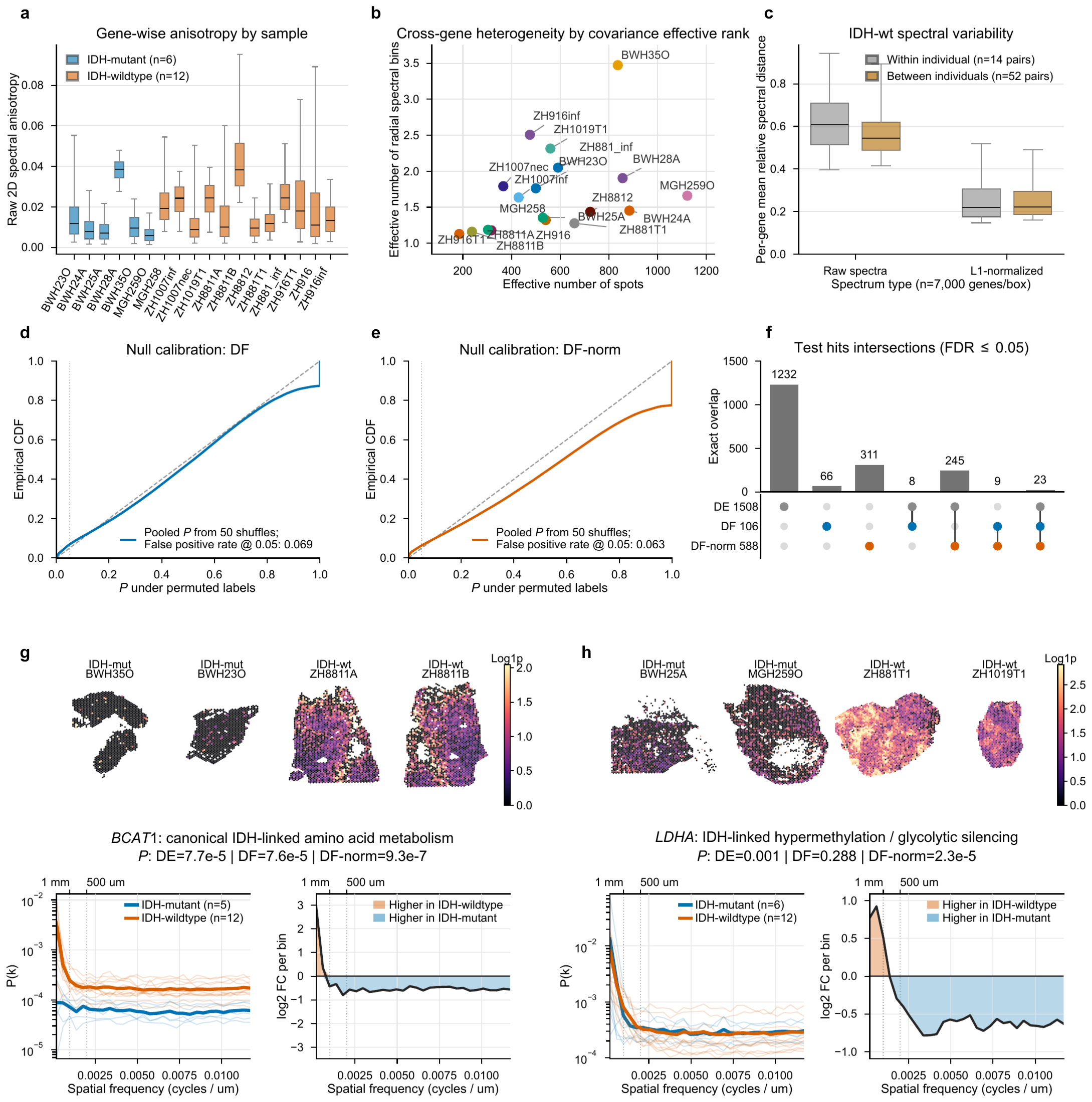}{\suppFigTwelveTitle}{\suppFigTwelveLegend}
\suppfigurepage{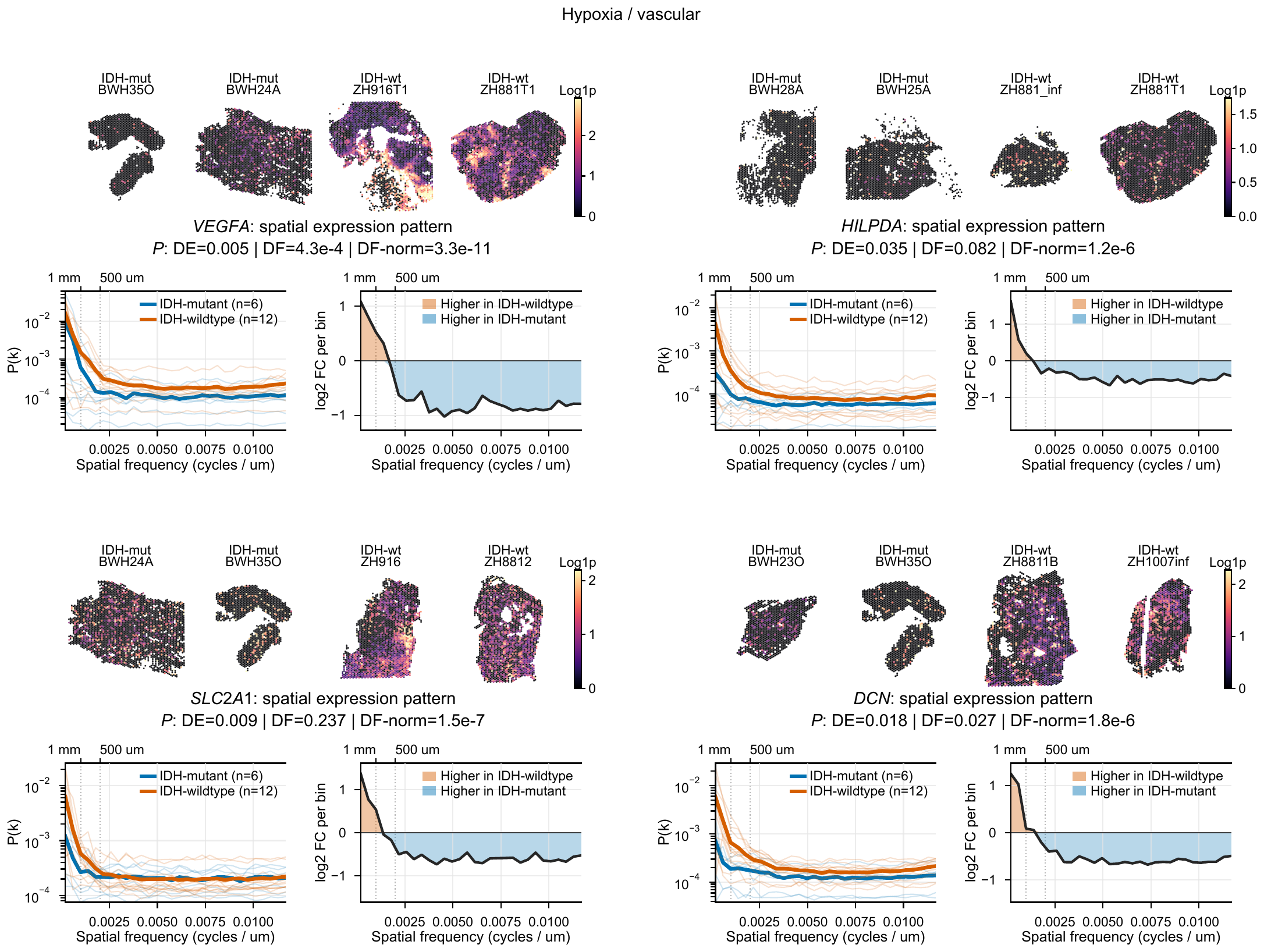}{\suppFigThirteenTitle}{\suppFigThirteenLegend}
\suppfigurepage{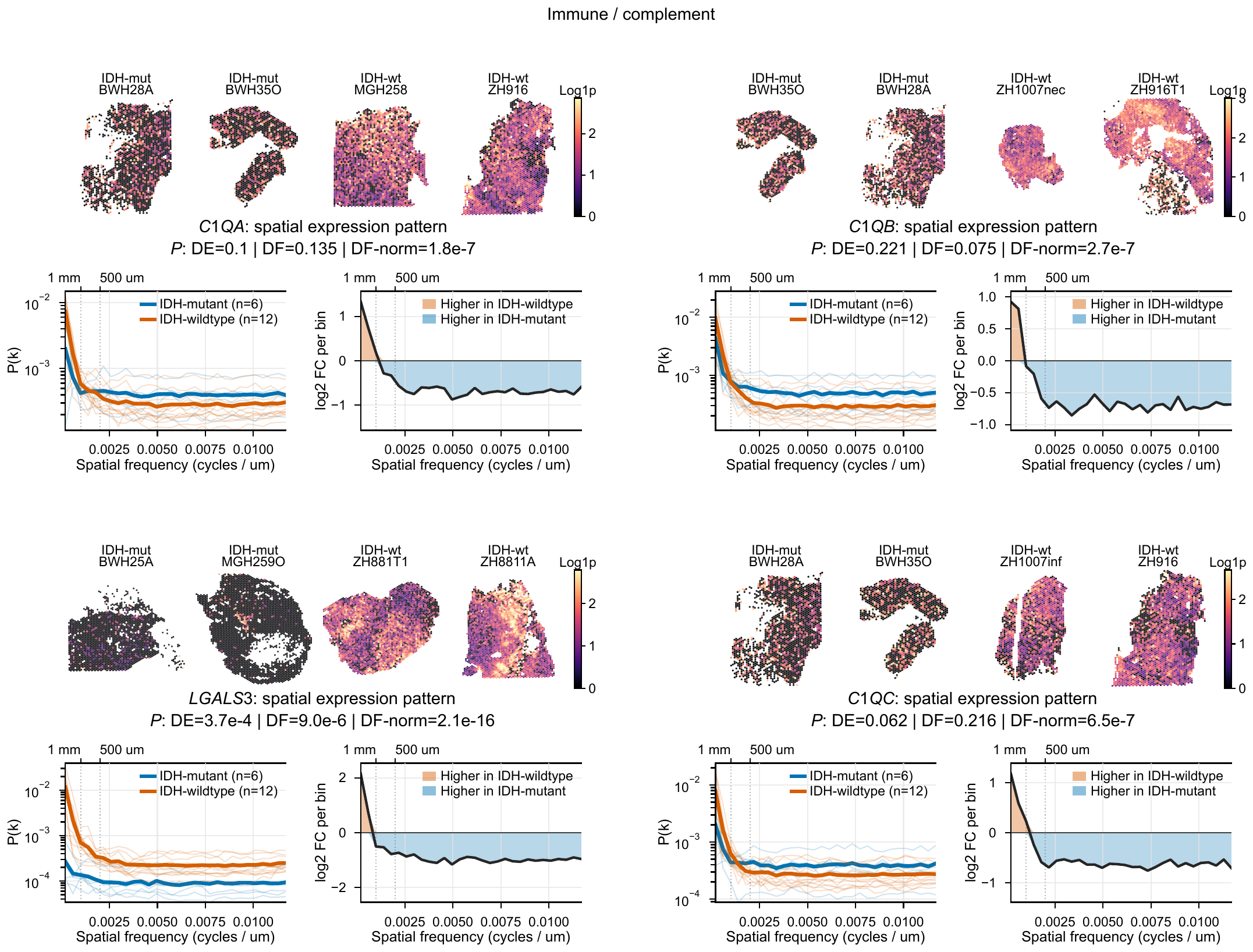}{\suppFigFourteenTitle}{\suppFigFourteenLegend}
\suppfigurepage{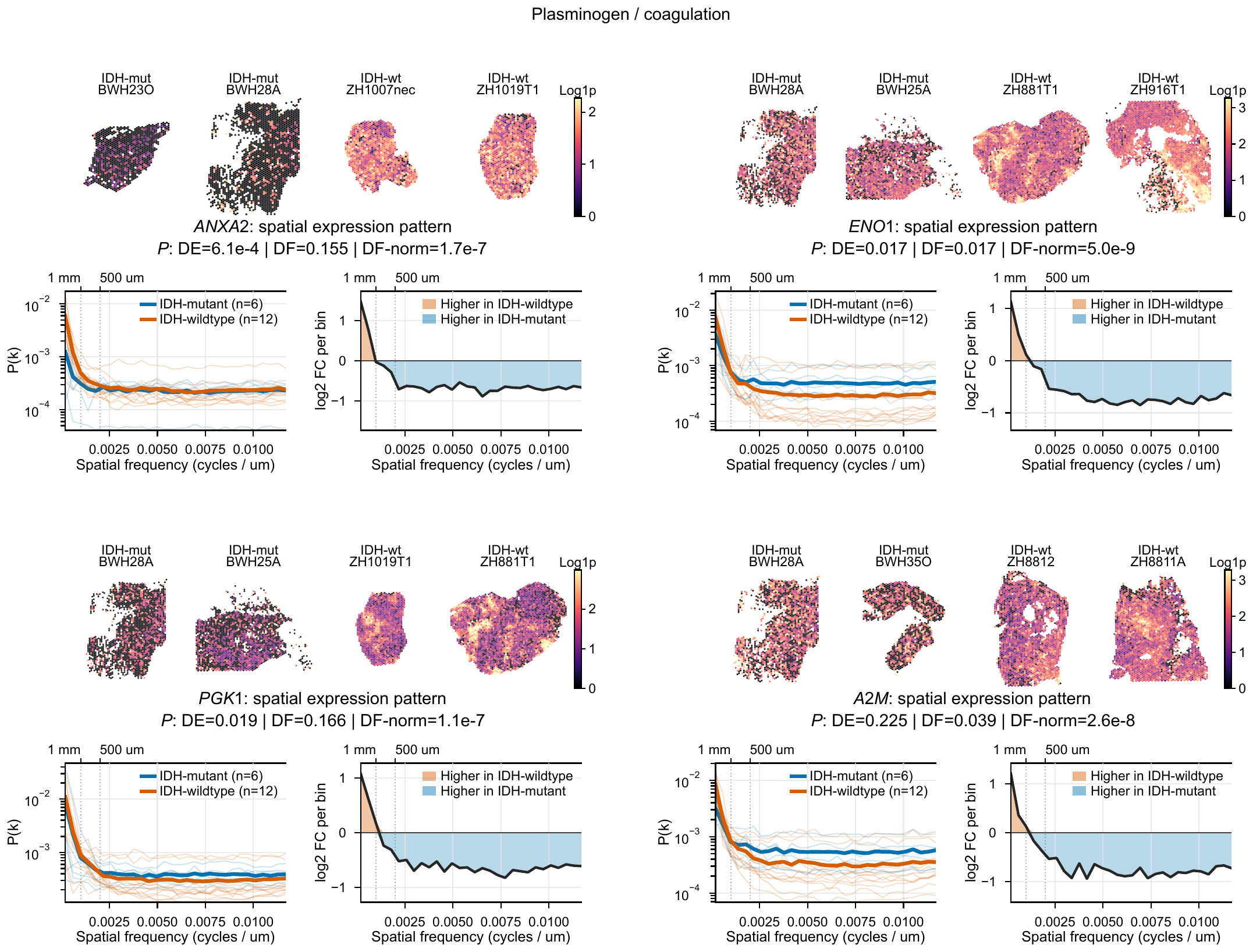}{\suppFigFifteenTitle}{\suppFigFifteenLegend}
\suppfigurepage{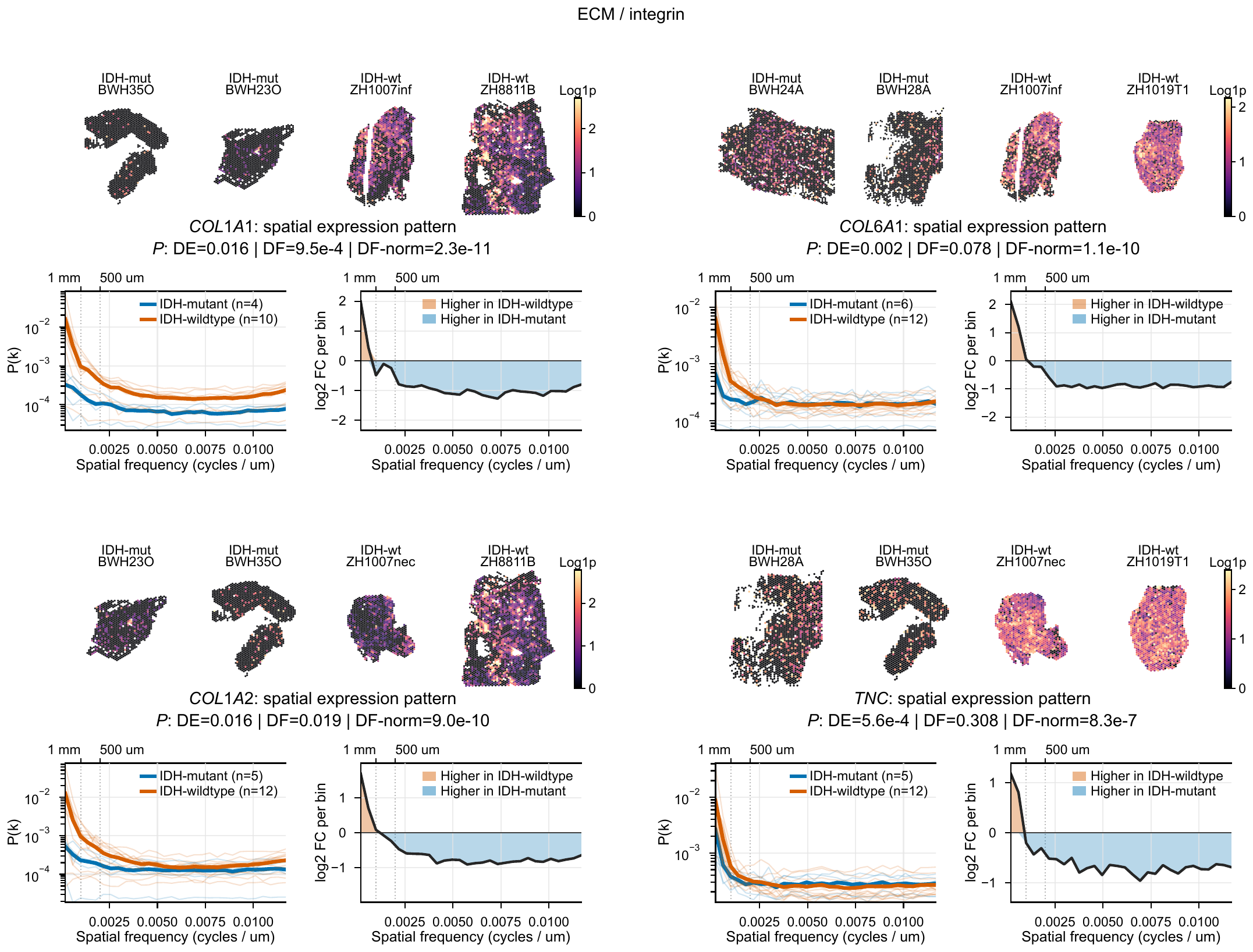}{\suppFigSixteenTitle}{\suppFigSixteenLegend}
\suppfigurepage{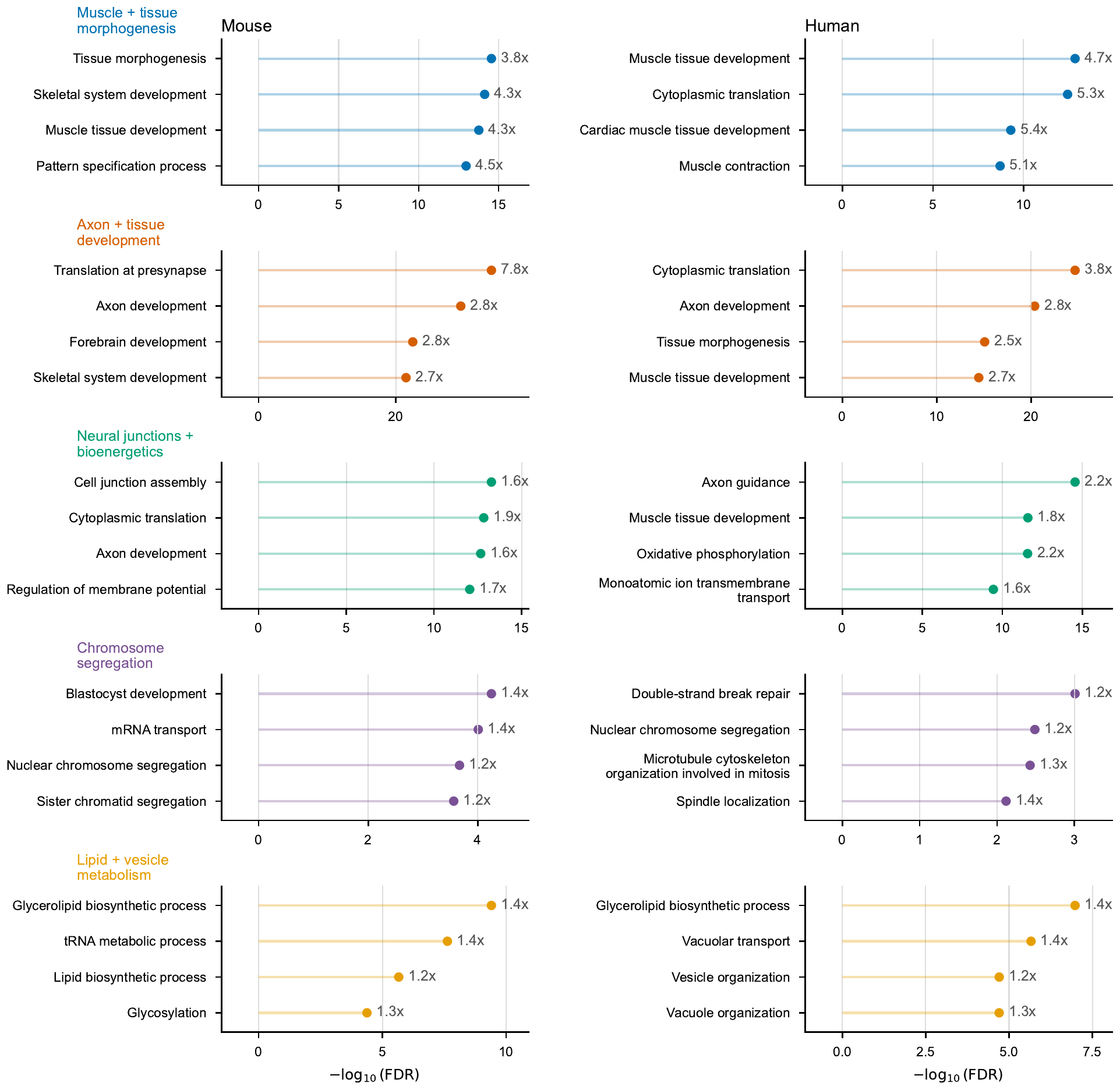}{\suppFigSeventeenTitle}{\suppFigSeventeenLegend}
\suppfigurepage{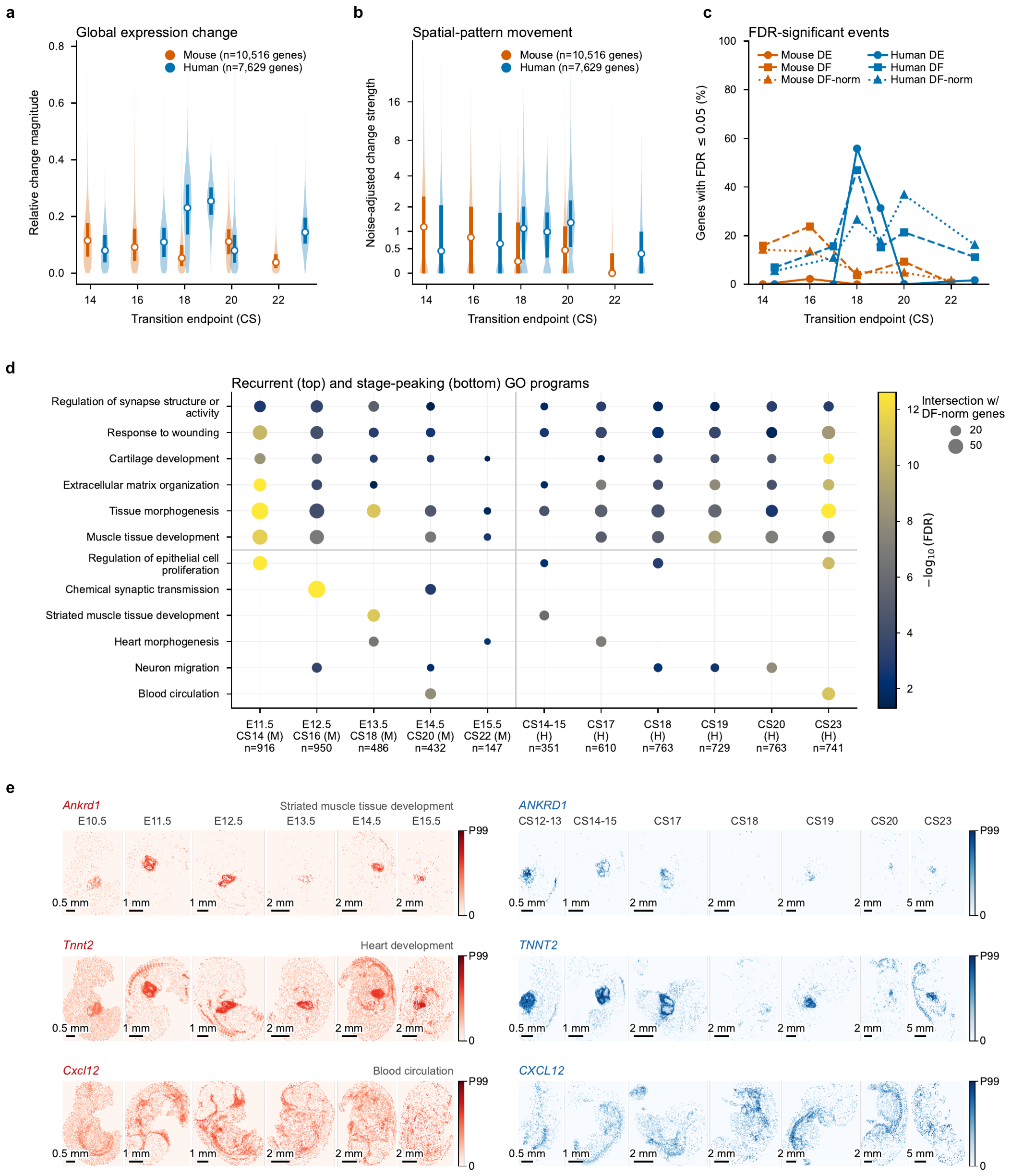}{\suppFigEighteenTitle}{\suppFigEighteenLegend}
\clearpage
\endgroup

\clearpage
\begin{refsection}

\appendix

\renewcommand{\thesection}{\arabic{section}}
\titleformat{\section}
  {\normalfont\large\bfseries} 
  {Supplementary Note~\thesection} 
  {1em} 
  {}

\section{Clarification on kernel terminology}\label{supnote:1}

The term ``kernel'' is notoriously overloaded in computational sciences, carrying distinct definitions across statistics, functional analysis, and machine learning. Here we distinguish kernel functions, their matrices and operators, and the associated function spaces.

\subsection{The kernel function and kernel matrix}\label{supidx:n1:kernel-function}
In its most fundamental form, a \textbf{kernel function} $k(s, s')$ is a bivariate function mapping pairs of inputs from a space $\mathcal{S}$ (e.g., two spatial locations) to a real number (e.g., distance between the two locations). We assume the kernel is symmetric, so that $k(s, s') = k(s', s)$. The \textbf{kernel matrix}, denoted as $\mathbf{K} \in \mathbb{R}^{n \times n}$, is the discrete evaluation of this function on a finite set of $n$ locations $\{s_1, \dots, s_n\}$, where $\mathbf{K}_{ij} = k(s_i, s_j)$ (e.g., $n^2$ pairwise distances).

In machine learning contexts (e.g., Gaussian Processes), kernels are typically required to be positive semidefinite (PSD). However, in the statistical theory of U-statistics and V-statistics, which underpins parts of this work, the kernel function $k$ \textit{need not} be PSD. In that context, it serves primarily as a symmetric weighting function quantifying the relationship between $s$ and $s'$.

\subsection{The kernel operator and its spectrum}\label{supidx:n1:kernel-operator}
We write $P_S$ for the distribution of sampled locations on $\mathcal S$. A symmetric, square-integrable kernel defines an operator on $L^2(P_S)$,
\[
    (\mathcal Kf)(s)
    =\int_{\mathcal S}k(s,s')f(s')\,dP_S(s').
\]
Thus, $\mathcal Kf$ averages the kernel-weighted values of $f$ over the sampling distribution.

For locations drawn independently from $P_S$, this integral is approximated by
\[
    (\mathcal Kf)(s)
    \approx \frac1n\sum_{j=1}^n k(s,s_j)f(s_j).
\]
At the observed locations, the corresponding matrix operation is $\mathbf K\mathbf x/n$, where $x_i=f(s_i)$. For uniform sampling on a region of volume $|\mathcal S|$, $dP_S(s)=ds/|\mathcal S|$. An integral with respect to unnormalized volume therefore corresponds to $|\mathcal S|\mathbf K\mathbf x/n$.

\textbf{The kernel spectrum} refers to the eigendecomposition of the kernel operator or matrix:
\begin{itemize}
\item (Continuous) The eigenvalues $\lambda_j$ and eigenfunctions $\phi_j$ of $\mathcal{K}$ satisfy $\int_{\mathcal S} k(s,s')\phi_j(s')\,dP_S(s')=\lambda_j\phi_j(s)$.
\item (Discrete) The eigenvalues $\lambda_j^{(n)}$ and eigenvectors $\mathbf{v}_j$ of $\mathbf{K}$ satisfy $\mathbf{K}\mathbf{v}_j = \lambda_j^{(n)} \mathbf{v}_j$.
\end{itemize}
Under suitable regularity conditions, the matrix spectrum approximates the operator spectrum after normalization: $\lambda_j^{(n)}/n\to\lambda_j$.

\textbf{Remark.}
Not all limits of kernel matrices result in a \textit{global} integral kernel operator defined across $\mathcal{S}$. A prime example is the Graph Laplacian. Consider a kernel defined by proximity $\mathbf{W}_{ij} = \mathbb{I}(\|s_i - s_j\| < r)$ where $r$ is a given distance cutoff. The nature of the limit operator depends entirely on the behavior of the radius $r$ relative to the sample size $n$.

\begin{itemize}
\item If $r$ remains fixed as $n\to\infty$, $\mathbf W/n$ approximates the neighborhood integral operator
\[
    (\mathcal K_r f)(s)
    =\int_{B_r(s)}f(s')\,dP_S(s').
\]
Its value depends on both the local pattern and the sampling density.

\item For uniform sampling and a smooth function away from the boundary, the normalized neighborhood average satisfies \[\lim_{r \to 0} \frac{1}{r^2} \left( f(x) - \frac{\int_{B_r(x)} f(y) \, dy}{\text{Vol}(B_r(x))} \right) = - C_d \Delta f(x),\] where $C_d$ is a dimension-dependent constant. Here, the limit is not an integral operator, but the second-order \textit{local} differential operator $\Delta f = \sum_i \partial^2 f / \partial x_i^2$.

\end{itemize}
While $\Delta$ is not an integral operator, spectral theory still applies. A self-adjoint Laplacian on a bounded domain with suitable boundary conditions and compact resolvent has a discrete eigenbasis, although its eigenvalues need not be bounded.

\subsection{Reproducing kernel Hilbert space (RKHS)}\label{supidx:n1:rkhs}
When a kernel function $k$ is symmetric and PSD, it uniquely determines a function space known as a \textbf{Reproducing Kernel Hilbert Space (RKHS)}, denoted by $\mathcal{H}_k$. An RKHS can be formulated using many equivalent definitions. One standard way is as a Hilbert space of functions where the evaluation functional $f \mapsto f(s)$ is bounded for all $s \in \mathcal{S}$. By the Riesz Representation Theorem, the boundedness implies the existence of a representer $k(\cdot, s) \in \mathcal{H}_k$ satisfying the \textit{reproducing property}
\[
f(s) = \langle f, k(\cdot, s) \rangle_{\mathcal{H}_k}, \quad \forall f \in \mathcal{H}_k, \, \forall s \in \mathcal{S}.
\]
This property allows us to analyze complex, potentially infinite-dimensional functions $f \in \mathcal{H}_k$ implicitly via the inner product and evaluations of the kernel on finite data.

For a continuous PSD kernel on a compact domain, Mercer's theorem gives the expansion
\[
    k(s,s')=\sum_j\lambda_j\phi_j(s)\phi_j(s'),
\]
on the support of $P_S$, where the eigenfunctions are orthonormal in $L^2(P_S)$. A corresponding feature map has coordinates $\sqrt{\lambda_j}\phi_j(s)$.

Another central concept in kernel hypothesis testing is the \textbf{kernel mean embedding}. We can represent a probability distribution $\mathbb{P}$ over $\mathcal{S}$ as an element $\mu_{\mathbb{P}} \in \mathcal{H}_k$, defined such that $\mathbb{E}_{S \sim \mathbb{P}}[f(S)] = \langle f, \mu_{\mathbb{P}} \rangle_{\mathcal{H}_k}$. Explicitly, this embedding is the expectation of the feature map
\[
\mu_{\mathbb{P}} = \mathbb{E}_{S \sim \mathbb{P}}[k(\cdot, S)].
\]
Given a finite sample $\{s_i\}_{i=1}^n$, we estimate this embedding empirically as $\hat{\mu}_{\mathbb{P}} = \frac{1}{n}\sum_{i=1}^n k(\cdot, s_i)$. HSIC compares the embedding of the joint distribution with the product of the marginal embeddings. Its relationship to mutual information is discussed in Supplementary Note~\ref{supnote:2.mi_approx}.

\subsection{Kernels in convolutional neural networks (CNNs)}\label{supnote:1.cnn}
In image processing, a ``kernel'' typically refers to a discrete convolution filter, $\mathbf{A}$, consisting of a small array of learnable weights. This filter slides across an input image $\mathbf{X}$ with $n$ pixels to extract local features. Mathematically, the operation can be expressed as a matrix-vector product $\mathbf{K}_{\text{conv}} \cdot \text{vec}(\mathbf{X})$, where $\mathbf{K}_{\text{conv}} \in \mathbb{R}^{n \times n}$ is a doubly block Toeplitz matrix generated by $\mathbf{A}$. Although $\mathbf{K}_{\text{conv}}$ is rarely constructed explicitly due to its size, its spectrum is fully determined by the filter $\mathbf{A}$. Under periodic boundary conditions (where $\mathbf{K}_{\text{conv}}$ becomes block circulant), its eigenvalues correspond to the 2D Fourier transform coefficients of the zero-padded filter $\tilde{\mathbf{A}}$ with the same dimensions as $\mathbf{X}$,
\[
\text{eig}(\mathbf{K}_{\text{conv}}) = \mathcal{F}(\tilde{\mathbf{A}}).
\]
This connection allows us to analyze CNN filters using the same spectral tools applied to other pattern detection methods, with two key distinctions. First, $\mathbf{K}_{\text{conv}}$ is symmetric only if $\mathbf{A}$ is point symmetric (i.e., $A[i, j] = A[-i, -j]$). Second, $\mathbf K_{\mathrm{conv}}$ is PSD exactly when its Fourier eigenvalues are nonnegative. An autocorrelation filter $\mathbf A=\mathbf B*\operatorname{flip}(\mathbf B)$ is one construction. 

When these conditions hold, a fixed convolution kernel defines a global (image-level) quadratic statistic $Q=\text{vec}(\mathbf{X})^\top \mathbf{K}_{\text{conv}} \text{vec}(\mathbf{X})$. However, testing with a learned filter requires accounting for its selection (Supplementary Note~\ref{supnote:3.data_driven_kernel}).

\newpage

\section{Equivalence between spatial variability tests}\label{supnote:2}

In this section, we demonstrate that virtually all spatial variability tests—despite originating from distinct statistical traditions such as autocorrelation, mixed models, and kernel methods—can be connected through a shared quadratic statistic $Q_n$, which measures spatial variation in the mean across locations (Theorem~\ref{theorem:q_mean_dependence}).

\subsection{Problem formulation}
\label{supnote:2.formulation}

We begin by formalizing spatial variability through the lens of statistical dependency.

\begin{definition}[Spatial Variability]\label{def:sv}
    Let $X$ be a random variable representing a feature of interest (e.g., gene expression) and $S$ be a random variable representing spatial coordinates. We say $X$ exhibits spatial variability if $X$ and $S$ are statistically dependent. That is, the joint distribution differs from the product of its marginals,
    \[
    P_{SX}\ne P_S\otimes P_X.
    \]
\end{definition}

Unless stated otherwise, the sampling-based results below use the following data-generating model: We first draw locations $S_1,\ldots,S_n$ independently from a spatial distribution $P_S$ on $\mathcal{S}$, then independently draw each expression value $X_i$ from $P_{X\mid S=S_i}$. The observed pairs $(S_i,X_i)$ are therefore independent and identically distributed draws from $P_{SX}$.
For fixed locations, such as a predefined measurement grid, only expression is sampled, where each $X_i$ is drawn independently from the distribution at location $s_i$. Under the null, this distribution is the same across all locations.

We derive null moments for the observed locations and a prespecified spatial kernel, treating both as fixed. Population and V-statistic limits use the random-location model. Model-specific departures from these assumptions are stated where needed.

\begin{proposition}\label{prop:iid}
    Under $X\perp S$, the expression values $X_1,\ldots,X_n$ are i.i.d. from $P_X$ conditional on the observed locations.
\end{proposition}

\textbf{Remark.}
In the random-location formulation, the $X_i$ are marginally i.i.d. even under spatial alternatives. The null additionally requires their conditional distributions to be identical across locations. It is often difficult to estimate these conditionals because each location is only observed once in spatial omics data.

Let $\mathbf{s}=\{s_1, \dots, s_n\} \in \mathcal{S}^n$ be the spatial coordinates of $n$ samples, and $\mathbf{x} =[x_1, \dots, x_n]^\top \in \mathcal{X}^n$ be the vector of observations (e.g., the expression of a single gene) at the $n$ sampled locations. We assume that $X$ takes a scalar value and $\mathcal{X} \subseteq\mathbb{R}$. The multivariate case (e.g., isoform usage ratios within a gene)\autocite{su2026mapping} is discussed in Supplementary Note~\ref{supidx:n2:multivariate}. 

Under the null, observations at all locations share a common mean $\mu=\mathbb{E}[X]$ and variance $0<\sigma^2=\operatorname{Var}(X)<\infty$. Standardizing expression with these population quantities gives
\[
    \mathbf{z}=\frac{\mathbf{x}-\mu\mathbf{1}}{\sigma},
    \qquad
    \mathbb{E}[\mathbf{z}]=\mathbf{0},
    \qquad
    \operatorname{Var}(\mathbf{z})=\mathbf{I}.
\]
We first derive the null moments using this population standardization.

\begin{definition}[\(Q\)-statistic]
    Given standardized data $\mathbf{z} \in\mathbb{R}^n$ and a symmetric kernel matrix $\mathbf{K} \in \mathbb{R}^{n \times n}$, we define the \(Q\)-statistic as the quadratic form
    \[
    Q_n = \mathbf{z}^\top \mathbf{K} \mathbf{z}.
    \]
\end{definition}

\textbf{Remark (empirical standardization).}
In practice, we use $\mathbf{z}=\mathbf{H}\mathbf{x}/\hat\sigma$, where $\mathbf{H}=\mathbf{I}-\mathbf{1}\mathbf{1}^{\top}/n$ and $\hat\sigma^2=\mathbf{x}^{\top}\mathbf{H}\mathbf{x}/(n-1)$. This imposes $\mathbf{1}^{\top}\mathbf{z}=0$ and $\|\mathbf{z}\|^2=n-1$, so the standardized entries are no longer independent. Writing $\tilde{\mathbf{K}}=\mathbf{H}\mathbf{K}\mathbf{H}$ gives
\[
    Q_n
    =\mathbf{z}^{\top}\mathbf{K}\mathbf{z}
    =\frac{\mathbf{x}^{\top}\tilde{\mathbf{K}}\mathbf{x}}
    {\hat\sigma^2}.
\]
Kernel centering accounts for the fitted mean (Proposition~\ref{prop:spectrum_centered}), while the empirical variance makes $Q_n$ a ratio of quadratic forms. Its exact Gaussian null law and the associated Dirichlet variance correction are derived in Supplementary Note~\ref{supnote:eng}.

Even though different SVG tests were developed for distinct alternative hypotheses, these connections allow us to systematically compare them, examine test consistency under Definition~\ref{def:sv} and improve test design (Supplementary Note~\ref{supnote:3}).

\subsection{Asymptotic behavior of the quadratic form \(Q\)-statistic}
\label{supnote:2.asymptotic}

For population-standardized observations with a finite fourth moment, the null moments conditional on the locations are
\begin{align}\label{eq:null_q}
    \mathbb{E}[Q_n] &= \mathbb{E}[\text{tr}(\mathbf{z}^\top \mathbf{K} \mathbf{z})] = \text{tr}(\mathbf{K} \mathbb{E}[\mathbf{z}\mathbf{z}^\top]) = \text{tr}(\mathbf{K})\\
    \text{Var}(Q_n) &= \sum_i\sum_j\sum_k\sum_l K_{ij}K_{kl}\mathbb{E}[z_iz_jz_kz_l]-(\text{tr}(\mathbf{K}))^2=(\mu_4 - 3)\sum_{i=1}^n K_{ii}^2 + 2\text{tr}(\mathbf{K}^2), \nonumber
\end{align}
where $\mu_4 = \mathbb{E}[z_i^4]$ is the kurtosis of the distribution of $Z$. If the data are Gaussian, $\mu_4=3$, simplifying the variance to $\text{Var}(Q_n)=2\text{tr}(\mathbf{K}^2)$. Note that the size of $\mathbf{K}$ and thus its trace will vary based on sample size $n$. For simplicity, we write $\mathbf K$ for $\mathbf K_n$ when the sample size is clear. In hypothesis testing, we assume $\mathbf{K}$ is generated from spatial coordinates via a symmetric function, such that $\mathbf{K}_{ij}=k(s_i, s_j)$. Defining $y_i=(z_i, s_i)$, we can rewrite the statistic as
\[
Q_n = \sum_{i=1}^n \sum_{j=1}^n z_i z_j k(s_i, s_j) \eqqcolon \sum_{i=1}^n \sum_{j=1}^n H(y_i, y_j).
\]
Under the null, $Y=(Z,S)$ has distribution $P_Y=P_Z\otimes P_S$, and $H((z,s),(z',s'))=zz'k(s,s')$ is a degenerate kernel. Its integral operator is
\[
    (\mathcal Hf)(y)=\int H(y,y')f(y')\,dP_Y(y').
\]
Because $Z$ and $S$ are independent, $\mathcal H$ is the tensor product of the spatial operator $\mathcal K$ and the linear data operator
\[
    (\mathcal Lu)(z)=\int zz'u(z')\,dP_Z(z').
\]
The latter has a single nonzero eigenvalue, $\mathbb E[Z^2]=1$, so the nonzero eigenvalues of $\mathcal H$ are those of $\mathcal K$.

For square-integrable $k$ with $\int |k(s,s)|\,dP_S(s)<\infty$, the degenerate V-statistic limit is\autocite{mises1947asymptotic,serfling2009approximation}
\[
    \frac{Q_n}{n}
    \xrightarrow{d}
    \int k(s,s)\,dP_S(s)
    +\sum_j\lambda_j(\xi_j^2-1),
    \qquad
    \xi_j\stackrel{\mathrm{iid}}{\sim}\mathcal N(0,1).
\]
When the eigenvalues are absolutely summable and $\sum_j\lambda_j=\int k(s,s)\,dP_S(s)$, this simplifies to
\begin{equation}\label{eq:v_stats_null}
    \frac{Q_n}{n}\xrightarrow{d}\sum_j\lambda_j\xi_j^2.
\end{equation}
These conditions hold, for example, for bounded continuous PSD kernels on a compact domain.

The matrix eigenvalues $\lambda_j^{(n)}/n$ approximate the operator eigenvalues $\lambda_j$. They therefore provide finite-sample approximations to the spectral null distribution. Moment-based alternatives are described in Supplementary Note~\ref{supnote:eng}. We refer to the testing procedure using the statistic \(Q\) and the null moments in \eqref{eq:null_q} or spectral null calibration in \eqref{eq:v_stats_null} (or other null approximations) as a \textbf{\(Q\)-test}.

\textbf{Remark.} With empirical standardization, the asymptotic null uses the spatial kernel centered with respect to $P_S$, estimated by $\mathbf H\mathbf K\mathbf H$. The sample-variance denominator does not change this limit because $\hat\sigma^2/\sigma^2$ converges to one.

What these asymptotic null calibrations do not tell us, however, is the behavior of \(Q\) under the alternative hypothesis that $X$ is spatially variable. The consistency (i.e., whether the test almost always rejects the null under \textit{any} alternative hypothesis as sample size increases) and statistical power of \(Q\)-tests depend critically on the spectral design of $\mathbf{K}$.

Table~\ref{tab:sv_tests} summarizes the quadratic forms underlying the methods discussed below, with their assumptions and approximations.

\begin{table}[H]
\centering
\small
\captionsetup{font=small,labelfont=bf,justification=raggedright,singlelinecheck=false,skip=7pt}
\caption{Quadratic-form representations of spatial variability tests, $Q_n=\mathbf{z}^{\top}\mathbf{K}\mathbf{z}$.}
\label{tab:sv_tests}
\setlength{\tabcolsep}{5pt}
\renewcommand{\arraystretch}{1.12}
\begin{tabular}{@{}>{\raggedright\arraybackslash}p{0.35\linewidth}>{\raggedright\arraybackslash}p{0.27\linewidth}>{\raggedright\arraybackslash}p{\dimexpr0.38\linewidth-4\tabcolsep\relax}@{}}
\toprule
\textbf{Method or model} & \textbf{Alternative tested} & \textbf{Effective kernel} $\mathbf{K}$ \\
\midrule
Adjacency- and distance-based statistics\autocite{moran1948interpretation,detomaso2021hotspot,miller2021characterizing,dries2021giotto,edsgard2018identification}
& Non-trivial spatial autocorrelation or association
& $\mathbf{W}$ or a statistic-specific transform\textsuperscript{a} \\
\addlinespace[5pt]
Linear fixed-effects model
& $\boldsymbol{\beta}\neq\mathbf{0}$
& $\mathbf{K}_{\mathrm{proj}}$ (projection) \\
\addlinespace[5pt]
Linear mixed model (LMM)\autocite{svensson2018spatialde,hao2021somde,weber2023nnsvg}
& $\theta>0$ (spatial variance component)
& $\mathbf{K}_{\mathrm{cov}}$ (spatial covariance) \\
\addlinespace[5pt]
Generalized linear model (GLM)\autocite{cable2022cell}
& $\boldsymbol{\beta}\neq\mathbf{0}$
& $\mathbf{K}_{\mathrm{proj}}$ (projection)\textsuperscript{b} \\
\addlinespace[5pt]
Generalized linear mixed model (GLMM)\autocite{sun2020statistical,bintayyash2021non}
& $\theta>0$ (spatial variance component)
& $\mathbf{K}_{\mathrm{cov}}$ (spatial covariance)\textsuperscript{b} \\
\addlinespace[5pt]
HSIC with a linear data kernel\autocite{zhu2021spark,su2026mapping}
& $\mathbb{E}[X\mid S]\neq\mathbb{E}[X]$
& $\mathbf{K}$ \\
\addlinespace[5pt]
KL divergence (quadratic approximation)\autocite{vandenbon2020clustering}
& Gene-detection probability varies with location
& $\mathbf{K}_{\mathrm{dens}}$ (from density estimation) \\
\addlinespace[5pt]
Cell-type-specific spatial variability\autocite{cable2022cell,yu2022ctsv,yu2024spvc,su2025stance,shang2025celina,wang2025mmm}
& Nonzero spatial effects within cell types
& Weighted spatial kernel\textsuperscript{c}\\
\bottomrule
\end{tabular}

\vspace{4pt}
\begin{minipage}{\linewidth}
\footnotesize\raggedright
\setlength{\parskip}{3pt}
Representative formulations are shown; data transformation and approximation details differ across methods.

\textsuperscript{a}Moran's \(I\) uses spatial weights; the \texttt{BinSpect} table determinant uses a modularity transform of the adjacency matrix; \texttt{trendsceek} uses distance-dependent kernels that vary by statistic.

\textsuperscript{b}The standardized data $\mathbf{z}$ depend on the assumed mean--variance relationship and fitted null parameters. GLMMs additionally require a model-dependent approximation.

\textsuperscript{c}For the illustrative squared-quasiscore construction, $\mathbf{K}=\mathbf{V}^{-1/2}(\mathbf{A}\mathbf{A}^{\top}\circ\mathbf{K}')\mathbf{V}^{-1/2}$, where $\mathbf{K}'$ is the spatial kernel and $A_{ik}=m_iw_{ik}$ is the expected expression contributed by cell type $k$ at spot $i$ under the null. Here, $m_i$ is the total expected expression, $w_{ik}$ is that type's share, and $\mathbf{V}$ contains the null observation variances. The standardized data are the Pearson residuals $\mathbf{z}=\mathbf{V}^{-1/2}(\mathbf{x}-\mathbf{m})$.

\end{minipage}
\end{table}

\subsection{Moran's \(I\) as a \(Q\)-test}\label{supidx:n2:moran}
Moran's \(I\), initially introduced by \textcite{moran1948interpretation} and popularized by \textcite{cliff1981spatial}, remains the canonical and arguably most important measure of spatial autocorrelation. Given observations $\mathbf{z} \in \mathbb{R}^n$, it is defined as
\[
I = \frac{n}{\sum_{i,j} W_{ij}} \frac{\mathbf{z}^\top \mathbf{H} \mathbf{W} \mathbf{H} \mathbf{z}}{\mathbf{z}^\top \mathbf{H} \mathbf{z}},
\]
where $\mathbf{W}$ is the spatial weight matrix and $\mathbf{H} = \mathbf{I} - \frac{1}{n}\mathbf{1}\mathbf{1}^\top$ is the centering matrix. For standardized data, $\mathbf{z}^\top \mathbf{H} \mathbf{z} = n-1$ and $\mathbf{H}\mathbf{z} = \mathbf{z}$. Thus, we can reformulate Moran's \(I\) as a quadratic form governed by the weight matrix
\begin{equation}
    I = \frac{n}{(n-1)\sum_{i,j} W_{ij}} \mathbf{z}^\top \mathbf{W} \mathbf{z} \propto Q_{\text{Moran}} := \mathbf{z}^\top \mathbf{W} \mathbf{z}.
\end{equation}
Here, we assume $\mathbf{W}$ is symmetric (e.g., from mutual $k$-nearest-neighbor graphs). Moran's \(I\) can also be defined for asymmetric $\mathbf{W}'$, which is technically not a kernel matrix by Supplementary Note~\ref{supnote:1}. In practice, it is common to symmetrize the weight matrix via $\mathbf{W}=\frac{1}{2}(\mathbf{W}'+\mathbf{W}'^\top)$ which leaves Moran's \(I\) unchanged.
Many variants of Moran's \(I\) have been developed for SVG detection, such as \texttt{Hotspot} \autocite{detomaso2021hotspot} and \texttt{MERINGUE} \autocite{miller2021characterizing}. These methods typically differ in their strategies for constructing $\mathbf{W}$. If the corresponding centered kernel has eigenvalues of both signs, its quadratic form permits spectral cancellation, a problematic phenomenon we formally introduce in Supplementary Note~\ref{supnote:3}.

Moran's \(I\) also has a parametric score-test interpretation. \textcite{burridge1980cliff} first showed that the test based on Moran's \(I\) is equivalent to the score test for a Simultaneous Autoregressive (SAR) spatial random process model. The same quadratic form also arises from mixed models with more general covariance structures.

\subsection{Linear model tests as \(Q\)-tests}\label{supidx:n2:linear-models}
\subsubsection{Linear fixed-effects models}
Consider a linear model in which spatial variation is described by fixed features,
\[
    \mathbf{x}=b\mathbf{1}+\mathbf{L}\boldsymbol{\beta}
    +\boldsymbol{\epsilon},
    \qquad
    \boldsymbol{\epsilon}\sim
    \mathcal{N}(\mathbf{0},\sigma^2\mathbf{I}).
\]
Here, $b$ is the intercept and $\mathbf{L}$ contains $d$ centered, linearly independent spatial features, such as splines, distances to landmarks or cell-type abundance maps. We test the null hypothesis of no spatial effect, $H_0:\boldsymbol{\beta}=\mathbf{0}$.

Under the null, the fitted mean is $\bar{x}$. For known observation variance $\sigma^2$, the score vector and Fisher information are
\[
    \mathbf{U}_0
    =\frac{\mathbf{L}^{\top}(\mathbf{x}-\bar{x}\mathbf{1})}{\sigma^2},
    \qquad
    \mathcal{I}_0=\frac{\mathbf{L}^{\top}\mathbf{L}}{\sigma^2}.
\]
The resulting score statistic is
\[
    \mathbf{U}_0^{\top}\mathcal{I}_0^{-1}\mathbf{U}_0
    =
    \frac{
        (\mathbf{x}-\bar{x}\mathbf{1})^{\top}
        \mathbf{K}_{\mathrm{proj}}
        (\mathbf{x}-\bar{x}\mathbf{1})
    }{\sigma^2},
    \qquad
    \mathbf{K}_{\mathrm{proj}}
    =\mathbf{L}(\mathbf{L}^{\top}\mathbf{L})^{-1}\mathbf{L}^{\top}.
\]

Replacing $\sigma^2$ with the empirical null variance
$\hat\sigma^2=\|\mathbf{x}-\bar{x}\mathbf{1}\|^2/(n-1)$ gives
\begin{equation}\label{eq:score_fixed}
    S_{\mathrm{fixed}}
    =\mathbf{z}^{\top}\mathbf{K}_{\mathrm{proj}}\mathbf{z},
    \qquad
    \mathbf{z}=\frac{\mathbf{x}-\bar{x}\mathbf{1}}{\hat\sigma}.
\end{equation}
Thus, the linear test is a \(Q\)-test whose kernel projects expression onto the selected spatial features. Each included mode receives weight one, while orthogonal modes receive zero weight. With known variance, the Gaussian score statistic follows $\chi_d^2$ exactly under the null.

\subsubsection{Linear mixed models (LMMs)}\label{supnote:2.lmm}
Following \texttt{SpatialDE}\autocite{svensson2018spatialde}, consider a linear mixed model in which spatial structure enters through a covariance component,
\begin{equation}\label{eq:lmm}
    \mathbf{x}\sim
    \mathcal{N}\!\left(
        b\mathbf{1},
        \sigma^2\mathbf{I}+\theta\mathbf{K}_{\mathrm{cov}}
    \right).
\end{equation}
Here, $\mathbf{K}_{\mathrm{cov}}$ is a fixed positive semidefinite spatial covariance matrix, $\sigma^2$ is the independent noise variance and $\theta\geq0$ controls the spatial component. We test $H_0:\theta=0$ against $H_1:\theta>0$.

To account for the fitted intercept, we use residual likelihood, which evaluates covariance on the $n-1$ dimensions orthogonal to the constant vector. Under the null, the fitted variance is the empirical variance $\hat\sigma^2$ defined above. Differentiating the residual log-likelihood at $\theta=0$, with this fitted value held fixed, gives
\begin{equation}\label{eq:score_lmm}
    U(0)
    =\frac{
        Q_{\mathrm{LMM}}-\operatorname{tr}(\mathbf{H}\mathbf{K}_{\mathrm{cov}}\mathbf{H})
    }{2\hat\sigma^2},
    \qquad
    Q_{\mathrm{LMM}}
    =\mathbf{z}^{\top}\mathbf{H}\mathbf{K}_{\mathrm{cov}}\mathbf{H}\mathbf{z}
    =\mathbf{z}^{\top}\mathbf{K}_{\mathrm{cov}}\mathbf{z}.
\end{equation}
For fixed null parameters, the score increases with $Q_{\mathrm{LMM}}$, so an upper-tail score test corresponds to an upper-tail \(Q\)-test.

The spatial component has a Gaussian prior with covariance $\theta\mathbf{K}_{\mathrm{cov}}$. In the eigenbasis of the centered kernel,
\[
    Q_{\mathrm{LMM}}
    =\sum_j\lambda_j(\mathbf{u}_j^{\top}\mathbf{z})^2.
\]
Modes assigned larger prior variance receive greater weight in the test. Kernel choice therefore encodes prior expectations about spatial patterns, giving the \(Q\)-statistic and consequently pattern detection a Bayesian interpretation.

\textbf{Remark.} \texttt{SpatialDE}\autocite{svensson2018spatialde} uses a likelihood ratio (LR) test with a $\chi_1^2$ null for hypothesis testing. Under suitable single-variance-component boundary regularity conditions and matching calibration, LR and score tests based on \eqref{eq:score_lmm} are asymptotically equivalent under the null and local alternatives (e.g., $\theta\approx0$). A theoretical nuance arises because the null hypothesis $\theta=0$ lies on the boundary of the parameter space ($\theta\geq0$). Under these conditions, with a fixed kernel and identifiable nuisance parameters, the asymptotic distribution of the LR statistic is not the standard $\chi_1^2$, but rather a mixture $0.5\chi_0^2+0.5\chi_1^2$, a result established by \textcite{chernoff1954distribution}. Furthermore, LMM (and its generalized version) is also called Gaussian Process (GP) regression, and more scalable implementations have since been developed, such as \texttt{SOMDE}\autocite{hao2021somde} and \texttt{nnSVG}\autocite{weber2023nnsvg}, both of which also use LR for hypothesis testing.

Finally, different model parameterizations of spatial covariance, as in \eqref{eq:lmm}, can yield identical quadratic score statistics at the null. The simultaneous autoregressive (SAR) model studied by \textcite{burridge1980cliff} provides a related example,
\[
    \mathbf{x}=b\mathbf{1}+\mathbf{u},
    \qquad
    \mathbf{u}=\rho\mathbf{W}\mathbf{u}+\boldsymbol{\epsilon},
    \qquad
    \boldsymbol{\epsilon}\sim\mathcal{N}(\mathbf{0},\sigma^2\mathbf{I}).
\]
Here, $\mathbf{W}$ is a nonnegative row-normalized adjacency matrix with zero diagonal, and $0\leq\rho<1$ is the spatial autoregressive parameter. The resulting covariance is
\[
    \operatorname{Var}(\mathbf{u})
    =\sigma^2
    \left[
        (\mathbf{I}-\rho\mathbf{W})^{\top}
        (\mathbf{I}-\rho\mathbf{W})
    \right]^{-1}.
\]
At $\rho=0$, its derivative is $\sigma^2(\mathbf{W}+\mathbf{W}^{\top})$. Using the same residual-likelihood treatment of the intercept gives
\[
    U_\rho(0)
    =\mathbf{z}^{\top}\mathbf{W}\mathbf{z}
    -\operatorname{tr}(\mathbf{H}\mathbf{W}\mathbf{H}).
\]
The quadratic term recovers Moran's numerator, while the trace correction depends only on the spatial weights. With corresponding calibration, this gives an upper-tail test based on Moran's \(I\).

\subsection{Generalized linear model tests as \(Q\)-tests}\label{supidx:n2:glm-tests}
Generalized linear models (GLMs) accommodate non-Gaussian observations through a link function and a specified mean--variance relationship. These assumptions can improve power when they capture the data-generating process. For spatial variability testing, however, the GLM score test has the same projected quadratic form as its linear-model counterpart under the simplifying assumptions below. The difference lies in standardization: the GLM uses the model-based null variance, whereas the linear test uses the empirical variance. Their ratio approaches one under a correctly specified null but can differ when the empirical variance includes spatial signal.

\subsubsection{Generalized linear models (GLMs)}
\label{supnote:2.glm}
Consider a GLM where the independent observations $\mathbf{x}$ follow a distribution from the exponential family, such as Poisson or negative binomial. A link function $g(\cdot)$ relates the mean $\boldsymbol{\mu}=\mathbb{E}[\mathbf{x}]$ to fixed spatial effects,
\[
    g(\boldsymbol{\mu})=b\mathbf{1}+\mathbf{L}\boldsymbol{\beta},
\]
where $b$ is the intercept and $\mathbf{L}$ contains $d$ spatial features. We take these features to be centered and linearly independent. We test the null hypothesis of no spatial effect, $H_0:\boldsymbol{\beta}=\mathbf{0}$.

For clarity, assume equal sampling depth and a shared mean--variance relationship across locations. Under the null, the fitted mean is constant, $\hat\mu_0=\bar{x}$, as is the fitted observation variance $V_0$, including any dispersion parameter. The score vector and Fisher information for the spatial coefficients are
\[
    \mathbf{U}_0=
    \frac{\mathbf{L}^{\top}(\mathbf{x}-\hat\mu_0\mathbf{1})}
    {V_0g'(\hat\mu_0)},
    \qquad
    \mathcal{I}_0=
    \frac{\mathbf{L}^{\top}\mathbf{L}}
    {V_0[g'(\hat\mu_0)]^2}.
\]
The link derivative cancels when we form the score statistic. Defining standardized residuals
$\mathbf{z}_0=(\mathbf{x}-\hat\mu_0\mathbf{1})/\sqrt{V_0}$ gives
\begin{equation}\label{eq:score_glm}
    S_{\mathrm{GLM}}
    =\mathbf{U}_0^{\top}\mathcal{I}_0^{-1}\mathbf{U}_0
    =\mathbf{z}_0^{\top}\mathbf{K}_{\mathrm{proj}}\mathbf{z}_0
    =\frac{\hat\sigma^2}{V_0}S_{\mathrm{fixed}},
\end{equation}
where $\mathbf{K}_{\mathrm{proj}}=\mathbf{L}(\mathbf{L}^{\top}\mathbf{L})^{-1}\mathbf{L}^{\top}$ and $S_{\mathrm{fixed}}$ uses the same spatial features with residuals standardized by the empirical variance $\hat\sigma^2$, as in \eqref{eq:score_fixed}. Under a correctly specified null and the usual regularity conditions, $S_{\mathrm{GLM}}$ is asymptotically $\chi_d^2$.

The assumed mean--variance relationship also determines how sensitivity changes with abundance. To see this, consider a small relative change $a$ around a baseline expression level $\mu$. The mean shift is $\mu a$, and the corresponding squared signal-to-noise ratio (SNR) per observation is
\[
    \mathrm{SNR}^2=\frac{(\mu a)^2}{V(\mu)}.
\]

Holding the relative pattern and sampling design fixed gives:
\begin{itemize}
    \item \textbf{Gaussian.} Under linear-model standardization, $V_0=\hat\sigma^2$, recovering $S_{\mathrm{fixed}}$. For fixed observation variance $\sigma^2$,
    $\mathrm{SNR}^2=\mu^2a^2/\sigma^2$.
    \item \textbf{Poisson.} Here, $V_0=\hat\mu_0$ and $\mathrm{SNR}^2=\mu a^2$.
    \item \textbf{Negative binomial.} Here, $V_0=\hat\mu_0+\hat\mu_0^2/r$, where $r$ is the size parameter, and
    \[
        \mathrm{SNR}^2=\frac{\mu a^2}{1+\mu/r},
    \]
    which approaches $ra^2$ as abundance increases with $r$ fixed.
\end{itemize}

Compared with a Gaussian model with fixed variance, Poisson and negative-binomial models predict slower gains in squared SNR as abundance increases. However, these scaling results do not determine which test performs better on the same data. A Poisson model can underestimate noise in overdispersed counts and inflate false positives, whereas a larger negative-binomial variance reduces evidence for the same observed deviation. Valid inference therefore depends on whether the variance model fits the data.

In practice, sampling depth often varies across locations. This extends the analysis above by replacing the scalar $V_0$ with location-specific null variances and using a weighted projection. The score statistic retains its quadratic form, with standardization adjusted separately at each location. Additional null covariates are also accounted for in the projection. For example, \texttt{C-SIDE}\autocite{cable2022cell} fits a more detailed model of cell-type-specific expression and uses Wald tests based on estimated effects and their standard errors; we discuss this model in Supplementary Note~\ref{supnote:2.ctsv}.

\textbf{Remark (score, LR and Wald tests).} Score tests are attractive for variance-component testing because they require fitting only the null model once, even with high-dimensional covariates such as genotypes\autocite{wu2011rare}. Likelihood-ratio and Wald tests instead estimate coefficients \(\hat{\boldsymbol{\beta}}\) under the alternative. Those estimates are useful for interpreting effect magnitude and direction, but are unnecessary for detection. For regular fixed-effect tests, score, likelihood-ratio and Wald tests are asymptotically equivalent under local alternatives, so the score test captures the same local signal without fitting the full alternative model. Variance-component tests require the boundary qualifications discussed in Supplementary Note~\ref{supnote:2.lmm}. Agreement outside local alternatives, and the quality of \(\hat{\boldsymbol{\beta}}\), depend on the likelihood surface (which is usually well-behaved for common parametric families).

\subsubsection{Generalized linear mixed models (GLMMs)}
\label{supnote:2.glmm}
Spatial methods such as \texttt{SPARK}\autocite{sun2020statistical} and \texttt{GPcounts}\autocite{bintayyash2021non} model count data with spatial random effects. To connect such models to \(Q\)-tests, consider
\[
    g(\mathbb{E}[\mathbf{x}\mid\boldsymbol{\epsilon}])
    =b\mathbf{1}+\boldsymbol{\epsilon},
    \qquad
    \boldsymbol{\epsilon}\sim
    \mathcal{N}\!\left(\mathbf{0},\tau^2(\mathbf{I}+\theta\mathbf{K})\right).
\]
Here, $\tau^2>0$ controls independent variation on the link scale, while $\theta\geq0$ scales the spatial covariance component defined by $\mathbf{K}$.

Testing $H_0:\theta=0$ through the marginal likelihood requires integrating out $\boldsymbol{\epsilon}$. We use a Gaussian working approximation around zero random effect to make this calculation tractable. As in the GLM section, assume equal sampling depth and a common variance model, with fitted null parameters held fixed during the score calculation. Let $\hat\mu_0=g^{-1}(\hat b)$ denote the fitted mean at zero random effect, and let $V_0=V(\hat\mu_0)$ be the corresponding conditional observation variance, including dispersion. The working residual and its variance on the link scale are
\[
    \mathbf{r}_0
    =g'(\hat\mu_0)(\mathbf{x}-\hat\mu_0\mathbf{1}),
    \qquad
    V_{\mathrm{work},0}
    =V_0[g'(\hat\mu_0)]^2.
\]
The working model treats these residuals as the random effects plus Gaussian observation noise,
\[
    \mathbf{r}_0\mid\boldsymbol{\epsilon}
    \approx
    \mathcal{N}(\boldsymbol{\epsilon},V_{\mathrm{work},0}\mathbf{I}).
\]
Integrating out the random effects gives
\[
    \mathbf{r}_0\approx
    \mathcal{N}\!\left(
        \mathbf{0},
        (\tau^2+V_{\mathrm{work},0})\mathbf{I}
        +\tau^2\theta\mathbf{K}
    \right).
\]
Differentiating this Gaussian marginal log-likelihood at $\theta=0$ yields
\[
    U(0)\approx
    \frac{\tau^2\,\mathbf{r}_0^{\top}\mathbf{K}\mathbf{r}_0}
    {2(\tau^2+V_{\mathrm{work},0})^2}
    -
    \frac{\tau^2}{2(\tau^2+V_{\mathrm{work},0})}
    \operatorname{tr}(\mathbf{K}).
\]
Define
\[
    \kappa=\frac{\tau^2}{\tau^2+V_{\mathrm{work},0}},
    \qquad
    \tilde{\mathbf{z}}
    =\frac{\tau\mathbf{r}_0}{\tau^2+V_{\mathrm{work},0}}
    =\frac{\tau\sqrt{V_0}\,g'(\hat\mu_0)}
    {\tau^2+V_{\mathrm{work},0}}\mathbf{z}_0,
\]
where $\mathbf{z}_0=(\mathbf{x}-\hat\mu_0\mathbf{1})/\sqrt{V_0}$ uses the GLM standardization evaluated at the same baseline. The score then becomes
\begin{equation}\label{eq:score_glmm}
    U(0)\approx
    \frac12\left[
        Q_{\mathrm{GLMM}}-\kappa\operatorname{tr}(\mathbf{K})
    \right],
    \qquad
    Q_{\mathrm{GLMM}}
    =\tilde{\mathbf{z}}^{\top}\mathbf{K}\tilde{\mathbf{z}}.
\end{equation}
For fixed null parameters, $U(0)$ increases with $Q_{\mathrm{GLMM}}$, so an upper-tail score test is equivalent to an upper-tail \(Q\)-test within this approximation.

\textbf{Remark.}
Under the Gaussian working approximation with a constant baseline mean and variance across locations, the GLMM residual transformation is a scalar multiple of the GLM transformation. For nonlinear links, large random effects (for example, large $\tau$) can make the linear expansion around zero inaccurate. Unequal sampling depths or null covariates can also require location-specific weights, as in the GLM extension above.

\subsection{HSIC-based tests as \(Q\)-tests}\label{supidx:n2:hsic}
The HSIC representation connects quadratic spatial tests to kernel dependence testing\autocite{gretton2005measuring,gretton2007kernel}. With a linear data kernel, the population quantity measures spatial mean dependence, providing a basis for studying test consistency (Theorem~\ref{theorem:q_mean_dependence}).

\subsubsection{HSIC background}
To motivate the kernel perspective, consider testing linear dependence between expression $X$ and spatial coordinates $S\in\mathbb R^d$. Let $\mathbf S$ contain the centered coordinates and let $\mathbf z$ be standardized expression. The sum of squared sample covariances is
\[
    \frac{\|\mathbf S^\top\mathbf z\|^2}{(n-1)^2}
    =\frac{\mathbf z^\top\mathbf S\mathbf S^\top\mathbf z}{(n-1)^2}.
\]
It is therefore proportional to a \(Q\)-statistic with the linear spatial kernel $\mathbf K=\mathbf S\mathbf S^\top$.

\textcite{gretton2005measuring} generalize the notion of covariance to non-linear dependencies by mapping variables into RKHSs. Let $\mathcal{H}_X$ and $\mathcal{H}_S$ be RKHSs on $\mathcal{X}$ and $\mathcal{S}$ with kernels $l(\cdot, \cdot)$ and $k(\cdot, \cdot)$, respectively. The population HSIC is defined as the squared Hilbert-Schmidt norm of the cross-covariance operator $C_{SX}$,
\[
\text{HSIC}(P_{SX}, \mathcal{H}_S, \mathcal{H}_X) := \|C_{SX}\|_{\text{HS}}^2 = \left\| \mathbb{E}_{SX}[k(S, \cdot) \otimes l(X, \cdot)] - \mathbb{E}_S[k(S, \cdot)] \otimes \mathbb{E}_X[l(X, \cdot)] \right\|_{\mathcal{H}_S \otimes \mathcal{H}_X}^2.
\]

Intuitively, HSIC measures the distance between the kernel mean embedding of the joint distribution and the product of its marginals. If the product kernel distinguishes the joint distribution from the product of its marginals, HSIC is zero only under independence. Characteristic product kernels, including Gaussian kernels on Euclidean domains, provide this property.

\subsubsection{A \(Q\)-test is a kernel dependence test with a linear data kernel}

\texttt{SPARK-X} \autocite{zhu2021spark} and \texttt{SPLISOSM} \autocite{su2026mapping} use the empirical estimator of HSIC with a \textit{linear} expression kernel ($l(x, x') = xx'$) and a spatial kernel $k(s, s')$. We have previously shown that any low-rank spatial kernel, including \texttt{SPARK-X}, leads to substantial power loss \autocite{su2026mapping}.

The HSIC-based test statistic for SVG is defined as
\[
S_{\text{HSIC}} = \frac{1}{(n-1)^2} \text{tr}(\mathbf{K} \mathbf{H} \mathbf{L}_X \mathbf{H}),
\]
where $\mathbf{K}$ is the spatial kernel matrix $K_{ij} = k(s_i, s_j)$, $\mathbf{L}_X = \mathbf{x}\mathbf{x}^\top$ is the linear kernel matrix for expression, and $\mathbf{H} = \mathbf{I} - \frac{1}{n}\mathbf{1}\mathbf{1}^\top$ is the centering matrix. Substituting $\mathbf{L}_X$ into the equation reveals the equivalence to the \(Q\)-statistic as
\begin{equation}\label{eq:score_hsic}
S_{\text{HSIC}} \propto \text{tr}(\mathbf{K} \mathbf{H} \mathbf{x}\mathbf{x}^\top \mathbf{H}) = \mathbf{x}^\top \mathbf{H} \mathbf{K} \mathbf{H} \mathbf{x} \propto \mathbf{z}^\top \mathbf{K} \mathbf{z} = Q_{\text{HSIC}}.
\end{equation}

This gives a covariance interpretation of the \(Q\)-statistic alongside its prior-weighting interpretation in the LMM (Supplementary Note~\ref{supnote:2.lmm}). A finite PSD kernel matrix has a factorization $\mathbf K=\boldsymbol\Phi\boldsymbol\Phi^\top$, where $\boldsymbol\Phi\in\mathbb R^{n\times r}$ and $r=\operatorname{rank}(\mathbf K)$. Then $Q=\|\boldsymbol\Phi^\top\mathbf z\|^2$, which sums squared associations with the spatial features. For a single nonzero centered feature $\mathbf u$ and its projection kernel $\mathbf K=\mathbf u\mathbf u^\top/\|\mathbf u\|^2$, this reduces to $Q=(n-1)r_{\mathbf u,\mathbf z}^2$, where $r_{\mathbf u,\mathbf z}$ is Pearson correlation.

\subsubsection{Linear data kernels measure spatial mean dependence}

The kernel perspective allows us to study which forms of spatial dependence a \(Q\)-test captures. The empirical HSIC converges to its population value\autocite{gretton2007kernel}. With a linear data kernel, this population quantity measures changes in conditional mean, as the following result shows.

\begin{theorem}[\(Q\)-statistic measures mean dependence]
\label{theorem:q_mean_dependence}
Let $\mathcal{H}_S$ be the RKHS on the spatial domain $\mathcal{S}\subseteq\mathbb{R}^d$ associated with a fixed, bounded,
continuous characteristic kernel $k(s,s')$. Let $\mathcal{H}_X$ be the RKHS associated with the linear data kernel
$l(x,x')=xx'$, where $X$ is scalar and $0<\operatorname{Var}(X)<\infty$. Define the population \(Q\)-statistic as
\[
    Q(P_{SX})
    :=\operatorname{HSIC}(P_{SX},\mathcal{H}_S,\mathcal{H}_X)
    =\|C_{SX}\|_{\mathrm{HS}}^2.
\]
Then $Q(P_{SX})=0$ if and only if the conditional mean of $X$ does not vary with location:
\begin{equation}\label{eq:mean_ind}
    \mathbb{E}[X\mid S]=\mathbb{E}[X]
    \quad\text{almost surely}.
\end{equation}
This condition is called mean independence.
\end{theorem}

\begin{proof}
Because the data kernel is linear, $C_{SX}=0$ if and only if $\operatorname{Cov}(f(S),X)=0$ for every $f\in\mathcal{H}_S$. Writing $h(S)=\mathbb{E}[X\mid S]-\mathbb{E}[X]$, the law of total expectation gives
\begin{equation}\label{eq:th1_e1}
    \operatorname{Cov}(f(S),X)
    =\mathbb{E}[f(S)h(S)].
\end{equation}

To show that zero covariance forces $h=0$, separate the positive and negative mean deviations:
\[
    h_+(s)=\max\{h(s),0\},
    \qquad
    h_-(s)=\max\{-h(s),0\}.
\]
Note that they share the same expectations (say $m \geq 0$) because $\mathbb{E}[h(S)]=0$. If $m > 0$, we could define two probability distributions by reweighting the location distribution $P_S$ as
\[
    dP_+(s)=\frac{h_+(s)}{m}\,dP_S(s),
    \qquad
    dP_-(s)=\frac{h_-(s)}{m}\,dP_S(s).
\]
Here, $P_+$ assigns mass only to locations with above-average conditional expression, and $P_-$ only to locations with
below-average conditional expression. If $Q(P_{SX})=0$, \eqref{eq:th1_e1} implies that these distributions have identical
kernel embeddings. Because $k$ is characteristic, this would require $P_+=P_-$, contradicting their concentration on disjoint sets. Thus $m=0$ and $h=0$ almost surely.

Conversely, if $h=0$ almost surely, \eqref{eq:th1_e1} vanishes for every $f$, giving $Q(P_{SX})=0$.
\end{proof}

The characteristic condition ensures that distinct spatial distributions have different kernel embeddings, so every mean shift gives $Q(P_{SX})>0$. Mean independence gives $Q(P_{SX})=0$ for any bounded PSD spatial kernel, but a non-characteristic kernel may also assign zero to genuine mean shifts, introducing additional blind spots.

Mean independence is weaker than the statistical independence in Definition~\ref{def:sv}. Expression variance, skewness or other distributional features may vary across locations while the mean remains constant. All such changes leave the population \(Q\)-statistic at zero. The following corollaries show how this limits detection. 

For clarity, we restate the relevant definitions and assumptions. The location and data sampling models are introduced in Supplementary Note~\ref{supnote:2.formulation}. Let $k$ be a fixed, bounded, continuous PSD spatial kernel, with $K_{ij}=k(S_i,S_j)$. Define
\[
    Q_n=\mathbf{z}^{\top}\mathbf{K}\mathbf{z},
    \qquad
    z_i=\frac{X_i-\bar X}{\hat\sigma},
\]
where $\hat\sigma^2$ is the sample variance.

At a fixed nominal significance level $0<\alpha<1$, the test rejects when $Q_n>c_{n,\alpha}$. We assume that the critical value is computed from the spatial kernel and satisfies
\begin{equation}\label{eq:q_critical_rate}
    \frac{c_{n,\alpha}}{n}\xrightarrow{p}c_\alpha,
    \qquad 0<c_\alpha<\infty.
\end{equation}
This condition is important for establishing test behavior under mean-independent alternatives, where the population quadratic form is zero. When the population quadratic form is positive, $Q_n$ grows at rate $n^2$, so consistency only requires $c_{n,\alpha}=o_p(n^2)$. For the mean-independent alternatives in Corollary~\ref{corollary:q_test_mean_independent}, however, $Q_n$ remains of order $n$. The stronger condition \eqref{eq:q_critical_rate} allows us to compare the statistic and cutoff on this scale and determine the limiting rejection probability. Neither rate condition alone guarantees a false-positive rate of $\alpha$.

\textbf{Example (spectral null and moment-matching approximations).}
For a centered null spectrum with at least one positive eigenvalue, the spectral calibration from Supplementary Note~\ref{supnote:2.asymptotic}, based on \eqref{eq:v_stats_null}, satisfies \eqref{eq:q_critical_rate}. Its critical value is the $(1-\alpha)$-quantile of
\[
    \sum_{j=1}^{n}
    \lambda_j(\mathbf{H}\mathbf{K}\mathbf{H})\,\xi_j^2,
    \qquad \xi_j\stackrel{\mathrm{iid}}{\sim}N(0,1),
\]
where $\lambda_j(\mathbf{H}\mathbf{K}\mathbf{H})$ are the eigenvalues of the centered spatial kernel matrix. The Welch--Satterthwaite and Liu moment-matched approximations also satisfy \eqref{eq:q_critical_rate} when their spectral moments, normalized by the corresponding powers of $n$, are computed exactly or estimated consistently and the limiting upper-tail cutoff is positive\autocite{welch1938significance,satterthwaite1941synthesis,liu2009new}. Their limiting cutoffs need not equal the exact null quantile. Finally, when moments are estimated using a fixed number of stochastic trace probes, relative estimation error may not vanish as sample size increases. Nevertheless, an \(O_p(n)\) cutoff remains negligible compared with the order-\(n^2\) statistic under a positive population signal.

\begin{corollary}[\(Q\)-tests consistently detect mean shifts]
\label{corollary:q_test_mean_shift}
Under the setup above, the upper-tail test is consistent against every fixed alternative with $Q(P_{SX})>0$. If the spatial kernel is characteristic, this includes every alternative satisfying
\[
    \mathbb{E}[X\mid S]\neq\mathbb{E}[X]
    \quad\text{on a set of positive }P_S\text{ probability}.
\]
That is, $\Pr(Q_n>c_{n,\alpha})\to1$.
\end{corollary}

\begin{proof}
Consistency of the empirical HSIC estimator and the sample variance gives
\[
    \frac{Q_n}{n^2}
    \xrightarrow{p}
    \frac{Q(P_{SX})}{\operatorname{Var}(X)}
    >0.
\]
The variance factor appears because $Q_n$ uses standardized expression. By \eqref{eq:q_critical_rate}, $c_{n,\alpha}/n^2\to0$ in probability, so the rejection probability approaches one. For a characteristic spatial kernel, Theorem~\ref{theorem:q_mean_dependence} ensures $Q(P_{SX})>0$ whenever the conditional mean varies with location.
\end{proof}

\begin{corollary}[\(Q\)-tests are not consistent against changes that preserve the mean]
\label{corollary:q_test_mean_independent}
Under the same setup, consider a fixed alternative satisfying
\[
    \mathbb{E}[X\mid S]=\mathbb{E}[X]
    \quad\text{almost surely}.
\]
Then
\[
    \lim_{n\to\infty}\Pr(Q_n>c_{n,\alpha})<1,
\]
so the test is not consistent against this alternative. This result does not require the spatial kernel to be characteristic.

If the conditional variance is also constant,
\[
    \operatorname{Var}(X\mid S)=\operatorname{Var}(X)
    \quad\text{almost surely},
\]
the limiting rejection probability equals that under the independence null with the same marginal distributions. It equals $\alpha$ when the calibration is asymptotically exact.
\end{corollary}

\begin{proof}
Define $Z=(X-\mathbb{E}[X])/\sqrt{\operatorname{Var}(X)}$ and the centered spatial feature map
\[
    \psi(s)=k(s,\cdot)-\mathbb{E}[k(S,\cdot)].
\]
Mean independence gives $\mathbb{E}[Z\psi(S)]=0$. The bounded kernel and finite expression variance ensure a finite second moment, so the central limit theorem in the RKHS yields
\[
    \frac{Q_n}{n}
    =
    \left\|\frac{1}{\sqrt n}
    \sum_{i=1}^{n}Z_i\psi(S_i)\right\|_{\mathcal H_S}^2
    +o_p(1)
    \xrightarrow{d}
    \|G\|_{\mathcal H_S}^2
    =\sum_j\nu_j\xi_j^2,
\]
where $G$ is a zero-mean Gaussian element, $\xi_j$ are independent standard normal variables, and $\nu_j$ are the eigenvalues of the covariance operator
\[
    \Gamma
    =
    \mathbb{E}\!\left[
        \operatorname{Var}(Z\mid S)\,
        \psi(S)\otimes\psi(S)
    \right].
\]
Empirical centering and standardization do not change this limit.

The distribution of $\|G\|_{\mathcal H_S}^2$ is continuous at positive thresholds and assigns positive probability below every such threshold. Together with \eqref{eq:q_critical_rate}, this gives
\[
    \lim_{n\to\infty}\Pr(Q_n>c_{n,\alpha})
    =
    \Pr\!\left(\|G\|_{\mathcal H_S}^2>c_\alpha\right)
    <1.
\]

If the conditional variance is constant, then $\operatorname{Var}(Z\mid S)=1$ and
\[
    \Gamma=\Gamma_0
    :=\mathbb{E}[\psi(S)\otimes\psi(S)],
\]
the covariance operator under independence. The statistic therefore has the same limiting distribution as under the corresponding null. Because the critical value is computed from the spatial kernel, its limit is also unchanged. The limiting rejection probabilities coincide, equaling $\alpha$ when $c_\alpha$ is the exact $(1-\alpha)$-quantile of the null limit.
\end{proof}

Spatial variance shifts can change the rejection probability, so a \(Q\)-test may still have some power to detect them. However, Corollary~\ref{corollary:q_test_mean_independent} states that this power does not approach one as sample size grows. When both the conditional mean and variance are constant, changes in higher moments have the same limiting rejection probability as the independence null and are thus ``invisible'' to the \(Q\)-test in the large-sample limit.

\textbf{Example (spatial variance shifts).}
Consider
\[
    X=\mu+\sigma(S)\varepsilon,
\]
where $\varepsilon$ is independent of $S$, has zero mean and unit variance, and $0<\mathbb{E}[\sigma^2(S)]<\infty$. The conditional mean is $\mu$ everywhere, whereas the conditional variance $\sigma^2(S)$ differs across locations, meaning that some regions may show more heterogeneous expression. Although this variance pattern can still change the covariance operator $\Gamma$ and its spectral weights $\nu_j$, and therefore the probability of rejection, the \(Q\)-test remains inconsistent. We note that this result is compatible with spatial random-effects models (\eqref{eq:lmm}). Although a random field may have zero prior mean, conditioning on a particular realization $f$ gives
\[
    \mathbb{E}[X\mid S=s,f]=\mu+f(s).
\]
A nonconstant realized field therefore creates a mean pattern rather than a variance shift.

Sensitivity to these forms of dependence can be improved by replacing the linear data kernel. For example, squaring expression before testing (the data kernel $l(x,x')=x^2{x'}^2$) can turn a variance shift into a mean shift, which a \(Q\)-test can consistently detect under the stated kernel and sampling assumptions, with a finite fourth moment of $X$. More generally, characteristic spatial and data kernels allow HSIC to detect broader distributional dependence under the corresponding assumptions, at a higher computational cost. Finally, some custom data transformations also change what type of spatial variability a test can detect. For example, binarization, as used in \texttt{BinSpect}\autocite{dries2021giotto}, tests changes in the probability of exceeding a threshold, which can also result from a variance shift.

\subsection{Other dependence-based spatial variability tests}\label{supidx:n2:other-tests}
In this section, we review the remaining SVG detection methods categorized as ``dependence-based'' by \textcite{yan2025categorization} that were not previously discussed (e.g., those based on neither Moran's \(I\) nor HSIC).

\subsubsection{Relationship between mutual information (MI) and HSIC}
\label{supnote:2.mi_approx}
Mutual information (MI) measures dependence through the Kullback--Leibler (KL) divergence between the joint distribution and the product of its marginals. A density-estimation perspective connects MI to HSIC through quadratic measures of departure from independence.

Let $h_X$ and $h_S$ be normalized, square-integrable smoothing kernels, such as Gaussian densities. The corresponding kernel density estimates are
\[
    \hat f_X(x)=\frac1n\sum_i h_X(x-x_i),
    \qquad
    \hat f_S(s)=\frac1n\sum_i h_S(s-s_i),
    \qquad
    \hat f_{SX}(s,x)=\frac1n\sum_i h_S(s-s_i)h_X(x-x_i).
\]
The squared distance between the joint estimate and the product of its marginals measures their departure from independence. Expanding this distance gives an HSIC statistic (Section~4 of \textcite{gretton2010consistent}),
\begin{equation}\label{eq:hsic_kde}
    D_n
    =\iint
      \bigl[\hat f_{SX}(s,x)-\hat f_S(s)\hat f_X(x)\bigr]^2
      \,ds\,dx
    =\frac{1}{n^2}
      \operatorname{tr}(\mathbf K\mathbf H\mathbf L_X\mathbf H),
\end{equation}
where $K_{ij}=k(s_i,s_j)$ and $(L_X)_{ij}=l(x_i,x_j)$ use the overlap kernels
\[
    k(s,s')=\int h_S(v-s)h_S(v-s')\,dv,
    \qquad
    l(x,x')=\int h_X(u-x)h_X(u-x')\,du.
\]
These kernels measure how strongly the smoothing functions around two observations overlap. With the normalization used earlier, $S_{\mathrm{HSIC}}=n^2D_n/(n-1)^2$.

Near independence, MI also has a quadratic approximation. When the estimated density ratio $\hat f_{SX}/(\hat f_S\hat f_X)$ is close to one, a second-order expansion gives
\begin{equation}\label{eq:mi_approx}
    \widehat I(X;S)
    :=D_{\mathrm{KL}}(\hat f_{SX}\,\|\,\hat f_S\hat f_X)
    \approx
    \frac12\iint
    \frac{
      \bigl[\hat f_{SX}(s,x)-\hat f_S(s)\hat f_X(x)\bigr]^2
    }{\hat f_S(s)\hat f_X(x)}
    \,ds\,dx.
\end{equation}
The first-order term integrates to zero because both density estimates integrate to one. This approximation weights the squared difference by the inverse marginal densities. Including these weights in the overlap kernels, through $dv/\hat f_S(v)$ and $du/\hat f_X(u)$, gives the same HSIC quadratic form for twice the approximate MI. Thus, density estimation provides an intuitive connection between MI and kernel dependence statistics.

\textbf{Remark.}
For a linear data kernel $l(x,x')=xx'$, the expression feature reduces to $x$ itself. The resulting \(Q\)-statistic measures associations between expression and spatial features, with sensitivity determined by changes in the conditional mean $\mathbb E[X\mid S]$ (Theorem~\ref{theorem:q_mean_dependence}).

\subsubsection{BinSpect (Dries et al., 2021)}
\texttt{BinSpect} \autocite{dries2021giotto} from the Giotto toolbox uses a binarized representation of gene expression: $\mathbf{x} \in \{0,1\}^n$. Given a spatial network (here represented by its adjacency matrix $\mathbf{W}$, taken to be symmetric and binary with zero diagonal), \texttt{BinSpect} counts edges joining $0$-$0$, $0$-$1$, $1$-$0$, and $1$-$1$ vertices, constructs the $(2 \times 2)$-contingency table, and reports the associated odds ratio and $P$ value from Fisher's exact test. The entries of this contingency table can be expressed in vector notation as
\begin{align*}
n_{11} &= \mathbf{x}^\top \mathbf{W} \mathbf{x}, & n_{10} &= \mathbf{x}^\top \mathbf{W} (\mathbf{1}-\mathbf{x}), \\
n_{01} &= (\mathbf{1}-\mathbf{x})^\top \mathbf{W} \mathbf{x}, & n_{00} &= (\mathbf{1}-\mathbf{x})^\top \mathbf{W} (\mathbf{1}-\mathbf{x}).
\end{align*}
These expressions count each undirected edge in both orientations, giving a total of $m=\mathbf1^\top\mathbf W\mathbf1$ table entries.

The odds ratio $\frac{n_{00}n_{11}}{n_{01}n_{10}}$ is closely related to the $\chi^2$-statistic, which serves as a common approximation to Fisher's exact test. For positive margins, Pearson's statistic is \[\chi^2 = \frac{m(n_{00}n_{11}-n_{01}n_{10})^2}{(n_{00} + n_{01})(n_{10} + n_{11})(n_{00} + n_{10})(n_{01} + n_{11})},\] so that the odds ratio is $1$ iff $\chi^2 = 0$. 
Furthermore, by substituting in the expression for the entries of the contingency table, $\Delta \coloneqq n_{00} n_{11} - n_{01} n_{10}$ can be rewritten as a quadratic form involving the graph structure. Let $m = \mathbf{1}^\top \mathbf{W} \mathbf{1}$ be the total number of ordered edges and $\mathbf{k} = \mathbf{W}\mathbf{1}$ be the degree vector. We have
\begin{align*}
\Delta &= (\mathbf{x}^\top \mathbf{W} \mathbf{x})((\mathbf{1}-\mathbf{x})^\top \mathbf{W} (\mathbf{1}-\mathbf{x})) - (\mathbf{x}^\top \mathbf{W} (\mathbf{1}-\mathbf{x}))((\mathbf{1}-\mathbf{x})^\top \mathbf{W} \mathbf{x}) \\
&= (\mathbf{x}^\top \mathbf{W} \mathbf{x})(\mathbf{1}^\top \mathbf{W} \mathbf{1}) - (\mathbf{x}^\top \mathbf{W} \mathbf{1})(\mathbf{1}^\top \mathbf{W} \mathbf{x}) \\
&= m (\mathbf{x}^\top \mathbf{W} \mathbf{x}) - (\mathbf{x}^\top \mathbf{k})(\mathbf{k}^\top \mathbf{x}) \\
&= \mathbf{x}^\top \left( m \mathbf{W} - \mathbf{k} \mathbf{k}^\top \right) \mathbf{x}.
\end{align*}

The matrix $\mathbf{M} = m \mathbf{W} - \mathbf{k} \mathbf{k}^\top$ is a scaled version of the \emph{modularity matrix} $\mathbf{B} = \mathbf{W} - \frac{1}{m}\mathbf{k} \mathbf{k}^\top$ used in community detection. Thus, the \texttt{BinSpect} table determinant is proportional to the modularity of the partition induced by binarized expression,
\[
    \Delta = \mathbf{x}^\top \mathbf{M} \mathbf{x} \equiv (\mathbf{x}-\bar{x}\mathbf{1})^\top \mathbf{M} (\mathbf{x}-\bar{x}\mathbf{1}).
\]

Note that the data $\mathbf{x}$ are binarized separately for each gene using the same criterion (e.g., $k$-means with 2 clusters, or rank thresholding at 30\%). A fixed rank threshold approximately fixes $\bar{x}$, whereas $k$-means can produce different detection fractions. The degree-weighted table margins can also vary across genes and affect their ranking.

\textbf{Remark.} The spectrum of the unscaled modularity matrix $\mathbf B$ has been well-studied. In particular: (i) $\mathbf{1}$ is in the null space (eigenvalue 0), (ii) eigenvalues are interlaced with those of $\mathbf{W}$, (iii) the sum of eigenvalues is negative, and (iv) positive eigenvalues identify directions of increased modularity.  When $\mathbf B$ has eigenvalues of both signs, its quadratic form permits spectral cancellation. For \texttt{BinSpect}, cancellation also depends on which binary patterns the graph permits.

\textbf{Remark.} The denominator in $\chi^2$ contains the squared degree-weighted variance of $\mathbf{x}$. Assuming that $\mathbf{W}$ is symmetric and writing $\chi^2=\Delta^2/[m^3\mu^2(1-\mu)^2]$, the normalization is 
\begin{equation*}
\frac{(\mathbf{k}^\top \mathbf{x})^2 (\mathbf{k}^\top (\mathbf{1} - \mathbf{x}))^2}{m} = \frac{(m\mu)^2 (m(1-\mu))^2}{m} = m^3 [\mu(1-\mu)]^2,
\end{equation*}
\noindent where $\mu = \frac{\mathbf{k}^\top \mathbf{x}}{m}$ is the degree-weighted mean of $\mathbf{x}$. Furthermore, for a binary vector, the degree-weighted variance is $\mu(1-\mu)$.

\subsubsection{trendsceek (Edsgard et al., 2018)}
\texttt{trendsceek} \autocite{edsgard2018identification} uses various two-point statistics for marked point processes to detect spatial expression trends. These statistics compare the dependency between expression values ($\mathbf{x} = [x_1, \dots, x_n]^\top$) at two locations ($s_i, s_j$) conditioned on the distance $d(s_i, s_j) = r$. Specifically, they propose:
\begin{enumerate}[(a)]
    \item Stoyan's mark-correlation: \[\rho(r) = \frac{\mathbb{E} [x_i x_j \mid d(s_i, s_j) = r]}{(\mathbb{E}[X])^2}\]
    \item Mark-variogram: \[\gamma(r) = \frac{\mathbb{E} [(x_i - x_j)^2 \mid d(s_i, s_j) = r]}{2}.\]
    \item Mean-mark: \[E(r) = \mathbb{E} [x_i \mid d(s_i, s_j) = r]\]
    \item Variance-mark: \[V(r) = \mathbb{E}[(x_i - E(r))^2 \mid d(s_i, s_j) = r] = \mathbb{E}[x_i^2 \mid d(s_i, s_j) = r] - E(r)^2\]
\end{enumerate}

For each of these statistics $Q \in \{\rho, E, V, \gamma\}$, they compare the observed value of these statistics to the expected value under the null distribution obtained by permuting expression values \[\Delta_Q(r) = \abs{Q(r) - Q^{\text{exp}}(r)}.\] 

Estimation of these spatial statistics from data requires density estimation, as implemented in \texttt{spatstat} \autocite{baddeley2005spatstat}. To cast them into our quadratic form framework, we define a distance-based spatial weight matrix $\mathbf{W}^{(r)}$. In the simplest case (box kernel), $\mathbf{W}^{(r)}_{ij} = \frac{1}{|\mathcal{P}_r|}$ if $d(s_i, s_j) \approx r$ and $0$ otherwise, where $|\mathcal{P}_r|$ is the number of ordered pairs of distinct locations at distance $r$. More generally, $\mathbf{W}^{(r)}$ represents a kernel smoother centered at $r$. For simplicity, let us fix $r > 0$, assume there are ample pairs in the set $\mathcal{P}_r$, and that $\mathbf{W}^{(r)}$ is symmetric, nonnegative, has zero diagonal and sums to one. The statistics are estimated as follows:

\begin{enumerate}[(a)]
    \item The estimate of mark-correlation is:
    \[ \hat{\rho}(r) = \frac{\sum_{i,j} \mathbf{W}^{(r)}_{ij} x_i x_j}{\frac{1}{n^2} \sum_{i,j} x_i x_j} = \frac{\mathbf{x}^\top \mathbf{W}^{(r)} \mathbf{x}}{\bar x^2}. \]
    Let $\mathbf B=(\mathbf1\mathbf1^\top-\mathbf I)/[n(n-1)]$ average over distinct locations. Under permutation, the expected mark-correlation is $\mathbf x^\top\mathbf B\mathbf x/\bar x^2$. Hence
    \[
        \Delta_\rho(r)
        =\left|
          \frac{\mathbf x^\top(\mathbf W^{(r)}-\mathbf B)\mathbf x}
               {\bar x^2}
          \right|
        \propto
        \left|\mathbf x^\top(\mathbf W^{(r)}-\mathbf B)\mathbf x\right|,
    \]
    provided $\bar x\ne0$.
    \textcite{edsgard2018identification} compute a $P$ value by ranking the observed statistic against values obtained by permuting expression. The denominator in the ratio above is permutation-invariant and therefore does not affect these rankings.
    
    \item With $\mathbf{W} = \mathbf{W}^{(r)}$, the estimated mark-variogram is
    \[ \hat{\gamma}(r) = \frac{1}{2} \sum_{i,j} \mathbf{W}_{ij} (x_i - x_j)^2 = \mathbf{x}^\top \mathbf{L}_{\mathbf{W}} \mathbf{x}, \]
    where $\mathbf{L}_{\mathbf{W}} = \mathbf{D}_{\mathbf{W}} - \mathbf{W}$ is the Laplacian matrix associated with $\mathbf{W}$. The expected statistic is $\mathbf{x}^\top \mathbf{L}_{\mathbf{B}} \mathbf{x}$, where $\mathbf{L}_{\mathbf{B}}$ is the Laplacian associated with $\mathbf{B}$. Hence the deviation is
    \[ \Delta_\gamma(r) = \left| \mathbf{x}^\top (\mathbf{L}_{\mathbf{W}} - \mathbf{L}_{\mathbf{B}}) \mathbf{x} \right|. \]
    
    \item With $\mathbf{W} = \mathbf{W}^{(r)}$, the estimated mean-mark is
    \[ \hat{E}(r) = \sum_{i,j} \mathbf{W}_{ij} x_i = \mathbf{1}^\top \mathbf{W} \mathbf{x}. \]
    Hence the deviation from the permutation expectation is $\Delta_E(r) = \left| \mathbf{1}^\top (\mathbf{W} - \mathbf{B}) \mathbf{x} \right|$. While this is not a quadratic form, its square is
    \[ \Delta_E(r)^2 = \mathbf{x}^\top \boldsymbol{\delta} \boldsymbol{\delta}^\top \mathbf{x}, \]
    where $\boldsymbol{\delta} = \mathbf{W}^\top \mathbf{1} - \mathbf{B}^\top \mathbf{1}$ is the difference of the degree vectors.

    \item The estimated variance-mark is
    \[ \hat{V}(r) = \sum_{i,j} \mathbf{W}_{ij} x_i^2 - \left( \sum_{i,j} \mathbf{W}_{ij} x_i \right)^2 = \mathbf{x}^\top ( \mathbf{D}_{\mathbf{W}} - \mathbf{k} \mathbf{k}^\top) \mathbf{x}, \]
    where $\mathbf{k} = \mathbf{W}^\top\mathbf{1}$ is the column degree vector. Write $\mathbf K_{\mathbf W}=\mathbf D_{\mathbf W}-\mathbf k\mathbf k^\top$.
    Under permutation, the expected variance-mark is $(1-\|\mathbf k\|^2)\hat\sigma^2$, where $\hat\sigma^2=\mathbf x^\top\mathbf H\mathbf x/(n-1)$. Thus
    \[
        \Delta_V(r)=
        \left|\mathbf x^\top
        \left[
            \mathbf K_{\mathbf W}
            -\frac{1-\|\mathbf k\|^2}{n-1}\mathbf H
        \right]\mathbf x\right|.
    \]
\end{enumerate}

At each fixed distance $r$, all four metrics in \texttt{trendsceek} reduce to comparing a specific quadratic form $\mathbf{x}^\top \mathbf{K}^{(r)} \mathbf{x}$ against a null quadratic form derived from the global distribution. The kernels $\mathbf{K}^{(r)}$ for mark-variogram, mean-mark and variance-mark (but not mark-correlation) are double-centered, meaning that the test statistics do not depend on the mean of $\mathbf{x}$, thus connecting to the standardized data $\mathbf{z}$. The complete procedure evaluates these deviations across distances using permutation envelopes.

Indeed, from a kernel design perspective these four statistics (for fixed $r$) use well-known constructions: weight matrices, Laplacian matrices, and mean-field approximations thereof.  What is particularly interesting about these spatial statistics is that the weight matrices generalize the usual weight matrix $\mathbf{W}^{(0)}$ which is peaked for nearby points ($r = 0$) and instead consider a family of such matrices that are sensitive to points at a specific distance $r \in [0, \infty)$. However, once $r > 0$, we see that $\mathbf{W}^{(r)}$ has trace zero, and as a consequence has both positive and negative eigenvalues (if not identically zero). The raw Laplacian $\mathbf{L}^{(r)}$ is positive semidefinite. Its deviation from the permutation mean uses $\mathbf L_{\mathbf W}-\mathbf L_{\mathbf B}$, which can have eigenvalues of both signs.

\subsubsection{singleCellHaystack (Vandenbon et al., 2020)}\label{supnote:2.singleCellHaystack}
The original binary formulation of \texttt{singleCellHaystack}\autocite{vandenbon2020clustering} uses an information-theoretic measure to quantify the independence of cell spatial and expression distributions. More precisely, they consider a distribution \(Q\) of all cells in the space, along with distributions $P_0, P_1$ of cells in space with undetectable or detected expression of a gene ($x_i \in \{0,1\}$), and define the divergence of the gene to be: \[D_{KL} = \KL{P_0}{Q} + \KL{P_1}{Q},\] where $\operatorname{KL}$ is the Kullback-Leibler divergence. The distributions $Q, P_0, P_1$ are estimated from finite samples $\{(x_i, s_i)\}$ as mixtures and can be parametrized by their mixture weights $\mathbf{y} \in [0,1]^n$ (normalized to have sum one): given grid points $p$, the normalized grid probabilities are
\[
    f_{\mathbf y}(p)
    =\frac{\mathbf y^\top\mathbf h(p)}{c_{\mathbf y}},
    \qquad
    c_{\mathbf y}=\mathbf y^\top\mathbf a,
    \qquad
    a_i=\sum_p h(s_i,p),
\]
where $\mathbf h(p)=(h(s_1,p),\ldots,h(s_n,p))^\top$.
We hold the spatial density contributions fixed and omit the numerical pseudocount. With this notation, the density estimate of \(Q\) is given by $f_\frac{\mathbf{1}}{n}$; the density estimate of $P_0$ (resp.\ $P_1$) is given by $f_\frac{\mathbf{1} - \mathbf{x}}{n - \norm{\mathbf{x}}_1}$ (resp.\ $f_\frac{\mathbf{x}}{\norm{\mathbf{x}}_1}$). (In the original method, $h(s, p) \propto e^{-d(s,p)^2/2}$, but the precise form doesn't matter for our result.)

From the discussion in Supplementary Note~\ref{supnote:2.mi_approx}, we can approximate the KL divergence, and hence $D_{KL}$, by quadratic forms.  More specifically, if $\mathbf{y}$ and $\mathbf{y}'$ are two sufficiently close sets of mixture weights and $f_{\mathbf y}(p)>0$ at every grid point, then 
\begin{equation*}
    \KL{f_{\mathbf{y}'}}{f_\mathbf{y}} = (\mathbf{y}' - \mathbf{y})^\top \mathbf{K}(\mathbf y) (\mathbf{y}' - \mathbf{y}) + O(\norm{\mathbf{y}' - \mathbf{y}}_2^3),
\end{equation*}
\noindent where
\[
    K_{ij}(\mathbf y)
    =\frac{1}{2c_{\mathbf y}^{\,2}}
      \sum_p
      \frac{
        [h(s_i,p)-f_{\mathbf y}(p)a_i]
        [h(s_j,p)-f_{\mathbf y}(p)a_j]
      }{f_{\mathbf y}(p)}
\]
is related to the Fisher information metric.

Applying this local approximation at $\mathbf y=\mathbf1/n$, write $\mathbf K=\mathbf K(\mathbf1/n)$. Then
\begin{align*}
    \KL{P_0}{Q} &\approx \mathbf{z}^{(0)\top} \mathbf{K} \mathbf{z}^{(0)}, \\
    \KL{P_1}{Q} &\approx \mathbf{z}^{(1)\top} \mathbf{K} \mathbf{z}^{(1)}, \\
    D_{KL} &\approx \mathbf{z}^{(0)\top} \mathbf{K} \mathbf{z}^{(0)} + \mathbf{z}^{(1)\top} \mathbf{K} \mathbf{z}^{(1)}, 
\end{align*}
\noindent where
\begin{align*}
    \mathbf{z}^{(0)} &= \frac{\mathbf{1} - \mathbf{x}}{n - \norm{\mathbf{x}}_1} - \frac{\mathbf{1}}{n}, \\
    \mathbf{z}^{(1)} &= \frac{\mathbf{x}}{\norm{\mathbf{x}}_1} - \frac{\mathbf{1}}{n}.
\end{align*}

Furthermore, observe that $\mathbf{z}^{(0)} = -\frac{\norm{\mathbf{x}}_1}{n-\norm{\mathbf{x}}_1} \mathbf{z}^{(1)}$, so that we may simplify:
\begin{equation*}
    D_{KL} \approx \left(1 + \left(\frac{\norm{\mathbf{x}}_1}{n-\norm{\mathbf{x}}_1}\right)^2\right) \mathbf{z}^{(1)\top} \mathbf{K} \mathbf{z}^{(1)}
\end{equation*}

The matrix $\mathbf K$ is a weighted Gram matrix of the normalized density derivatives at the grid points. It is positive semidefinite and satisfies $\mathbf K\mathbf1=0$. With $J$ grid points, its rank is at most $\min(n-1,J-1)$. The binary high-dimensional implementation uses 100 grid points by default, limiting the spatial directions represented by this local quadratic approximation.

\subsection{\(Q\)-tests for cell-type-specific spatial variability testing}\label{supnote:2.ctsv}

In many biological contexts, spatial patterns may be driven by confounding factors, such as cell-type distribution. For instance, a gene may appear spatially variable simply because it is a marker for a cell type that is spatially clustered. To distinguish intrinsic spatial regulation from confounding effects, we introduce conditional spatial variability.

\begin{definition}[Conditional Spatial Variability]
    Let $C$ be a confounding random variable (e.g., cell type). We say $X$ exhibits conditional spatial variability given $C$ if $X$ and $S$ are conditionally dependent given $C$. Formally, the null hypothesis of conditional independence ($X \perp \!\!\! \perp S \mid C$) implies
    \[
    f_{X|S, C}(x|s, c) = f_{X|C}(x|c), \quad \forall x, s, c.
    \]
    Conversely, $X$ is conditionally spatially variable if this equality fails for a set of $c$ with non-zero probability.
\end{definition}

Cell-type-specific spatial testing is most useful for lower-resolution assays, where each spot contains a mixture of cell types whose effects must be disentangled. As resolution approaches the single-cell level, cell-type membership becomes nearly binary, and spatial variability can be tested within a selected type by analyzing cells with \(c_i=1\). Below, we focus on the lower-resolution setting, where \(C\) is a continuous, simplex-valued vector of cell-type proportions.

The cell-type-specific SVG methods considered here use parametric models to account for spatial effects and cell-type composition. We examine \texttt{C-SIDE}\autocite{cable2022cell} in more detail to show how these proportions enter the test. For $K$ cell types indexed by $k\in\{1,\ldots,K\}$, the model is 
\begin{align*}
    x_i \mid \lambda_i &\sim \operatorname{Poisson}(n_i \lambda_i) \\
    \log \lambda_i &= \log \left(\sum_{k=1}^K \pi_{ik} \mu_{ik}\right) + \epsilon_i \\
    \log \mu_{ik} &= \alpha_{k0} + \sum_{l=1}^d \alpha_{kl} h_l(s_i)
\end{align*}
\noindent where $\{(x_i, s_i)\}$ denotes the spatially resolved transcriptomic data as before, $n_i$ is a library size factor at spot $s_i$, $\lambda_i$ is the size-normalized expression rate, $\pi_{ik}$ is the proportion of cell type $k$ at spot $s_i$, $\mu_{ik}$ is the cell-type-specific expression rate, and $\epsilon_i \sim N(0, \sigma^2)$ is a noise term. Furthermore, the cell-type-specific expression rate $\mu_{ik}$ is decomposed into a spatially invariant term $\alpha_{k0}$, and contributions $\alpha_{kl}$ that depend on the $d$ spatial covariates $h_l(s_i)$ encoding transformations of spatial coordinates. 

Ostensibly, this framework permits simultaneous inference of cell-type proportions and cell-type-specific spatial effects. In practice, \texttt{C-SIDE} estimates proportions separately, for example using \texttt{RCTD}\autocite{cable2022robust}, and treats them as fixed when fitting expression effects. This two-step approach ignores the uncertainty in the proportion estimation and can bias the inferred spatial effects. Nevertheless, it allows us to apply a local working approximation similar to those previously used for GLMs and GLMMs. 

To examine how cell-type composition enters a \(Q\)-test, we linearize the marginal mean around $H_0:\alpha_{kl}=0$ for all spatial coefficients. Assume independent random effects $\epsilon_i\sim\mathcal{N}(0,\sigma^2)$, and hold the library sizes, cell-type proportions and fitted null parameters fixed in this calculation. Integrating over the random effects gives
\begin{align*}
    m_i
    &\coloneqq\mathbb{E}[x_i\mid H_0]
    =n_i e^{\sigma^2/2}\sum_k\pi_{ik}e^{\alpha_{k0}},\\
    V_i
    &\coloneqq\operatorname{Var}(x_i\mid H_0)
    =m_i+m_i^2(e^{\sigma^2}-1).
\end{align*}
The factor $e^{\sigma^2/2}$ arises from the mean of the multiplicative random effect.

For small spatial effects,
\[
    \mathbb{E}[x_i]
    \approx m_i+m_i\sum_{k,l}w_{ik}\alpha_{kl}h_l(s_i),
    \qquad
    w_{ik}=
    \frac{\pi_{ik}e^{\alpha_{k0}}}
    {\sum_{k'}\pi_{ik'}e^{\alpha_{k'0}}}.
\]
Here, $w_{ik}$ is the expected fraction of this gene's expression contributed by cell type $k$ at spot $i$ under the null.

Define $\Phi_{i,(k,l)}=m_iw_{ik}h_l(s_i)$ and $\mathbf{V}=\operatorname{diag}(V_1,\ldots,V_n)$. The quasiscore and its null covariance are
\[
    \mathbf{U}=\boldsymbol{\Phi}^{\top}\mathbf{V}^{-1}(\mathbf{x}-\mathbf{m}),
    \qquad
    \mathcal{I}=\boldsymbol{\Phi}^{\top}\mathbf{V}^{-1}\boldsymbol{\Phi}.
\]
As in the GLM section, the information-normalized statistic is
\[
    S=\mathbf{U}^{\top}\mathcal{I}^+\mathbf{U}
    =\mathbf{z}^{\top}\mathbf{K}_{\mathrm{proj}}\mathbf{z},
    \qquad
    \mathbf{z}=\mathbf{V}^{-1/2}(\mathbf{x}-\mathbf{m}),
\]
where
\[
    \mathbf{K}_{\mathrm{proj}}
    =\mathbf{V}^{-1/2}\boldsymbol{\Phi}
    (\boldsymbol{\Phi}^{\top}\mathbf{V}^{-1}\boldsymbol{\Phi})^+
    \boldsymbol{\Phi}^{\top}\mathbf{V}^{-1/2}.
\]
Here, $+$ denotes the Moore--Penrose inverse. \texttt{C-SIDE} uses Wald tests, which differ from the quasiscore-based test above. Their relationship is discussed in Supplementary Note~\ref{supnote:2.glm}.

Summing the squared quasiscores gives another related quadratic statistic with a simpler kernel,
\[
    T=\mathbf{U}^{\top}\mathbf{U}
    =\mathbf{z}^{\top}\mathbf{K}\mathbf{z},
    \qquad
    \mathbf{K}=\mathbf{V}^{-1/2}\boldsymbol{\Phi}\boldsymbol{\Phi}^{\top}\mathbf{V}^{-1/2}.
\]
Here, $\mathbf{z}$ is the same unit-variance Pearson residual vector used above. This statistic retains the relative scaling of the spatial features, whereas $S$ normalizes by their information. Setting $A_{ik}=m_iw_{ik}$ and $L_{il}=h_l(s_i)$ gives
\[
    \mathbf{K}=\mathbf{V}^{-1/2}(\mathbf{A}\mathbf{A}^{\top}\circ\mathbf{L}\mathbf{L}^{\top})\mathbf{V}^{-1/2},
\]
where $\circ$ denotes the Hadamard product. Cell-type expression contributions and null observation variances therefore weight the spatial kernel. If each cell type's proportion is constant across spots, and library sizes and null variances are also constant, this construction reduces to a scalar-weighted ordinary spatial \(Q\)-test. If each spot contains one known cell type, the kernel is block diagonal after grouping spots by type, and $T$ decomposes into a sum of within-type quadratic statistics.

Other methods can be treated similarly. They differ from \texttt{C-SIDE} by assuming a different emission probability distribution (e.g., zero-inflated negative binomial instead of Poisson for \texttt{CTSV} \autocite{yu2022ctsv}), different functional forms relating $\lambda_i$ to $\mu_{ik}$, or different ways in which spatial information is encoded (e.g., \texttt{spVC} \autocite{yu2024spvc} uses learned splines). Observe that the parametric form of \texttt{C-SIDE} from above results in the spatial information ``entering'' as a linear kernel. Recent approaches such as \texttt{STANCE} \autocite{su2025stance}, \texttt{Celina} \autocite{shang2025celina}, and \texttt{MMM} \autocite{wang2025mmm} use more general (e.g., distance-based) kernels  $\mathbf{K}'$ directly instead of linear kernels.  

More generally, replacing $\mathbf{L}\mathbf{L}^{\top}$ with a spatial kernel $\mathbf{K}'$ gives
\[
    \mathbf{K}=\mathbf{V}^{-1/2}(\mathbf{A}\mathbf{A}^{\top}\circ\mathbf{K}')\mathbf{V}^{-1/2}.
\]
Because $\mathbf{A}\mathbf{A}^{\top}$ is positive semidefinite, the Schur product theorem and positive diagonal scaling guarantee that $\mathbf{K}$ is positive semidefinite whenever $\mathbf{K}'$ is, provided $V_i>0$. If $\mathbf{K}'$ is strictly positive definite and every row of $\mathbf{A}$ is nonzero, then $\mathbf{K}$ is also strictly positive definite. This weighting generally changes the eigenvalues and eigenvectors. Therefore, despite their diversity, the behavior of these cell-type-specific methods largely mirrors that of the corresponding non-cell-type-specific tests, with cell-type contributions reweighting the spatial kernel.

\subsection{Equivalence between multivariate spatial variability tests}\label{supidx:n2:multivariate}

In this section, we extend our results to a multivariate setting where the random variable of interest is a \(p\)-dimensional vector, such as the usage ratios of \(p\) isoforms of a given gene. We first define a similar quadratic form statistic for multivariate spatial variability detection, derive its null distribution, and subsequently demonstrate its connection to the HSIC framework.

\subsubsection{A multivariate quadratic form from multinomial GLMM and the score test}

Following \textcite{su2026mapping}, we consider a multinomial GLMM with a shared spatial variance structure across isoforms. Let \(\mathbf{Y} \in \mathbb{R}^{n \times p}\) be the spatial count matrix for \(n\) spots and \(p\) isoforms, where the count vector \(\mathbf{y}_i\) for the \(i\)-th spot follows a multinomial distribution,
\[
\mathbf{y}_i\mid\mathbf{p}_i \sim \text{Multinom}(N_i, \mathbf{p}_i).
\]
Here, \(N_i=\sum_{j=1}^p y_{ij}\) denotes the total count of the gene at spot \(i\), and \(\mathbf{p}_i \in \Delta^{p-1}\) represents the expected isoform proportions, with \(\Delta^{p-1}=\{\mathbf{p}\in[0,1]^p:\sum_{j=1}^p p_j=1\}\). To enforce the compositional constraint and model spatial dependency, we utilize a baseline-category logit link. Assuming all expected proportions are strictly positive and selecting the \(p\)-th isoform as the reference, we model the remaining \(q = p - 1\) linearly independent dimensions. For \(j = 1, \dots, q\), the linear predictor is
\[
\eta_{ij} = \ln\left(\frac{p_{ij}}{p_{ip}}\right) = b_j + e_{ij},
\]
where \(\mathbf{b} \in \mathbb{R}^q\) is the baseline log-odds vector, and \(\mathbf{E} \in \mathbb{R}^{n \times q}\) is the matrix of spatial random effects. We assume the random effects for each isoform \(j\) follow a multivariate normal distribution with a shared spatial correlation structure, \(\mathbf{e}_{:j} \sim \mathcal{N}(\mathbf{0}, \boldsymbol{\Sigma}_\theta)\), where \(\boldsymbol{\Sigma}_\theta = \tau^2(\mathbf{I} + \theta \mathbf{K})\). Furthermore, \(\mathbf{e}_{:j}\) is assumed to be independent across isoforms \(j = 1, \dots, q\).

Rather than variation in the univariate gene expression \(\mathbf{N}=[N_1,\ldots,N_n]^\top\), we are interested in detecting spatial variability in the multivariate isoform usage \(\mathbf{P}=[\mathbf{p}_1,\ldots,\mathbf{p}_n]^\top\). This is a hidden variable with empirical estimates \(\mathbf{X}=[\mathbf{x}_1,\ldots,\mathbf{x}_n]^\top\in[0,1]^{n\times p}\), where \(\mathbf{x}_i=\mathbf{y}_i/N_i\in\Delta^{p-1}\). We condition on the totals \(N_i>0\) throughout this derivation.

To test \(H_0:\theta=0\), we integrate out the random effects and differentiate the marginal log-likelihood,
\[
    \ell(\theta)=\log\int f(\mathbf{Y}\mid\mathbf{E})f_\theta(\mathbf{E})\,d\mathbf{E}.
\]
Holding the null parameters fixed, the score satisfies
\[
    U(0)=\frac{1}{2\tau^2}\mathbb{E}\!\left[
    \operatorname{tr}(\mathbf{E}^\top\mathbf{K}\mathbf{E})\mid\mathbf{Y},H_0
    \right]-\frac{q}{2}\operatorname{tr}(\mathbf{K}).
\]
The posterior expectation includes both the estimated random effects and their uncertainty.

As in the univariate GLMM section (Supplementary Note~\ref{supnote:2.glmm}), we use a Gaussian working approximation around zero random effect. Let \(\hat{\mathbf{p}}_0\) denote the baseline probabilities obtained from the fitted intercepts, and write \(\tau^2>0\) for the fitted null random-effect variance. The multinomial score and information at \(\mathbf{e}_i=\mathbf{0}\) are
\begin{align*}
    \mathbf{U}_i &= \mathbf{y}_{i,1:q}-N_i\hat{\mathbf{p}}_{0,1:q}, \\
    \mathcal{I}_i &= N_i\hat{\boldsymbol{\Omega}},
    \qquad \hat{\boldsymbol{\Omega}}
    =\operatorname{diag}(\hat{\mathbf{p}}_{0,1:q})
    -\hat{\mathbf{p}}_{0,1:q}\hat{\mathbf{p}}_{0,1:q}^\top.
\end{align*}
Under the null, the joint log-density separates across spots. For \(\mathcal{Q}_i(\mathbf{e}_i)=\log f(\mathbf{y}_i\mid\mathbf{e}_i)+\log f_0(\mathbf{e}_i)\), the second-order expansion is
\begin{equation}\label{eq:approx_pl_mglmm}
    \mathcal{Q}_i(\mathbf{e}_i)\approx\text{const}
    +\mathbf{U}_i^\top\mathbf{e}_i
    -\frac12\mathbf{e}_i^\top\mathcal{I}_i\mathbf{e}_i
    -\frac{1}{2\tau^2}\mathbf{e}_i^\top\mathbf{e}_i.
\end{equation}
This gives the Gaussian posterior approximation
\[
    \mathbf{e}_i\mid\mathbf{y}_i,H_0
    \approx\mathcal{N}(\hat{\mathbf{e}}_i,\mathbf{C}_i),
    \qquad
    \mathbf{C}_i=(\mathcal{I}_i+\tau^{-2}\mathbf{I})^{-1},
    \qquad
    \hat{\mathbf{e}}_i=\mathbf{C}_i\mathbf{U}_i.
\]
Defining \(\tilde{\mathbf{z}}_i=\hat{\mathbf{e}}_i/\tau\) and \(\tilde{\mathbf{Z}}=[\tilde{\mathbf{z}}_1,\ldots,\tilde{\mathbf{z}}_n]^\top\), we obtain the multivariate quadratic statistic
\begin{equation}\label{eq:q_mglmm}
    Q_{\text{multinom-GLMM}}
    =\operatorname{tr}(\tilde{\mathbf{Z}}^\top\mathbf{K}\tilde{\mathbf{Z}})
    =\sum_{j=1}^q\tilde{\mathbf{Z}}_{:j}^\top\mathbf{K}\tilde{\mathbf{Z}}_{:j},
\end{equation}
with the approximate score
\[
    U(0)\approx\frac12\left[
    Q_{\text{multinom-GLMM}}
    +\tau^{-2}\sum_iK_{ii}\operatorname{tr}(\mathbf{C}_i)
    -q\operatorname{tr}(\mathbf{K})\right].
\]
For fixed null parameters and count totals, the remaining terms are constants. An upper-tail score test is therefore equivalent to an upper-tail \(Q\)-test within this working approximation.

Substituting \(\mathbf{U}_i\) gives the transformed residuals directly in terms of the empirical ratios,
\begin{equation}\label{eq:mglmm_transformation}
\begin{aligned}
    \tilde{\mathbf{z}}_i
    &=\frac{\hat{\mathbf{e}}_i}{\tau}
    =\tau\left(\tau^2\hat{\boldsymbol{\Omega}}+\frac1{N_i}\mathbf{I}\right)^{-1}
    (\mathbf{x}_{i,1:q}-\hat{\mathbf{p}}_{0,1:q}) \\
    &:=\mathbf{W}_i(\mathbf{x}_{i,1:q}-\hat{\mathbf{p}}_{0,1:q}).
\end{aligned}
\end{equation}
This transformation reduces the contribution of ratios supported by few counts. Fitting the null model gives estimates of the random-effect variance and baseline probabilities. These probabilities generally differ from the empirical mean ratios.

\subsubsection{The null distribution of the multivariate quadratic form}

We now consider a general multivariate quadratic form. Let $Z$ be a zero-mean, $q$-dimensional variable independent of location $S$ under the null, with finite fourth moments. Under the random-location sampling model, assume that the spatial kernel satisfies the integrability and diagonal trace conditions used in \eqref{eq:v_stats_null}. Bounded continuous PSD kernels on a compact domain provide a sufficient case. The same argument then applies to the joint observations $a_i=(\mathbf z_i,s_i)$,
\[
    Q_n = \text{tr}(\mathbf{Z}^\top \mathbf{K} \mathbf{Z})=\sum_{i=1}^n\sum_{j=1}^n H(a_i, a_j),
\]
with the degenerate kernel \(H(a_i, a_j)=\mathbf{z}_i^\top\mathbf{z}_jk(s_i, s_j)\).

The corresponding integral operator $\mathcal{H}$ is the tensor product of $\mathcal{K}$ (as in \eqref{eq:v_stats_null}) and $\mathcal{L}$, where $\mathcal{L}(f)(\mathbf{z})=\int \mathbf{z}^\top\mathbf{z}' f(\mathbf{z}')\,dP_Z(\mathbf{z}')$. Its eigenfunctions with nonzero eigenvalues take a linear form $\psi(\mathbf{z})=\mathbf{z}^\top \mathbf{v}$ for some vector $\mathbf{v} \in \mathbb{R}^q$. Defining the feature covariance matrix \(
\boldsymbol{\Sigma}_Z = \mathbb{E}[\mathbf{z}\mathbf{z}^\top] = \int \mathbf{z}\mathbf{z}^\top\,dP_Z(\mathbf{z}) \in \mathbb{R}^{q \times q}\),
an eigenvalue-eigenfunction pair $(\gamma, \psi)$ of $\mathcal{L}$ must satisfy
\[
\int \mathbf{z}^\top\mathbf{z}' (\mathbf{z}'^\top \mathbf{v})\,dP_Z(\mathbf{z}') = \mathbf{z}^\top \boldsymbol{\Sigma}_Z \mathbf{v} = \gamma \mathbf{z}^\top \mathbf{v}
\] 
for $P_Z$-almost every \(\mathbf{z}\). Thus, the nonzero eigenvalues of $\mathcal{L}$ are those of $\boldsymbol{\Sigma}_Z$. Under the stated null and sampling assumptions, we have
\begin{equation}\label{eq:v_stats_null_mv}
    \frac{1}{n}Q_{n}  = \frac{1}{n}\text{tr}(\mathbf{Z}^\top \mathbf{K} \mathbf{Z}) \xrightarrow[]{d}  \sum_{i=1}^\infty\sum_{j=1}^q\lambda_i\gamma_j\chi^2_{1, ij}
\end{equation}
 where $\{\lambda_i\}$ and $\{\gamma_j\}$ are eigenvalues of $\mathcal{K}$ (potentially infinite) and $\boldsymbol{\Sigma}_Z$ (finite), respectively, and $\chi^2_{1, ij}$ are independent $\chi^2$ variables with one degree of freedom.

For finite-sample calculations, condition on the locations and a fixed symmetric kernel matrix $\mathbf{K}$. For independent zero-mean rows $\mathbf{z}_i$ with covariance matrices $\boldsymbol{\Sigma}_i$ and finite fourth moments,
\begin{equation}\label{eq:mv_q_moments}
\begin{aligned}
    \mathbb{E}[Q_n] &= \sum_iK_{ii}\operatorname{tr}(\boldsymbol{\Sigma}_i), \\
    \operatorname{Var}(Q_n)
    &=2\sum_{i,j}K_{ij}^2\operatorname{tr}(\boldsymbol{\Sigma}_i\boldsymbol{\Sigma}_j)
    +\sum_iK_{ii}^2\kappa_i,
\end{aligned}
\end{equation}
where the fourth-moment correction is
\[
    \kappa_i=\mathbb{E}\|\mathbf{z}_i\|^4
    -\bigl(\operatorname{tr}\boldsymbol{\Sigma}_i\bigr)^2
    -2\operatorname{tr}(\boldsymbol{\Sigma}_i^2).
\]
For identically distributed rows, these expressions reduce to
\begin{align*}
    \mathbb{E}[Q_n]&=\operatorname{tr}(\mathbf{K})\operatorname{tr}(\boldsymbol{\Sigma}_Z), \\
    \operatorname{Var}(Q_n)&=2\operatorname{tr}(\mathbf{K}^2)\operatorname{tr}(\boldsymbol{\Sigma}_Z^2)
    +\kappa_Z\sum_iK_{ii}^2.
\end{align*}
The correction is zero for Gaussian rows. For empirically centered features, the same moment formulas apply with $\mathbf{H}\mathbf{K}\mathbf{H}$ in place of $\mathbf{K}$, keeping the uncentered independent rows in the calculation.

\subsubsection{The null distribution of the quadratic form from a multinomial GLMM}

Under \(H_0\), isoform proportions vary independently across spots but share the same expectation \(\bar{\mathbf{p}}\). Write
\[
    \bar{\mathbf{p}}=\mathbb{E}_0[\mathbf{p}_i],
    \qquad \mathbf{B}=\operatorname{Cov}_0(\mathbf{p}_i),
    \qquad \boldsymbol{\Omega}_p
    =\operatorname{diag}(\bar{\mathbf{p}})-\bar{\mathbf{p}}\bar{\mathbf{p}}^\top.
\]
Conditional on the count totals \(N_i>0\), the exact moments are
\begin{align*}
    \mathbb{E}_0[\mathbf{y}_i]&=N_i\bar{\mathbf{p}},
    &\operatorname{Cov}_0(\mathbf{y}_i)
    &=N_i\boldsymbol{\Omega}_p+N_i(N_i-1)\mathbf{B}, \\
    \mathbb{E}_0[\mathbf{x}_i]&=\bar{\mathbf{p}},
    &\boldsymbol{\Sigma}_{X,i}:=\operatorname{Cov}_0(\mathbf{x}_i)
    &=\frac{\boldsymbol{\Omega}_p}{N_i}
    +\left(1-\frac1{N_i}\right)\mathbf{B}.
\end{align*}
Under the multinomial GLMM, increasing sequencing depth reduces sampling noise but does not remove variation in the underlying isoform proportions. Indeed, the ratio covariance approaches $\mathbf{B}$ as $N_i$ increases, and the ordinary multinomial formulas are only recovered when $\mathbf{B}=\mathbf{0}$.

To develop a test for the \(Q\)-statistic in \eqref{eq:q_mglmm}, we use the corresponding Gaussian working null,
\[
    \mathbf{U}_i\mid\mathbf{e}_i
    \approx\mathcal{N}(\mathcal{I}_i\mathbf{e}_i,\mathcal{I}_i).
\]
Integrating over \(\mathbf{e}_i\sim\mathcal{N}(\mathbf{0},\tau^2\mathbf{I})\) gives independent, approximately Gaussian residuals with
\[
    \mathbb{E}_{\mathrm{work}}[\tilde{\mathbf{z}}_i]=\mathbf{0},
    \qquad
    \boldsymbol{\Sigma}_i
    =\operatorname{Cov}_{\mathrm{work}}(\tilde{\mathbf{z}}_i)
    =\mathbf{I}-\tau^{-2}\mathbf{C}_i.
\]
Substituting these covariances into \eqref{eq:mv_q_moments}, with $\kappa_i=0$ for Gaussian rows, gives the working null moments. In particular,
\[
    \mathbb{E}_{\mathrm{work}}[Q_{\text{multinom-GLMM}}]
    =\sum_iK_{ii}\operatorname{tr}(\boldsymbol{\Sigma}_i),
    \qquad
    U(0)\approx\frac12\left(Q_{\text{multinom-GLMM}}
    -\mathbb{E}_{\mathrm{work}}[Q_{\text{multinom-GLMM}}]\right).
\]
As in the univariate case, unequal count totals require spot-specific covariance matrices, so the common-distribution limit in \eqref{eq:v_stats_null_mv} generally does not apply directly conditional on these totals. When we do have $N_i\equiv N$, the working residuals share a common covariance. We can then compare an empirical estimate with the working-model covariance,
\begin{enumerate}
    \item \textbf{Non-parametric:} $\widehat{\boldsymbol{\Sigma}}_Z=(n-1)^{-1}\tilde{\mathbf{Z}}^\top\mathbf{H}\tilde{\mathbf{Z}}$.
    \item \textbf{Gaussian working model:} $\boldsymbol{\Sigma}_Z^{\mathrm{work}}=\mathbf{I}-\tau^{-2}\mathbf{C}$, where $\mathbf{C}=(N\hat{\boldsymbol{\Omega}}+\tau^{-2}\mathbf{I})^{-1}$.
\end{enumerate}
This parallels the distinction between the empirical variance and the model-based null variance in the univariate GLM (\eqref{eq:score_glm}).

\subsubsection{HSIC-based multivariate spatial variability detection}

The derivation of \(\tilde{\mathbf{Z}}\) and its null relied heavily on the parametric assumption that \(\mathbf{y}_i\) follows a multinomial distribution. However, the transformation in \eqref{eq:mglmm_transformation} can be viewed more broadly as a feature map \(\tilde{\mathbf{z}}_i = \phi(\mathbf{y}_i)\) connecting the observed data space \(\mathcal{Y}\) to the weighted residual kernel space \(\mathcal{H}_Y\) with inner product \(\langle\phi(\mathbf{y}),\phi(\mathbf{y}')\rangle=\phi(\mathbf{y})^\top\phi(\mathbf{y}')=\tilde{\mathbf{z}}^\top \tilde{\mathbf{z}}'\). 
For a positive semidefinite spatial kernel, centering the transformed features gives an empirical HSIC statistic between the spatial coordinates and the transformed data:
\begin{equation}\label{eq:q_mhsic}
\begin{aligned}
    Q_{\text{multinom-HSIC}}
    &=\operatorname{tr}(\tilde{\mathbf{Z}}^\top\mathbf{H}\mathbf{K}\mathbf{H}\tilde{\mathbf{Z}}) \\
    &=\operatorname{tr}(\mathbf{K}\mathbf{H}\tilde{\mathbf{Z}}\tilde{\mathbf{Z}}^\top\mathbf{H})
    \propto\widehat{\operatorname{HSIC}}(P_{SY},\mathcal{H}_S,\mathcal{H}_Y).
\end{aligned}
\end{equation}
It coincides with \eqref{eq:q_mglmm} when $\mathbf{K}=\mathbf{H}\mathbf{K}\mathbf{H}$ or $\mathbf{H}\tilde{\mathbf{Z}}=\tilde{\mathbf{Z}}$.

This provides a direct theoretical bridge to non-parametric multivariate spatial variability tests. For instance, \textcite{su2026mapping} formulate the isoform usage spatial variability test as an evaluation of dependency between the spatial coordinates and the empirical isoform ratios $\mathbf{x}_i \in \Delta^{p-1}$. Upon a global transformation \(\mathbf{z}_i = t(\mathbf{x}_i)\) (such as the isometric log-ratio), the \textbf{HSIC-IR} statistic is computed as
\begin{equation}\label{eq:hsic_ir}
    Q_{\text{HSIC-IR}} = \text{tr}(\mathbf{K} \mathbf{H} \mathbf{L}_X \mathbf{H}) = \text{tr}(\mathbf{K} \mathbf{H}  t({\mathbf{X}})t(\mathbf{X})^\top \mathbf{H}) = \text{tr}(\mathbf{Z}^\top \mathbf{H} \mathbf{K} \mathbf{H}  {\mathbf{Z}}) := \text{tr}(\mathbf{Z}_c^\top \mathbf{K} \mathbf{Z}_c),
\end{equation}
where \(\mathbf{Z}=t(\mathbf{X})\), \(\mathbf{L}_X=\mathbf{Z}\mathbf{Z}^\top\), and \(\mathbf{Z}_c=\mathbf{H}\mathbf{Z}\) is the column-centered transformed ratio matrix. The transformation is applied to each row and produces $p$ columns for the identity map or $q=p-1$ columns for the isometric log-ratio map. The isometric log-ratio transform requires positive components, so zero ratios need a specified replacement rule.

To compare \(Q_{\text{multinom-HSIC}}\) and \(Q_{\text{HSIC-IR}}\), consider the simplest case where \(t\) is the identity function. The HSIC-IR statistic then evaluates the unweighted, centered empirical ratios directly. Writing $\mathbf{K}_c=\mathbf{H}\mathbf{K}\mathbf{H}$,
\[
\begin{aligned}
    Q_{\text{HSIC-IR}}
    &=\operatorname{tr}(\mathbf{X}^\top\mathbf{K}_c\mathbf{X}) \\
    &=\operatorname{tr}\!\left[
    (\mathbf{X}-\mathbf{1}\bar{\mathbf{p}}^\top)^\top\mathbf{K}_c
    (\mathbf{X}-\mathbf{1}\bar{\mathbf{p}}^\top)\right]
    :=Q_X.
\end{aligned}
\]
The equality holds because $\mathbf{K}_c\mathbf{1}=\mathbf{0}$. Under the multinomial GLMM null, \eqref{eq:mv_q_moments} gives
\begin{align*}
    \mathbb{E}_0[Q_X]
    &=\sum_i(K_c)_{ii}\operatorname{tr}(\boldsymbol{\Sigma}_{X,i}), \\
    \operatorname{Var}_0(Q_X)
    &=2\sum_{i,j}(K_c)_{ij}^2\operatorname{tr}(\boldsymbol{\Sigma}_{X,i}\boldsymbol{\Sigma}_{X,j})
    +\sum_i(K_c)_{ii}^2\kappa_{X,i},
\end{align*}
where $\boldsymbol{\Sigma}_{X,i}=\boldsymbol{\Omega}_p/N_i+(1-1/N_i)\mathbf{B}$ and
\[
    \kappa_{X,i}=\mathbb{E}_0\|\mathbf{x}_i-\bar{\mathbf{p}}\|^4
    -\bigl(\operatorname{tr}\boldsymbol{\Sigma}_{X,i}\bigr)^2
    -2\operatorname{tr}(\boldsymbol{\Sigma}_{X,i}^2).
\]
For equal totals $N_i\equiv N$, the ratios share a common covariance, which can be estimated empirically or specified by the full GLMM,
\begin{enumerate}
    \item \textbf{Non-parametric:} $\widehat{\boldsymbol{\Sigma}}_X=(n-1)^{-1}\mathbf{X}^\top\mathbf{H}\mathbf{X}$.
    \item \textbf{Multinomial GLMM:} $\boldsymbol{\Sigma}_X=\boldsymbol{\Omega}_p/N+(1-1/N)\mathbf{B}$.
\end{enumerate}

The approaches differ in feature weighting and null covariance estimation. Higher counts reduce sampling noise while independent variation in isoform usage remains. Neither weighting scheme guarantees greater power under a correctly calibrated null model.

\newpage

\section{Functional and spectral analysis of \(Q\)-tests}\label{supnote:3}
Theorem~\ref{theorem:q_mean_dependence} shows that the population quadratic form measures spatial mean dependence, while Corollary~\ref{corollary:q_test_mean_independent} establishes that \(Q\)-tests cannot reliably detect changes confined to variance or higher moments even as sample size increases. In destructive spatial omics experiments, each location is typically measured only once, leaving its conditional expression distribution $P(X\mid S)$ unobserved. Without further assumptions, observations from $n$ distinct locations cannot distinguish a spatially independent process $x_i\sim P(X)$ from a deterministic spatial pattern $x_i=f(s_i)$ fitted to those observations.

We therefore take a functional view of spatial patterns, representing the mean as a deterministic function. Kernel spectra characterize when nonconstant patterns yield a nonzero \(Q\)-statistic (Theorem~\ref{theorem:q_functional_consistency}) and guide the design of tests sensitive to different spatial features.

\subsection{A functional view of spatial patterns}\label{supidx:n3:functional}

We begin by formally defining spatial patterns as functions.
\begin{definition}[Spatial pattern]
Let $\mathcal S$ be a compact metric space with sampling distribution $P_S$. A spatial pattern is a function $f\in L^2(P_S)$, with
\[
    \int_{\mathcal S}|f(s)|^2\,dP_S(s)<\infty,
    \qquad
    \langle f,g\rangle
    =\int_{\mathcal S}f(s)g(s)\,dP_S(s).
\]
\end{definition}

We interpret $f(s)$ as the conditional mean at location $s$. The null corresponds to a constant function $P_S$-almost everywhere, and the alternative to a mean that varies over the sampled domain.

\begin{definition}[Spatial kernel operators]
Let $k$ be symmetric and satisfy
\[
    \iint |k(s,s')|^2\,dP_S(s)\,dP_S(s')<\infty.
\]
The spatial kernel operator and centering operator are
\[
    (\mathcal Kf)(s)
    =\int k(s,s')f(s')\,dP_S(s'),
    \qquad
    \mathcal Cf=f-\int f\,dP_S.
\]
The centered kernel operator is $\mathcal T=\mathcal C\mathcal K\mathcal C$.
\end{definition}

These assumptions make $\mathcal K$ and $\mathcal T$ compact, self-adjoint operators on $L^2(P_S)$.

\begin{definition}[Functional \(Q\)-statistic]
For a spatial pattern $f$, define
\[
    Q(f)
    =\langle\mathcal Cf,\mathcal K\mathcal Cf\rangle
    =\langle f,\mathcal Tf\rangle.
\]
\end{definition}

To connect the functional statistic to observed data, we evaluate $f$ at sampled locations. Functions in $L^2(P_S)$ are identified when they agree outside a set of zero sampling probability. Choosing a measurable representative therefore does not affect evaluations at locations drawn from $P_S$.

\begin{definition}[Finite-sample pattern, kernel and \(Q\)-statistic]
For locations $s_1,\ldots,s_n$, let
\[
    x_i=f(s_i),\qquad
    (\mathbf K_n)_{ij}=k(s_i,s_j),\qquad
    Q_n(\mathbf x)=\mathbf x^\top
        \mathbf H\mathbf K_n\mathbf H\mathbf x.
\]
If $\mathbf z=\mathbf H\mathbf x/\hat\sigma$, then $Q_n(\mathbf x)=\hat\sigma^2\mathbf z^\top\mathbf K_n\mathbf z$.
\end{definition}

\begin{proposition}[\(Q\)-statistic convergence]
Let $S_1,S_2,\ldots$ be independent draws from $P_S$. For $f\in L^2(P_S)$ and a bounded measurable symmetric kernel $k$, set $x_i=f(S_i)$. Then
\[
    \frac{Q_n(\mathbf x)}{n^2}
    \xrightarrow{\mathrm{a.s.}}
    \iint \bar f(s)k(s,s')\bar f(s')\,
        dP_S(s)\,dP_S(s')
    =Q(f),
\]
where $\bar f=\mathcal Cf$.
\end{proposition}

This follows from the strong law for V-statistics and convergence of the sample mean to $\int f\,dP_S$. For predefined grids, the same limit holds when the empirical distribution of locations converges weakly to $P_S$ and $f$ and $k$ are bounded and continuous.

\subsection{Spectral properties of kernel operators}\label{supidx:n3:spectra}
The kernel spectrum determines how each component of a pattern contributes to the quadratic form. Together with the null calibration, it also determines the detection properties studied below.

\begin{lemma}[Spectral decomposition of $\mathcal K$]\label{lemma:spectrum_k}
The symmetric, square-integrable kernel defines a compact self-adjoint operator on $L^2(P_S)$. Its nonzero eigenvalues are real and can be listed, with multiplicity, in decreasing order of absolute value. If this list is infinite, the eigenvalues approach zero. The corresponding orthonormal eigenfunctions satisfy
\[
    \mathcal Kf
    =\sum_j\mu_j\langle f,\phi_j\rangle\phi_j,
\]
with convergence in $L^2(P_S)$. Together with an orthonormal basis of $\ker\mathcal K$, they form a basis of $L^2(P_S)$.
\end{lemma}

\begin{proposition}[Spectral decomposition of $\mathcal T$]\label{prop:spectrum_centered}
The centered operator $\mathcal T=\mathcal C\mathcal K\mathcal C$ is compact and self-adjoint. On the subspace of zero-mean functions,
\[
    L^2_0(P_S)
    =\left\{f\in L^2(P_S):\int f\,dP_S=0\right\},
\]
it is the compression of $\mathcal K$ to that subspace. Its spectral properties are as follows:
\begin{enumerate}
    \item \textbf{Nullspace:}
    The constant function $\mathbf1$ lies in the nullspace of $\mathcal T$, with eigenvalue zero.

    \item \textbf{Eigenvalue interlacing:}
    For a finite symmetric kernel matrix $\mathbf K$, let
    $\mu_1\ge\cdots\ge\mu_n$ be its eigenvalues and
    $\lambda_1\ge\cdots\ge\lambda_{n-1}$ those of
    $\mathbf H\mathbf K\mathbf H$ restricted to $\mathbf1^\perp$.
    The min--max principle gives
    \[
        \mu_j\ge\lambda_j\ge\mu_{j+1},
        \qquad j=1,\ldots,n-1.
    \]
    The full centered matrix also has the zero eigenvalue of the constant vector. These inequalities use ordering by value, not absolute value.

    \item \textbf{Translation invariance:}
    For a translation-invariant kernel on a periodic domain with uniform $P_S$, the constant function is an eigenfunction of $\mathcal K$.
    More generally, whenever $\mathcal K\mathbf1=\mu_c\mathbf1$, centering preserves every eigenpair orthogonal to $\mathbf1$ and replaces the constant-mode eigenvalue by zero.
    Thus, restricting to mean-zero functions removes only the constant mode.
\end{enumerate}
\end{proposition}

\begin{proof}
The first property follows from $\mathcal C\mathbf1=0$, and the second from the min--max principle for compression to a subspace of codimension one. For the third, translation invariance and uniform sampling make
\[
    (\mathcal K\mathbf1)(s)
    =\int k(s,s')\,dP_S(s')
\]
independent of $s$. If $\mathcal K\phi=\mu\phi$ and $\phi\perp\mathbf1$, then
\[
    \mathcal T\phi
    =\mathcal C\mathcal K\mathcal C\phi
    =\mathcal C(\mu\phi)
    =\mu\phi.
\]
Hence all eigenpairs orthogonal to the constant function are preserved.
\end{proof}

Centering removes sensitivity to a constant mean. For translation-invariant kernels on periodic domains, the Fourier representation below makes the remaining modes explicit.

\begin{lemma}[Spectrum of stationary kernels]
\label{lemma:spectrum_convolution}
Suppose $\mathcal S$ is a periodic domain of side length $L$, $P_S$ is uniform, and $k(s,s')=\psi(s-s')$. The Fourier modes
\[
    \phi_{\mathbf m}(s)
    =e^{i\boldsymbol\omega_{\mathbf m}\cdot s},
    \qquad
    \boldsymbol\omega_{\mathbf m}
    =\frac{2\pi}{L}\mathbf m,
    \quad \mathbf m\in\mathbb Z^d,
\]
are orthonormal in $L^2(P_S)$ and diagonalize $\mathcal K$. Their eigenvalues are
\[
    \lambda_{\mathbf m}
    =\frac1{L^d}\int_{\mathcal S}
      \psi(u)e^{-i\boldsymbol\omega_{\mathbf m}\cdot u}\,du.
\]
\end{lemma}

\textbf{Remark.}
Periodicity is an additional assumption, not a consequence of compactness. If $\psi$ is obtained by periodizing an integrable kernel $\psi_{\mathbb R}$ on $\mathbb R^d$, then
\[
    \lambda_{\mathbf m}
    =\frac1{L^d}
      \widehat\psi_{\mathbb R}
      (\boldsymbol\omega_{\mathbf m}).
\]
The factor $L^{-d}$ comes from averaging with respect to $P_S$. As $L$ increases, the allowed frequencies become more closely spaced. With nonuniform sampling, the operator also depends on the sampling density, and Fourier modes need not be its eigenfunctions.

\subsection{Consistency conditions of \(Q\)-tests on spatial patterns}\label{supidx:n3:consistency}
We now ask when a \(Q\)-test detects every nonconstant mean pattern with probability approaching one as sample size increases. Corollary~\ref{corollary:q_test_mean_shift} connects this question to the population quadratic form. Writing $f(s)=\mathbb{E}[X\mid S=s]-\mathbb{E}[X]$, the random-location sampling assumptions introduced in Supplementary Note~\ref{supnote:2.formulation} give
\[
    \frac{Q_n}{n^2}\xrightarrow{p}
    \frac{Q(f)}{\operatorname{Var}(X)},
    \qquad
    Q(f)=\langle f,\mathcal{T}f\rangle
    =\sum_j\lambda_j|\langle f,\phi_j\rangle|^2.
\]
Here, $Q_n$ uses empirical standardization and $0<\operatorname{Var}(X)<\infty$. A positive functional signal therefore guarantees consistency of an upper-tail test whenever its critical value grows more slowly than $n^2$.

\begin{theorem}[Spectral conditions for consistent mean-pattern detection]
\label{theorem:q_functional_consistency}
Consider a fixed bounded continuous symmetric kernel on the compact spatial domain, under the random-location sampling assumptions above. Suppose the test rejects when $Q_n>c_{n,\alpha}$, with $c_{n,\alpha}/n\xrightarrow{p}c_\alpha\in(0,\infty)$ as in \eqref{eq:q_critical_rate}. The test is consistent against every fixed alternative with a nonconstant conditional mean if and only if the centered operator $\mathcal{T}$ is strictly positive definite on $L^2_0(P_S)$.

A strictly negative definite operator gives the corresponding result for an upper-tail test of $-Q_n$, provided its critical value satisfies the same condition.
\end{theorem}
\begin{proof}
\textit{Sufficiency.} Strict positivity gives $Q(f)>0$ for every nonzero $f\in L^2_0(P_S)$. Since $c_{n,\alpha}/n^2\xrightarrow{p}0$, the argument of Corollary~\ref{corollary:q_test_mean_shift} gives $\Pr(Q_n>c_{n,\alpha})\to1$.

\textit{Necessity.} If $\mathcal{T}$ has a negative eigenvalue, choose $f$ along its eigenfunction and add independent Gaussian observation noise. Then $Q(f)<0$, so the upper-tail rejection probability approaches zero.

Otherwise, $\mathcal{T}$ is nonnegative. If it is not strictly positive, choose a nonzero $f$ in its null space and again add independent Gaussian noise. The population-centered kernel is positive semidefinite on the support of $P_S$ and has a centered feature map $\psi$. Since
\[
    \left\|\mathbb{E}[f(S)\psi(S)]\right\|^2=Q(f)=0,
\]
the mean kernel embedding vanishes despite the nonconstant conditional mean. The same central limit argument used in Corollary~\ref{corollary:q_test_mean_independent} gives $Q_n/n\xrightarrow{d}\|G\|^2$ for a zero-mean Gaussian element $G$. This limit assigns positive probability below every positive threshold. Consequently, the limiting rejection probability is less than one.
\end{proof}

\textbf{Remark.}
Theorem~\ref{theorem:q_functional_consistency} concerns a fixed continuous kernel and its population operator, with a specified rejection tail and null calibration such as \eqref{eq:v_stats_null}. It does not automatically extend to changing bandwidths or arbitrary sequences of finite-sample graph kernels. Indeed, positive definiteness of a finite kernel matrix alone does not guarantee detection. For example, the identity kernel $\mathbf K_n=\mathbf I$ gives $Q_n=n-1$ for every empirically standardized pattern, thus yielding an uninformative statistic. We explain this distinction in the subsequent finite-sample analysis section by examining spectral weights relative to the null threshold.

The spectral expansion also explains two ways a nonconstant pattern can have $Q(f)=0$. A zero-weight eigenfunction lies in a \textbf{blind spot}. Opposite-sign eigenvalues allow \textbf{cancellation}. That is, if $\lambda_p>0>\lambda_q$, then
\[
    f=\sqrt{|\lambda_q|}\phi_p+\sqrt{\lambda_p}\phi_q
    \quad\Longrightarrow\quad Q(f)=0.
\]
Thus $Q(f)\neq0$ for every nonzero mean-zero function if and only if $\mathcal{T}$ is strictly definite.

\subsection{Finite-sample power and spectral cancellation}\label{supnote:3.finite_power}
At fixed sample size $n$ and fixed spatial locations, we can study \(Q\)-test power directly using the kernel matrix. Let $\widetilde{\mathbf{K}}=\mathbf{H}\mathbf{K}_n\mathbf{H}$ have eigenpairs $(\lambda_j,\mathbf{u}_j)$ on $\mathbf{1}^{\perp}$, including zero eigenvalues. Consider a centered mean pattern $\mathbf{f}\neq\mathbf{0}$ observed with independent Gaussian noise,
\[
    \mathbf{x}=b\mathbf{1}+t\mathbf{f}+\sigma\boldsymbol{\epsilon},
    \qquad
    \boldsymbol{\epsilon}\sim N(\mathbf{0},\mathbf{I}),
    \quad \sigma>0.
\]
Here, $t\geq0$ controls pattern strength. With empirical standardization and $m=n-1$,
\begin{equation}\label{eq:q_finite_spectral}
    Q_n
    =m\frac{\mathbf{x}^{\top}\widetilde{\mathbf{K}}\mathbf{x}}
             {\mathbf{x}^{\top}\mathbf{H}\mathbf{x}}
    \overset{d}{=}
    m\frac{\sum_{j=1}^{m}\lambda_j(\xi_j+t a_j)^2}
             {\sum_{j=1}^{m}(\xi_j+t a_j)^2},
    \qquad
    a_j=\frac{\mathbf{u}_j^{\top}\mathbf{f}}{\sigma},
\end{equation}
where the $\xi_j$ are independent standard normal variables. Setting $t=0$ gives the finite-sample Gaussian null. The denominator remains random because the data are standardized by their empirical variance.

For a fixed upper-tail cutoff $c_{n,\alpha}$ computed from the kernel, the finite-sample power is
\begin{equation}\label{eq:q_finite_power}
    \beta_n(t,\mathbf{f})
    =
    \Pr\!\left[
        \sum_{j=1}^{m}
        (m\lambda_j-c_{n,\alpha})(\xi_j+t a_j)^2>0
    \right].
\end{equation}
This is a tail probability of a weighted sum of independent noncentral chi-square variables, which can be evaluated numerically\autocite{imhof1961computing}. The formula applies to the chosen cutoff, including an approximate one. It shows how power depends on the full spectrum, the pattern's allocation across modes and the null threshold. For any given pattern $\mathbf{f}$, \(t^2a_j^2\) is the squared signal-to-noise ratio in eigenmode \(j\). Modes with \(\lambda_j>\frac{1}{m}c_{n,\alpha}\) contribute toward rejection, whereas those below this threshold oppose it. Assigning larger relative weights to the former modes can therefore improve power, but the exact benefit also depends on how reweighting changes the cutoff \(c_{n,\alpha}\) (which is itself a function of the kernel's spectrum).

\begin{corollary}[Spectral cancellation can suppress arbitrarily strong patterns]\label{corollary:finite_sample_zero_power}
If $\widetilde{\mathbf{K}}$ is indefinite or singular on $\mathbf{1}^{\perp}$, some nonzero centered pattern $\mathbf{f}$ satisfies $\mathbf{f}^{\top}\widetilde{\mathbf{K}}\mathbf{f}=0$. Under the observation model above, at fixed sample size,
\[
    Q_n\xrightarrow{p}0
    \qquad\text{as }t\to\infty.
\]
Consequently, an upper-tail test with a fixed cutoff $c_{n,\alpha}>0$ has $\beta_n(t,\mathbf{f})\to0$. The same conclusion holds for a two-sided test whose fixed acceptance interval contains zero in its interior.
\end{corollary}
\begin{proof}
A nonzero null-space vector has zero quadratic form. If $\lambda_p>0>\lambda_q$, the vector $\mathbf{f}=\sqrt{|\lambda_q|}\mathbf{u}_p+\sqrt{\lambda_p}\mathbf{u}_q$ also has zero quadratic form. Dividing the numerator and denominator of \eqref{eq:q_finite_spectral} by $t^2$ gives
\[
    Q_n\xrightarrow{p}
    (n-1)\frac{\mathbf{f}^{\top}\widetilde{\mathbf{K}}\mathbf{f}}
                   {\|\mathbf{f}\|^2}=0.
\]
If zero lies inside the acceptance region, rejection probability therefore approaches zero.
\end{proof}

For an exactly cancelled pattern, empirical standardization makes the \(Q\)-statistic converge to zero as the signal-to-noise ratio increases. If zero lies inside the acceptance region, the probability of missing the pattern therefore approaches one.

\subsection{Spectral limitations of SVG methods}\label{supnote:3.examples}

We next use this framework to examine blind spots and cancellation in spatial variability tests.

\textbf{Example 1: SPARK-X}\autocite{zhu2021spark}
\texttt{SPARK-X} replaces \texttt{SPARK}'s parametric count model and iterative fitting\autocite{sun2020statistical} with a covariance test on empirically standardized expression. This avoids the fitting failures reported for highly sparse counts. Our GLM analysis identifies a related difference in scaling: \texttt{SPARK}'s score test approximately uses a fitted null variance, whereas \texttt{SPARK-X} uses the empirical variance (Supplementary Note~\ref{supnote:2.glm}). While the two can differ when spatial signal contributes to the empirical variance, neither scaling choice guarantees greater power, as test sensitivity also depends on abundance, dispersion, calibration and the spatial kernel.

On spatial modeling, \texttt{SPARK-X} adopts a low-rank representation for memory and runtime savings. For analyzing its shared feature space, the test can be viewed as using a sum of 11 kernels of rank at most two,
\[
    k_\Sigma(s,s')
    =\sum_{a=1}^{11}
      \langle\phi_a(s),\phi_a(s')\rangle,
    \qquad \phi_a:\mathbb R^2\to\mathbb R^2.
\]
This representation captures the common blind spots of its 11 individual component statistics: because they are nonnegative, their sum, average and maximum all vanish exactly when every component vanishes. Their combined feature span has dimension at most 22, so any centered observation vector orthogonal to this span receives zero from every component. In the functional view, retaining a fixed finite set of features leaves mean patterns outside that span undetected, preventing consistent detection of every mean alternative. Related power limitations were previously demonstrated by \textcite{su2026mapping}.

Similar restrictions arise in other finite spatial representations. For instance, \texttt{C-SIDE}\autocite{cable2022cell} tests effects represented by supplied spatial covariates $\mathbf L\in\mathbb R^{n\times d}$, typically with $d\ll n$. Fixed-rank Nystr\"om approximations and fixed landmark grids likewise omit spatial directions. The default 100-point representation in \texttt{singleCellHaystack}\autocite{vandenbon2020clustering} provides another example (Supplementary Note~\ref{supnote:2.singleCellHaystack}). When approximation rank increases with sample size, consistency instead depends on the retained signal and calibration.

\textbf{Example 2: Moran's \(I\)}\autocite{moran1948interpretation}
The spatial kernel is the adjacency matrix $\mathbf{K}_n = \mathbf{W}$ of a neighborhood graph. For a continuum illustration, consider convolution with the ball indicator $h(u)=\mathbb{I}(\|u\|<r)$ on $\mathbb{R}^d$. Under the Fourier convention $\widehat h(\omega)=\int h(u)e^{-i\omega^\top u}\,du$, its spectral response is
\[
    \widehat h(\omega) = \left(\frac{2\pi r}{\|\omega\|}\right)^{d/2} J_{d/2}(r \|\omega\|),
\]
where $J_{d/2}$ is the Bessel function of the first kind and the value at zero is defined by continuity. Since the Bessel function oscillates around zero, this Fourier transform contains both positive values (clustering) and negative values (dispersion/checkerboard patterns), as well as roots where $\widehat h(\omega)=0$. This illustrates how opposite spectral weights permit cancellation. For a finite neighborhood graph, the relevant spectrum is that of its centered adjacency matrix. Composite patterns mixing clustering and dispersion can therefore receive a score near zero.

At fixed sample size, Corollary~\ref{corollary:finite_sample_zero_power} shows how exact cancellation can conceal even strong patterns from Moran's \(I\), including under a two-sided test when zero lies inside its acceptance interval. The same mechanism can affect the quadratic components of other graph-based methods, including \texttt{Hotspot}\autocite{detomaso2021hotspot}, \texttt{MERINGUE}\autocite{miller2021characterizing} and \texttt{BinSpect}\autocite{dries2021giotto}. One caveat is that the relevant spectrum is that of the centered adjacency matrix $\mathbf{H}\mathbf{W}\mathbf{H}$. Its eigenvalues interlace those of $\mathbf{W}$ (Proposition~\ref{prop:spectrum_centered}), so at least two eigenvalues of each sign in $\mathbf{W}$ guarantee that both signs survive centering. Accordingly, the complete graph with a centered adjacency matrix $-\mathbf{H}$ (no positive eigenvalues) is not susceptible to cancellation.

\textbf{Example 3: Kernel sparsification by distance thresholding}
To reduce the \(\mathcal{O}(n^2)\) memory footprint of dense kernel matrices, a common computational heuristic is to sparsify the matrix by hard-thresholding: setting \(K_{ij} = 0\) if the spatial distance \(\|s_i - s_j\| > r\). 
In the spatial domain, this thresholding operation is equivalent to the Hadamard (element-wise) product of the original dense kernel \(k_\text{dense}(s, s')\) and a binary indicator mask \(h(s, s') = \mathbb{I}(\|s-s'\| \le r)\). For stationary kernels on a periodic domain, the Fourier coefficients of this product are the convolution of the component Fourier coefficients,
\[
\mu_\text{sparse}(\omega) = (\mu_\text{dense} * \mu_h)(\omega).
\]
While the original dense kernel spectrum \(\mu_\text{dense}(\omega)\) may be strictly positive, the spectrum of the binary threshold mask \(\mu_h(\omega)\) oscillates with negative lobes, as demonstrated by the Bessel function in \textbf{Example 2}. Consequently, the sparsified kernel can become indefinite. If both signs remain after centering, spectral cancellation is possible. 

A positive semidefinite mask preserves positive semidefiniteness by the Schur product theorem. One such alternative to hard thresholding is the compactly supported Wendland function
\[
  h(d_{ij}) = 
  \begin{cases} 
  \left(1 - \frac{d_{ij}}{r}\right)^4 \left(4 \frac{d_{ij}}{r} + 1\right) & \text{if } d_{ij} < r \\
  0 & \text{if } d_{ij} \ge r
  \end{cases}
\]
valid for 2D and 3D space.

\subsection{Kernel design via spectral weighting}\label{supnote:3.kernel_design}
Kernel spectral weights determine how each spatial mode contributes to the statistic, and their precise effect on finite-sample power also depends on the noise and null threshold, as shown by \eqref{eq:q_finite_power}. Nevertheless, we can explore design principles by focusing on the large-sample limits. Following Lemma~\ref{lemma:spectrum_convolution}, we can decompose both the kernel and the spatial pattern into their frequency components, such that $Q(f)=\sum_{\omega} |\hat{f}_\omega|^2 \lambda(\omega)$. In this representation, the kernel spectrum $\lambda(\omega)$ acts as a ``frequency filter,'' determining the sensitivity profile of the test.

\subsubsection{No kernel is universally optimal}
Among the class of positive definite kernels, no single kernel is universally optimal for detecting all spatial patterns. To see this, observe that scaling the kernel spectrum by a constant factor scales the test statistic and the null distribution identically (\eqref{eq:v_stats_null}), leaving the $P$ value unchanged. Therefore, under a fixed trace constraint (e.g., $\sum \lambda_i = C$), increasing the spectral weight $\lambda(\omega)$ at specific frequencies to prioritize certain patterns necessitates decreasing the weight at others. These preferences can explain disagreements between consistent SVG tests, although calibration and the noise model also affect performance.

We compare three major kernel classes by their spectral response to spatial frequencies $|\omega|$.

\textbf{Example 4: Gaussian Kernel (soft low-pass filter).}
For the standard radial basis function $k(s, s') = \exp(-|s-s'|^2/2l^2)$, the spectrum is given by the Fourier transform of a Gaussian, which is itself a Gaussian
\[
\mu(\omega) \propto \exp(-l^2 \|\omega\|^2/2).
\] Since the spectral density is strictly positive for all $\omega$, the Gaussian kernel supports consistent detection of nonconstant mean patterns under the stated sampling and calibration assumptions. However, it acts as a soft low-pass filter with exponential decay. The spectral weighting is governed by the bandwidth parameter $l$. A large $l$ results in rapid spectral decay, placing greater weight on global, low-frequency trends and very little on local, high-frequency variation. Conversely, at fixed distinct locations, the kernel approaches the identity as $l\to0$. With empirical standardization, its limiting statistic is $\mathbf z^\top\mathbf z=n-1$ for every pattern, so it cannot distinguish spatial arrangements.

\textbf{Example 5: Graph Laplacian Kernel (high-pass filter).}
The unnormalized Graph Laplacian is $\mathbf{L}_n = \mathbf{D} - \mathbf{W}$. With suitable graph construction, rescaling and sampling assumptions, a graph Laplacian can converge to the negative Laplacian $-\Delta$ (Supplementary Note~\ref{supnote:1}). While technically not an integral kernel operator, we can still study $-\Delta$ using spectral theory. On a periodic domain, its Fourier eigenvalues grow quadratically with frequency,
\[
\lambda(\omega) \propto \|\omega\|^2.
\]
Consequently, it acts as a high-pass filter. The corresponding \(Q\)-statistic $Q_n=\mathbf z^\top\mathbf L_n\mathbf z$ is related to the Laplacian score\autocite{he2005laplacian}, with method-specific centering and normalization, and has been used for SVG testing\autocite{govek2019clustering}. It measures spatial ``roughness,'' placing greater weight on high-frequency variation than on smooth, global trends. High-frequency patterns may also indicate cell--cell interactions, as proposed by \textcite{maher2024harmonic}.

The Laplacian kernel is spectrally coupled to the adjacency kernel of Moran's \(I\). For a graph without isolated vertices, the normalized adjacency matrix $\mathbf{\tilde{W}} = \mathbf{D}^{-1/2}\mathbf{W}\mathbf{D}^{-1/2}$ has eigenvalues related to the normalized Laplacian $\mathbf{\tilde{L}}$ via $\lambda_i(\mathbf{\tilde{L}}) = 1 - \mu_i(\mathbf{\tilde{W}})$. Since $\mu_i(\mathbf{\tilde{W}}) \in [-1, 1]$, the Laplacian eigenvalues map to $\lambda_i(\mathbf{\tilde{L}}) \in [0, 2]$. While this linear shift resolves the indefiniteness of Moran's \(I\) (making the operator positive semidefinite), it amplifies the high-frequency ``checkerboard'' modes (where $\mu \approx -1 \implies \lambda \approx 2$) while suppressing the clustering modes (where $\mu \approx 1 \implies \lambda \approx 0$). With empirical standardization, the two statistics satisfy
\[
    Q_{\mathrm{Lap}}
    =\mathbf{z}^{\top}\mathbf{z}
     -\mathbf{z}^{\top}\widetilde{\mathbf{W}}\mathbf{z}
    =(n-1)-Q_{\mathrm{Moran}}.
\]
With corresponding null calibration, an upper-tail test of this normalized Laplacian statistic is therefore a lower-tail test of the normalized adjacency statistic. The affine transformation changes the rejection tail but cannot recover information lost by the original statistic.

\textbf{Remark.} Nonnegative functions of Laplacian eigenvalues define new positive semidefinite graph kernels. For example, the heat filter $e^{-t\mathbf L}$, with $t>0$, gives weights $e^{-t\lambda_j}$ and favors smooth graph patterns.

\textbf{Example 6: Regularized inverse Laplacian (polynomial low-pass filter).}
The regularized inverse-Laplacian (CAR) kernel is defined as $\mathbf{K} = (\mathbf{I} - \rho\mathbf{\tilde{W}})^{-1}$, corresponding to the covariance of a conditional autoregressive (CAR) process with autocorrelation $\rho \in (0,1)$ \autocite{su2023smoother}. Under a suitable regular-grid scaling, the corresponding continuum operator is the resolvent of the Laplacian operator, $(I - \alpha \Delta)^{-1}$. Its spectrum is given by
\[
    \lambda(\omega) \propto \frac{1}{1 + \alpha \|\omega\|^2}.
\]
Unlike the exponential decay of the Gaussian kernel, the CAR kernel exhibits polynomial decay (heavy-tailed). The formal Mat\'ern smoothness is $\nu=1-d/2$\autocite{lindgren2011explicit}, giving $\nu=1/2$ in one dimension. In two dimensions, $\nu=0$ is a generalized-field limit with divergent pointwise variance, while the finite-graph CAR matrix remains well defined. This heavy tail retains more relative spectral weight on intermediate and high-frequency patterns than the Gaussian kernel. Using the eigenvalue relation $\lambda_i(\mathbf{K}) = (1 - \rho \mu_i(\mathbf{\tilde{W}}))^{-1}$, we see that the CAR transforms the indefinite spectrum of $\mathbf{\tilde{W}}$ into a positive definite spectrum that prioritizes low-frequency clustering modes ($\mu \approx 1$) while damping high-frequency checkerboard modes ($\mu=-1\implies\lambda=1/(1+\rho)$, approaching $1/2$ as $\rho\to1$).

The limitation of Moran's \(I\) lies in its conflation of high and low frequencies through an oscillatory spectrum. Our analysis suggests that the adjacency spectrum gives rise to two different weightings, the high-pass Graph Laplacian (for roughness) and the low-pass Inverse Laplacian (for smoothness). One computational advantage of this approach is sparsity. The Graph Laplacian retains the sparsity structure of the adjacency matrix, reducing the memory complexity of kernel evaluation from $O(n^2)$ for dense kernels to $O(kn)$ for $k$-mutual-nearest-neighbors. CAR instead has a sparse precision matrix.

\subsubsection{Data-driven kernels do not guarantee better performance}\label{supnote:3.data_driven_kernel}
Data-driven approaches learn kernel parameters or structure from expression data, as proposed by \textcite{svensson2018spatialde,yan2025categorization}. Although hyperparameter tuning can improve GP predictions, maximizing a test statistic on the same data without accounting for selection amounts to ``double dipping'' and can inflate false-positive rates.
A sufficiently flexible gene-specific kernel can fit noise. Even selection of a single bandwidth, as in \texttt{SpatialDE}\autocite{svensson2018spatialde}, must be accounted for when calibrating significance.

Alternatively, one might optimize a set of global kernel hyperparameters across the entire dataset to avoid gene-specific overfitting. Global tuning still requires calibration for reuse of the tested data and can bias the test toward the dominant sources of variation in the data. If the dataset is plagued by technical noise or broad structural trends, the learned kernel will prioritize these features. Biologically distinct but rare patterns may be masked. Related preferences can arise in deep learning methods such as Convolutional Neural Networks (CNNs) and Graph Convolutional Networks (GCNs). For a linear periodic convolution layer, the learned filter determines $\mathbf K_{\mathrm{conv}}$ and its Fourier spectrum, as discussed in Supplementary Note~\ref{supnote:1.cnn}.

A prominent example is \texttt{SpaGCN} \autocite{hu2021spagcn}, which trains a GCN using an iterative unsupervised clustering objective and identifies SVGs via differential expression analysis between the resulting clusters. From a kernel testing perspective, this process is twofold: (1) The GCN layer aggregates data via a spatial adjacency matrix $\mathbf{A}$, acting as a low-pass filter that smooths expression between neighbors. (2) The resulting partition can be represented by the discrete kernel function $k(s, s') = \mathbb{I}(\text{domain}(s) = \text{domain}(s'))$ using the learned domain labels. When the same expression data are used to learn domains and test genes, calibration must account for this reuse\autocite{zhang2019valid,song2025synthetic}. Specifically, \texttt{SpaGCN}'s training objective learns domains from the input expression profiles, so domain-based testing can favor genes that align with the majority and drive the clustering.

A learned kernel reflects its training data and objective. Therefore, for general-purpose SVG discovery, a fixed, heavy-tailed kernel (like the CAR or Matérn) provides a practical default when suitable training data or controls are unavailable.

\newpage
\section{Scalable computation for detection and comparison}\label{supnote:eng}

Evaluating the spatial statistic \(Q_n = \mathbf{z}^\top \tilde{\mathbf{K}} \mathbf{z}\) with a dense \(n \times n\) kernel costs \(\mathcal{O}(n^2)\) time and memory per feature. Full eigendecomposition for null calibration adds \(\mathcal{O}(n^3)\) work. For spatial omics datasets with \(n \sim 10^3\)--\(10^7\) locations and thousands of features, these costs become prohibitive.

We first describe null approximations and structured kernel operations that reduce these costs, then extend the computations to bivariate \(R\)-statistics, Gaussian-process fitting for conditional differential usage, and frequency-based pattern comparison. Detection and comparison are implemented in \texttt{sonic-spatial} (\url{https://github.com/JiayuSuPKU/sonic}), with the Gaussian-process routines used in \texttt{SPLISOSM}\autocite{su2026mapping}. Throughout this note, $\tilde{\mathbf{K}}=\mathbf{H}\mathbf{K}\mathbf{H}$ denotes the centered kernel, and $c_p=\operatorname{tr}(\tilde{\mathbf{K}}^p)$ is its $p$-th spectral power sum. These quantities determine the Gaussian quadratic-form cumulants and support the null approximations below.

\subsection{Null-distribution approximations}\label{supidx:n4:null}

The finite-sample null distribution depends on how expression is standardized. For independent Gaussian observations with a common mean and known variance $\sigma^2$, centering expression gives
\[
    \frac{\mathbf{x}^{\top}\tilde{\mathbf{K}}\mathbf{x}}{\sigma^2}
    \overset{d}{=}
    \sum_{j=1}^{n-1}\lambda_j\xi_j^2,
    \qquad
    \xi_j\stackrel{\mathrm{iid}}{\sim}\mathcal{N}(0,1),
\]
where $\lambda_j$ are the eigenvalues of $\tilde{\mathbf{K}}$ on the subspace orthogonal to $\mathbf{1}$. This weighted chi-square law is exact under the Gaussian null. The V-statistic limit in Supplementary Note~\ref{supnote:2} provides a separate large-sample justification under its stated sampling and kernel assumptions.

When expression is standardized by its empirical variance,
\[
    \hat\sigma^2=\frac{\mathbf{x}^{\top}\mathbf{H}\mathbf{x}}{n-1},
\]
the null statistic instead has the ratio distribution
\[
    Q_n\overset{d}{=}
    (n-1)\,
    \frac{\sum_{j=1}^{n-1}\lambda_j\xi_j^2}
    {\sum_{j=1}^{n-1}\xi_j^2}.
\]
The normalized squared coefficients follow a Dirichlet distribution with parameters $1/2$, giving
\[
    \mathbb{E}[Q_n]=c_1,
    \qquad
    V_Q\coloneqq\operatorname{Var}(Q_n)
    =\frac{2[(n-1)c_2-c_1^2]}{n+1}.
\]
These moments are exact for empirically standardized Gaussian observations. For non-Gaussian data, this variance correction is an approximation rather than a distribution-free identity.

Given the eigenvalues, Davies's method evaluates weighted chi-square probabilities by numerical inversion\autocite{davies1980algorithm}. For the empirically standardized statistic, its upper-tail probability can also be written as
\[
    \Pr(Q_n>q)
    =\Pr\!\left[
        \sum_{j=1}^{n-1}\bigl((n-1)\lambda_j-q\bigr)\xi_j^2>0
    \right].
\]
Thus, evaluating the ratio law requires threshold-dependent weights. Applying Davies's method directly to the original $\lambda_j$ evaluates the known-variance quadratic form.

To avoid repeated numerical inversion when screening many features, we use moment-based approximations:
\begin{itemize}
    \item \textbf{Welch--Satterthwaite and normal approximations.} For a positive semidefinite kernel with nonzero null variance, Welch--Satterthwaite matches the mean and variance to a scaled chi-square distribution\autocite{welch1938significance,satterthwaite1941synthesis}. The normal approximation uses the same two moments and also accommodates indefinite kernels. Our implementation uses the finite-sample variance $V_Q$ above by default.

    \item \textbf{Four-moment approximations.} Liu's method uses the first four spectral power sums $c_1,\ldots,c_4$ to construct a shifted, scaled noncentral chi-square approximation\autocite{liu2009new}. For the Gaussian quadratic form, the corresponding cumulants are
    \[
        \kappa_p=2^{p-1}(p-1)!\,c_p.
    \]
    For empirically standardized data, calibration uses the moments of the ratio distribution. Let $d=n-1$, $\mathbf K_\circ=\tilde{\mathbf K}-(c_1/d)\mathbf H$, and $b_p=\operatorname{tr}(\mathbf K_\circ^p)$. In addition to the mean and variance above, its third and fourth cumulants under the Gaussian null are
    \[
    \begin{aligned}
        \kappa_3&=\frac{8d^2b_3}{(d+2)(d+4)},\\
        \kappa_4&=\frac{48d^3}{(d+2)(d+4)(d+6)}
          \left[b_4-\frac{2(d+3)b_2^2}{d(d+2)}\right].
    \end{aligned}
    \]
    The implementation uses a Liu fit when its parameters are admissible, otherwise trying a beta fit, with central chi-square or normal fallbacks. These moment-based approximations use upper-tail probabilities. The default normal test for Moran's \(I\) uses two-sided probabilities.
\end{itemize}

The spectral power sums can be calculated from the eigenvalues or estimated using Hutchinson probes\autocite{hutchinson1989stochastic},
\[
    c_p=\mathbb{E}[\mathbf{v}^{\top}\tilde{\mathbf{K}}^p\mathbf{v}]
    \approx\frac1B\sum_{b=1}^B
    \mathbf{v}_b^{\top}\tilde{\mathbf{K}}^p\mathbf{v}_b,
\]
where the $B$ probe vectors have independent entries taking values $-1$ and $1$ with equal probability. The FFT implementation obtains these sums directly from the spectrum. Other backends use available trace formulas or stochastic estimates.

\textbf{Remark.} Approximation accuracy depends on the spectrum, the standardization and the tail probability being evaluated. For Gaussian quadratic forms, the condition
\[
    \frac{\max_j|\lambda_j|}{\sqrt{\sum_j\lambda_j^2}}
    \longrightarrow0
\]
ensures that, after subtracting its null mean and dividing by its null standard deviation, the statistic approaches a standard normal distribution as sample size grows. The condition means that no single eigenmode contributes a substantial fraction of the total null variance. For the empirically standardized ratio, the corresponding condition applies to the eigenvalue deviations $\lambda_j-c_1/(n-1)$. Rapid spectral decay can leave a few modes dominant and therefore does not generally justify a normal approximation. Nor is there a universal accuracy ranking among Liu, Welch--Satterthwaite and normal approximations. When $(n-1)c_2$ and $c_1^2$ nearly cancel, trace-estimation errors can also strongly affect the corrected variance. These calibration questions are separate from the consistency conditions in Supplementary Note~\ref{supnote:2}.

\subsection{Structured and sparse kernel operators}\label{supidx:n4:operators}

To achieve sub-quadratic memory and computational scaling for large-scale spatial datasets, the dense kernel matrix \(\tilde{\mathbf{K}}\) is typically replaced by structured or sparse operators. For arbitrary spatial coordinate clouds \(\{s_i\}\), three strategies are widely adopted:

\begin{itemize}
    \item \textbf{Sparse kernels:} By enforcing sparsity directly on the kernel (e.g., utilizing the \(k\)-nearest spatial graph adjacency \(\mathbf{W}\) in Moran's \(I\) or applying a hard distance threshold), the operator contains only \(\mathcal{O}(nk)\) non-zero entries, where \(k\) is the average number of neighbors. Consequently, matrix-vector products (matvecs) and Hutchinson trace estimations scale efficiently as \(\mathcal{O}(nk)\). However, as demonstrated in Supplementary Note~\ref{supnote:3.examples} (Examples 2 and 3), many ad-hoc sparse kernels are indefinite. When both signs remain after centering, cancellation can reduce power for composite patterns.

    \item \textbf{Low-rank kernels:} By approximating the kernel as \(\mathbf{K} \approx \mathbf{U}\mathbf{U}^\top\) with \(\mathbf{U} \in \mathbb{R}^{n \times r}\) (\(r \ll n\)), matvecs scale as \(\mathcal{O}(nr)\), and the spectrum can be derived efficiently via the \(\mathcal{O}(r^3)\) eigendecomposition of the smaller \(r \times r\) matrix \(\mathbf{U}^\top\mathbf{U}\), after forming it in \(O(nr^2)\) time. The one-time low-rank approximation, naively \(\mathcal{O}(n^2r)\) for a top-\(r\) eigendecomposition, can often be accelerated through Nystr{\"o}m methods\autocite{williams2000using} or random Fourier features\autocite{rahimi2007random}, with construction costs depending on the method. A fixed finite feature space leaves blind spots, while consistency with increasing rank depends on the retained signal and calibration (Example 1 in Supplementary Note~\ref{supnote:3.examples}).

    \item \textbf{Sparse precision (CAR):} Conditional Autoregressive (CAR) models define the spatial relationship through a sparse precision matrix \(\mathbf{\Omega} = \mathbf{I} - \rho\widetilde{\mathbf{W}}\), where \(\widetilde{\mathbf{W}}\) is a normalized adjacency matrix and \(\rho\) controls the strength of spatial autocorrelation. To avoid realizing the dense kernel \(\mathbf{K} = \mathbf{\Omega}^{-1}\), we cache a sparse LU (or Cholesky) factorization of \(\mathbf{\Omega}\). For suitable two-dimensional lattice-like graphs with nested-dissection ordering, this factorization incurs a one-time complexity of \(\mathcal{O}(n^{3/2})\) time and \(\mathcal{O}(n \log n)\) memory. Subsequent matvec operations (\(\mathbf{K}\mathbf{v} = \mathbf{\Omega}^{-1}\mathbf{v}\)) are efficiently evaluated via sparse triangular solves in \(\mathcal{O}(n \log n)\) time. For general graphs, costs depend on factorization fill-in. This approach preserves both positive-definiteness and full rank without the memory bottleneck of dense matrices.
\end{itemize}

\subsection{Fourier-accelerated methods}\label{supidx:n4:fourier}

Fourier methods accelerate kernels that depend on spatial displacement. On a complete periodic grid, the discrete Fourier transform diagonalizes these kernels exactly. For irregular coordinates, a finite Fourier expansion provides an approximation evaluated using the non-uniform fast Fourier transform (NUFFT).

\textbf{Regular grids (FFT):}
Consider a complete periodic $n_y\times n_x$ grid with $n=n_y n_x$ locations. Let $\mathcal F$ denote the unnormalized two-dimensional discrete Fourier transform and $\mathbf F\in\mathbb C^{n\times n}$ its matrix representation on vectorized grid arrays. For $\mathbf v\in\mathbb R^n$,
\[
    \mathcal F(\mathbf v)=\mathbf F\mathbf v\in\mathbb C^n,
    \qquad
    \mathbf F^{\mathsf H}\mathbf F=n\mathbf I_n,
    \qquad
    \mathcal F^{-1}(\mathbf w)=\frac{1}{n}\mathbf F^{\mathsf H}\mathbf w,
\]
where $\mathbf w\in\mathbb C^n$ and $\mathsf H$ denotes conjugate transpose.

A real symmetric kernel that is translation-invariant on this periodic grid has the decomposition
\[
    \mathbf K
    =\frac{1}{n}\mathbf F^{\mathsf H}
      \operatorname{diag}(\boldsymbol\mu)\mathbf F
    \in\mathbb R^{n\times n},
\]
where $\boldsymbol\mu\in\mathbb R^n$ contains its eigenvalues. These can be obtained by transforming the kernel's first column or specified directly to define a periodic kernel. Matrix--vector multiplication then becomes
\[
    \mathbf K\mathbf v
    =\mathcal F^{-1}\!\left[
        \boldsymbol\mu\odot\mathcal F(\mathbf v)
      \right]\in\mathbb R^n.
\]
For standardized observations $\mathbf z\in\mathbb R^n$,
\[
    Q_n=\mathbf z^\top\mathbf K\mathbf z
    =\frac{1}{n}\sum_{k=1}^n
      \mu_k\left|[\mathcal F(\mathbf z)]_k\right|^2.
\]
Both calculations require $\mathcal O(n\log n)$ operations without forming $\mathbf K$. Centering removes the constant Fourier mode on the complete periodic grid. For a tissue mask, SONIC restricts the kernel to observed locations and centers it on that subset. Missing raster entries are mean-filled for FFT evaluation. If the raster contains $N$ entries but only $n$ observed locations, full-raster standardization scales the reported statistic by $(N-1)/(n-1)$. Null calibration incorporates this factor and uses the observed sample size. The first two traces are calculated directly, while higher moments use 60 Rademacher probes.

\textbf{Irregular coordinates (NUFFT):}
For $n$ observed locations $s_i$, retain $n'$ Fourier modes with angular frequencies $\boldsymbol\omega_k$ on an auxiliary regular grid. Locations and frequencies are vectors in the same spatial dimension. Define
\[
    \mathbf U\in\mathbb C^{n\times n'},
    \qquad
    U_{ik}=\exp(\mathrm{i}\boldsymbol\omega_k^\top s_i),
    \qquad
    \boldsymbol\mu\in\mathbb R^{n'}.
\]
The resulting finite Fourier kernel is
\[
    \mathbf K_N
    =\frac{1}{n'}\mathbf U
      \operatorname{diag}(\boldsymbol\mu)\mathbf U^{\mathsf H}
    \in\mathbb R^{n\times n},
\]
with conjugate-paired modes and equal weights ensuring a real symmetric matrix. The accuracy of this approximation to the intended spatial kernel depends on the chosen period and retained frequencies.

For $\mathbf v\in\mathbb R^n$, a type-1 NUFFT computes $\mathbf U^{\mathsf H}\mathbf v\in\mathbb C^{n'}$. After weighting these coefficients, the adjoint type-2 NUFFT returns the result at the $n$ observed locations:
\[
    \mathbf K_N\mathbf v
    =\frac{1}{n'}\underbrace{\mathbf U
      \left[
        \boldsymbol\mu\odot
        \underbrace{\mathbf U^{\mathsf H}\mathbf v}_{\text{type-1}}
      \right]}_{\text{type-2}}
    \in\mathbb R^n.
\]
At fixed numerical accuracy, this calculation costs $\mathcal O(n+n'\log n')$ operations. 

Because the columns of $\mathbf U$ are generally not orthogonal, $\boldsymbol\mu$ contains Fourier weights rather than the eigenvalues of $\mathbf K_N$. The first two traces can nevertheless be calculated without eigendecomposition. Define
\[
    \mathbf G=\mathbf U^{\mathsf H}\mathbf U\in\mathbb C^{n'\times n'},
    \qquad
    \phi(\boldsymbol\Delta)=\frac{1}{n}\sum_{i=1}^n
        \exp(\mathrm{i}\boldsymbol\Delta^\top s_i).
\]
Then $G_{kk'}=n\phi(\boldsymbol\omega_{k'}-\boldsymbol\omega_k)$ and
\[
    \boldsymbol\Psi\in\mathbb R^{n'\times n'},
    \qquad
    \Psi_{kk'}=\left|\phi(\boldsymbol\omega_{k'}-\boldsymbol\omega_k)\right|^2.
\]
The uncentered traces are
\begin{align*}
    \operatorname{tr}(\mathbf K_N)
    &=\frac{1}{n'}\operatorname{tr}\!\left[
        \mathbf G\operatorname{diag}(\boldsymbol\mu)\right]
      =\frac{n}{n'}\sum_{k=1}^{n'}\mu_k,\\
    \operatorname{tr}(\mathbf K_N^2)
    &=\frac{n^2}{(n')^2}\boldsymbol\mu^\top
        \boldsymbol\Psi\boldsymbol\mu.
\end{align*}
Because $\boldsymbol\Psi$ depends only on frequency differences, its multiplication can be evaluated as a Toeplitz convolution. We compute $\phi$ on the enlarged frequency-difference grid and zero-pad the weights to avoid circular wrap-around. This gives $\mathcal O(n+n'\log n')$ cost at fixed numerical accuracy without forming $\mathbf G$ or $\boldsymbol\Psi$.

For the centered kernel $\tilde{\mathbf K}_N=\mathbf H\mathbf K_N\mathbf H$, where $\mathbf H=\mathbf I_n-\mathbf1\mathbf1^\top/n\in\mathbb R^{n\times n}$ and $\mathbf1\in\mathbb R^n$, the spectral power sums used for null calibration are
\begin{align*}
    c_1
    &=\operatorname{tr}(\mathbf K_N)
      -\frac{\mathbf1^\top\mathbf K_N\mathbf1}{n},\\
    c_2
    &=\operatorname{tr}(\mathbf K_N^2)
      -\frac{2\,\mathbf1^\top\mathbf K_N^2\mathbf1}{n}
      +\frac{(\mathbf1^\top\mathbf K_N\mathbf1)^2}{n^2}.
\end{align*}
A single multiplication by $\mathbf K_N$ supplies both corrections, since $\mathbf1^\top\mathbf K_N^2\mathbf1=\|\mathbf K_N\mathbf1\|^2$. Unlike on a complete periodic grid, centering at irregular locations generally requires applying $\mathbf H$ before and after the kernel multiplication, rather than removing only the zero-frequency weight.

NUFFT calibration computes $c_1$ and $c_2$ from the centered trace formulas above. For moment matching after empirical standardization, 60 Rademacher probes estimate the remaining spectral summaries of $\mathbf K_\circ$, reusing two kernel applications per probe. Probes are projected onto the mean-zero subspace, and their estimates are normalized by the combined squared probe norm. The unshifted second trace $c_2$ is retained for \(R\)-test calibration.

\textbf{Remark.} For graph-based kernels such as Moran's \(I\), CAR and the graph Laplacian, the NUFFT implementation uses a neighborhood defined on the auxiliary grid. This defines a different operator from a kernel built on a $k$-nearest-neighbor graph at the observed locations. Their quadratic statistics therefore need not agree.

\subsection{Summary of \(Q\)-test complexity and limitations}\label{supidx:n4:complexity}

Table~\ref{tab:q_backends} summarizes the computational complexities and limitations of each scaling strategy. For continuous spatial kernels, \texttt{SONIC} provides a NUFFT implementation that can be paired with four-moment calibration, using \(\mathcal{O}(n+n')\) memory and \(\mathcal{O}(n+n'\log n')\) time at fixed numerical accuracy. For graph-based kernels, sparse-precision solves have costs determined by graph structure and fill-in, with the rates below applying to suitable two-dimensional lattice-like graphs.

\begin{table}[h]
\centering
\caption{Computational complexity of \(Q_n\) scaling strategies. Variables: \(n\) (number of samples), \(n'\) (auxiliary k-grid size, typically \(n' \gtrsim n\)), \(k\) (sparse-graph degree), \(r\) (kernel rank), \(K\) (number of integration nodes).}
\label{tab:q_backends}
\footnotesize
\begin{tabular}{p{3.5cm}p{2.4cm}p{1.4cm}p{2.2cm}p{3.9cm}}
\toprule
\textbf{Backend} & \textbf{Precompute} & \textbf{Memory} & \textbf{Time per \(Q\)} & \textbf{Key limitation}\\
\midrule
Dense + Davies (baseline) & \(\mathcal{O}(n^{3})\) & \(\mathcal{O}(n^{2})\) & \(\mathcal{O}(n^{2}) + \mathcal{O}(Kn)\) & Intractable for large \(n\) \\
Dense + moment approximation\(^{\dagger}\) & \(\mathcal{O}(n^{2})\) & \(\mathcal{O}(n^{2})\) & \(\mathcal{O}(n^{2})\) & Spectrum-dependent tail accuracy \\
\midrule
\textit{Kernel-specific} & & & \\
Sparse kernel via \(k\)-NN & \(\mathcal{O}(nk)\) & \(\mathcal{O}(nk)\) & \(\mathcal{O}(nk)\) & Possible spectral cancellation \\
Low-rank kernel (rank-\(r\)) & \(\mathcal{O}(nr^2 + r^{3})^{\ddagger}\)  & \(\mathcal{O}(nr)\) & \(\mathcal{O}(nr)\) & Omitted signal directions \\
Sparse precision (2D lattice-like) & \(\mathcal{O}(n^{3/2})\) & \(\mathcal{O}(n\log n)\) & \(\mathcal{O}(n\log n)\) & Sparse triangular solver is hard to parallelize\\
\midrule
\textit{Stationary kernels} & & & \\
FFT  (full spectrum) & \(\mathcal{O}(n\log n)\) & \(\mathcal{O}(n)\) & \(\mathcal{O}(n\log n)\) & Requires periodic grid \\
NUFFT & \(\mathcal{O}(n+n'\log n')\) & \(\mathcal{O}(n+n')\) & \(\mathcal{O}(n+n'\log n')\) & Mismatch with graph kernels\(^{\S}\) \\
\bottomrule
\end{tabular}

\vspace{1ex}
{\raggedright\scriptsize
Per-\(Q\) (i.e., per-feature) costs include the kernel matvec (dominant, but accelerable by modern hardware) plus a null evaluation step: numerical integration for Davies (linear in \(n\) but with a large constant \(K\)), or constant-time tail evaluation for a fitted chi-square, beta or normal distribution. Parameters required for the evaluation (e.g., kernel spectrum or cumulants) are cached once in the \textbf{Precompute} step and broadcast across all features. Graph-construction costs are excluded, and sparse-factorization rates assume suitable two-dimensional graph structure and ordering. Transform accuracy and the number of stochastic trace probes are held fixed.\\
\(^{\dagger}\) \textbf{Moment approximations:} Four-moment fits and Welch--Satterthwaite have the same asymptotic complexity but different constants, using four and two spectral summaries, respectively. The NUFFT trace and probe calculations are described above. \\
\(^{\ddagger}\) \textbf{Low-rank cost:} The listed cost assumes a supplied factor; its construction depends on the approximation method.\\
\(^{\S}\) \textbf{NUFFT for graph-based kernels:} NUFFT evaluates the selected stationary Fourier representation. Applying it to an arbitrary graph kernel requires a separate approximation, whose accuracy depends on how well this representation matches the graph operator.\\
}
\end{table}

\FloatBarrier
\subsection{Bivariate \(R\)-statistics}\label{supidx:n4:bivariate}

The bilinear form $R=\mathbf z_x^\top\mathbf K\mathbf z_y$ uses the same kernel operations as a \(Q\)-statistic. For two blocks of standardized features, $\mathbf Z_A\in\mathbb R^{n\times p}$ and $\mathbf Z_B\in\mathbb R^{n\times q}$, we compute
\[
    \mathbf R_{AB}
    =\mathbf Z_A^\top(\mathbf K\mathbf Z_B)
    \in\mathbb R^{p\times q}.
\]
Each column of $\mathbf K\mathbf Z_B$ is reused across all features in the first block. Sparse kernels and sparse-precision solves therefore carry over directly from the univariate calculation. Blocking requires $\mathcal O(n(p+q)+pq)$ working memory in addition to kernel storage, while a complete set of pairwise scores still grows quadratically with the number of features. For sparse expression matrices, centering can be implemented through corrections involving $\mathbf K\mathbf1$, avoiding a dense copy of the full centered expression matrix.

Fourier implementations replace the squared coefficients in $Q_n$ with cross-products. With $\widehat{\mathbf z}=\mathcal F(\mathbf z)$ on a complete periodic grid,
\[
    R=\frac{1}{n}\operatorname{Re}
      \sum_{k=1}^n\mu_k\,
      \overline{\widehat z_{x,k}}\,\widehat z_{y,k}.
\]
For the finite Fourier kernel $\mathbf K_N$, set $\mathbf a_x=\mathbf U^{\mathsf H}\mathbf z_x$ and $\mathbf a_y=\mathbf U^{\mathsf H}\mathbf z_y$, both in $\mathbb C^{n'}$. Then
\[
    R=\frac{1}{n'}\operatorname{Re}
      \sum_{k=1}^{n'}\mu_k\,\overline{a_{x,k}}a_{y,k}.
\]
A pair of type-1 NUFFTs thus suffices, without a type-2 transform back to spatial coordinates. Cached coefficients can be reused across pairs.

Under independent Gaussian features standardized by their sample means and variances, $\mathbb E[\mathbf z_x\mathbf z_x^\top]=\mathbb E[\mathbf z_y\mathbf z_y^\top]=\mathbf H$. Consequently,
\[
    \mathbb E[R]=0,\qquad
    \operatorname{Var}(R)=\operatorname{tr}(\tilde{\mathbf K}^2)=c_2.
\]
The normal approximation for $R/\sqrt{c_2}$ is suitable when many modes contribute and none dominates the variance. The implementation uses two-sided normal-tail $P$ values, reusing $c_2$ from \(Q\)-test calibration. For a positive semidefinite kernel, the descriptive coefficient $R/\sqrt{Q_xQ_y}$ lies in $[-1,1]$ when both denominators are positive. Significance is computed from the unnormalized \(R\).

\subsection{Frequency-domain Gaussian-process fitting for conditional differential usage}\label{supidx:n4:gp}

SPLISOSM's conditional differential-usage test removes smooth spatial trends before testing associations between isoform usage and candidate regulators to control for spatial autocorrelation\autocite{su2026mapping}. The main computational cost is fitting individual spatial Gaussian processes (GPs) across all tested genes and regulators. A stationary covariance makes both likelihood evaluation and residual calculation efficient in the frequency domain.

For a centered response $\mathbf y\in\mathbb R^n$, consider $\mathbf y=\mathbf f+\boldsymbol\epsilon$, where $\mathbf f\sim\mathcal N(\mathbf0,\mathbf K_{\mathrm{GP}})$ and $\boldsymbol\epsilon\sim\mathcal N(\mathbf0,\sigma_\epsilon^2\mathbf I)$. Here, $\mathbf K_{\mathrm{GP}}$ is the spatial covariance, including its signal variance, and $\sigma_\epsilon^2$ is the independent noise variance. On a complete periodic grid, let $\lambda_k^{\mathrm{GP}}$ be its Fourier eigenvalues and $\widehat{\mathbf y}=\mathcal F(\mathbf y)$. Twice the negative log marginal likelihood, up to a constant, is
\[
    -2\log L
    =\sum_{k=1}^n
      \left[
        \log(\lambda_k^{\mathrm{GP}}+\sigma_\epsilon^2)
        +\frac{|\widehat y_k|^2}
        {n(\lambda_k^{\mathrm{GP}}+\sigma_\epsilon^2)}
      \right].
\]
Thus the determinant and inverse reduce to scalar operations at each frequency. We fit positive hyperparameters on the log scale with bounded optimization. With a fixed kernel length scale, the spectrum and response transform are cached, making each likelihood evaluation $\mathcal O(n)$ after an $\mathcal O(n\log n)$ transform. Optimizing the length scale requires updating the kernel spectrum. 

The residual after subtracting the GP posterior mean is
\[
    \mathbf e
    =\sigma_\epsilon^2
      (\mathbf K_{\mathrm{GP}}+\sigma_\epsilon^2\mathbf I)^{-1}\mathbf y
    =\mathcal F^{-1}\!\left[
      \left(\frac{\sigma_\epsilon^2}
      {\boldsymbol\lambda^{\mathrm{GP}}+\sigma_\epsilon^2}\right)
      \odot\widehat{\mathbf y}\right],
\]
where division is elementwise. Fitting and residualization avoid the cubic factorization and quadratic storage of a dense GP. These formulas apply to the periodic raster used by the FFT implementation. Missing raster entries are mean-filled, or equivalently set to zero after centering.

For irregular coordinates, the NUFFT covariance is applied through the operator in Supplementary Note~\ref{supidx:n4:fourier}. Its Fourier weights are not eigenvalues, so the inverse requires conjugate-gradient solves. Hyperparameter fitting uses leading eigenpairs of this implicit covariance together with a tail approximation matching its remaining trace and squared trace. The tail contribution to the likelihood also uses the response energy outside the retained eigenspace. The Atera analysis retained 32 eigenpairs for this step. Residuals are then computed with the full NUFFT covariance operator and the fitted parameters. Each solver iteration costs $\mathcal O(n+n'\log n')$ at fixed transform accuracy, with the iteration count depending on conditioning and the requested tolerance.

\subsection{Shared frequency features and pattern comparison}\label{supidx:n4:comparison}

Pattern comparison reuses the Fourier transforms above but retains power rather than multiplying by a detection kernel. For one gene in section $i$, let $x_j^{(i)}$ denote expression at location $s_j^{(i)}$, with $n_i$ locations and sample mean $\bar x_i$. The centered Fourier coefficients and power are
\begin{equation}\label{eq:supp_physical_spectrum}
    \widehat f_i(\boldsymbol\omega)
    =\sum_{j=1}^{n_i}(x_j^{(i)}-\bar x_i)
      \exp(-2\pi\mathrm{i}\boldsymbol\omega^{\mathsf T}s_j^{(i)}),
    \qquad
    P_i(\boldsymbol\omega)=n_i^{-2}|\widehat f_i(\boldsymbol\omega)|^2,
\end{equation}
where $s_j^{(i)}\in\mathbb R^2$ is measured in micrometers and $\boldsymbol\omega\in\mathbb R^2$ in cycles per micrometer. We normalize power by the squared number of observed locations to express Fourier amplitudes on an average-per-location scale. Unobserved raster entries are excluded from this count. This physical frequency axis allows sections with different shapes, spot counts and resolutions to contribute comparable features. Only frequencies supported by all sections are used in comparison.

Before binning or orientation alignment, we compute anisotropy from the raw two-dimensional spectrum as
\begin{equation}\label{eq:supp_spectral_anisotropy}
    A_i=\frac{\left|\sum_{k:\,\boldsymbol\omega_{ik}\neq\mathbf0}
        P_i(\boldsymbol\omega_{ik})\exp(2\mathrm{i}\theta_{ik})\right|}
        {\sum_{k:\,\boldsymbol\omega_{ik}\neq\mathbf0}P_i(\boldsymbol\omega_{ik})},
\end{equation}
where $\theta_{ik}$ is the angle of the evaluated frequency $\boldsymbol\omega_{ik}$. The factor of two treats opposite directions as the same axis. Values range from 0 for angularly balanced power to 1 for power concentrated along one axis. The sum covers the full Fourier plane, excluding zero frequency. For transforms storing only half of the spectrum, we reconstruct the remaining half before summation. Spectra with no power away from zero frequency have undefined anisotropy.

To limit the memory cost of dense Fourier coefficients, we process genes in chunks. Each chunk's normalized power is immediately averaged into shared radial bins, or resampled onto a common radius--angle grid when directional information is needed. Radial binning excludes zero frequency and merges empty bins until all sections contribute. Directional comparison estimates one rotation per section from the transcriptome-wide geometric-mean spectrum, then applies that rotation to every gene. A first pass accumulates the spectrum used for rotation estimation and a second pass extracts the aligned profiles. For statistical testing and joint embedding, we only cache binned power profiles containing $mGB$ entries for $m$ sections, $G$ genes and $B$ frequency features.

\subsection{Differential-frequency tests}\label{supidx:n4:df}

DF and DF-norm use the same cached frequency features. DF retains spectral amplitude, whereas DF-norm first divides each nonzero profile by its total power.

For one gene, let $\mathbf U\in\mathbb R^{m\times B}$ be log-power profiles and $\mathbf A\in\mathbb R^{m\times p}$ be the intercept and sample-level predictors. A single matrix factorization supplies the least-squares coefficients, residuals and contrast,
\[
    \widehat{\boldsymbol\beta}=\mathbf A^+\mathbf U,\qquad
    \mathbf E=\mathbf U-\mathbf A\widehat{\boldsymbol\beta},\qquad
    \boldsymbol\Delta=\widehat{\boldsymbol\beta}^{\top}\mathbf c
    \in\mathbb R^B,
\]
where $\mathbf c\in\mathbb R^p$ specifies an estimable contrast. We batch these operations across genes. A binary design gives the difference between group-mean log spectra, while a continuous design gives the spectral change per unit of the covariate. Both use $T^2=\boldsymbol\Delta^\top\mathbf W\boldsymbol\Delta$, with $\mathbf W=\operatorname{diag}(w_1,\ldots,w_B)$ and $w_b=1/B$ in our analyses.

Frequency bins can be strongly correlated. To account for this dependence, we pool residual cross-products across genes,
\[
    \widehat{\boldsymbol\Sigma}
    =\frac{\sum_{g=1}^G\mathbf E_g^\top\mathbf E_g}
    {G\{m-\operatorname{rank}(\mathbf A)\}}.
\]
This shares information across genes when biological replication is limited, under the assumption of a common residual covariance. The null distribution is thus given by
\[
    T^2\ \approx\ \sum_{\ell=1}^B\gamma_\ell\chi^2_{1,\ell},
    \qquad
    \gamma_\ell=\operatorname{eig}_\ell
      \!\left(v\mathbf W^{1/2}\widehat{\boldsymbol\Sigma}
                    \mathbf W^{1/2}\right),
    \qquad
    v=\mathbf c^\top(\mathbf A^\top\mathbf A)^+\mathbf c
\]
using a Gaussian approximation to the contrast.

We diagonalize the $B\times B$ weighted covariance once, cache the four spectral power sums needed for Liu's approximation, and evaluate upper-tail $P$ values across genes. For a common design, residual covariance accumulation costs $\mathcal O(mGB^2)$ and eigendecomposition costs $\mathcal O(B^3)$. Residual cross-products are accumulated in gene blocks to avoid storing the full log-profile and residual arrays. Once contrasts and calibration parameters are available, each statistic costs $\mathcal O(B)$ and its approximate $P$ value costs constant time.

When some sections need to be excluded for some genes (e.g., due to expression filters), as in the glioma analysis, each gene is fitted on its eligible rows of $\mathbf A$. Residual cross-products are pooled over these fits and divided by the sum of their residual degrees of freedom. The contrast scale $v$ uses the gene's own design, while the covariance spectrum is shared. Genes without sufficient eligible sections or an estimable contrast are excluded from the test.

\clearpage
\section{Interpreting the spectral representation}\label{supnote:5}

\subsection{Purpose and information retained by the spectral representation}

For comparison, we want the representation to remain stable when the same observed pattern is shifted or rotated within a section. SONIC summarizes each expression map $\mathbf x^{(i)}\in\mathbb R^{n_i}$, after mean centering, as a vector of binned Fourier power, $\mathbf p_i=(P_i(1),\ldots,P_i(B))^{\mathsf T}$. To see what this representation preserves, write the Fourier coefficients from \eqref{eq:supp_physical_spectrum} in magnitude and phase form,
\[
    \widehat f_i(\boldsymbol\omega)
    =n_i\sqrt{P_i(\boldsymbol\omega)}
      \exp\!\bigl(\mathrm{i}\varphi_i(\boldsymbol\omega)\bigr),
\]
where $\varphi_i$ is phase and $\boldsymbol\omega\in\mathbb R^2$ is physical frequency in cycles per micrometer. On a complete periodic grid, reconstructing the mean-centered map requires both power and phase:
\[
    x_j^{(i)}-\bar x_i
    =\sum_{k=1}^{n_i}
      \sqrt{P_i(\boldsymbol\omega_{ik})}
      \exp\!\left[
        \mathrm{i}\bigl(2\pi\boldsymbol\omega_{ik}^{\mathsf T}s_j^{(i)}
        +\varphi_i(\boldsymbol\omega_{ik})\bigr)
      \right],
\]
where $\bar x_i$ is the sample mean, $s_j^{(i)}$ are the observed locations, and $\boldsymbol\omega_{ik}$ are the grid's Fourier frequencies. The phase term determines how the Fourier components combine at each location. Discarding it therefore makes power invariant to translation. Radial averaging additionally makes it invariant to rotation.

The measured power spectrum also depends on tissue size, its boundaries and the spatial resolution. On a regular grid, the discrete Fourier frequencies along dimension $d$ range from approximately $-1/(2\Delta s_d)$ to $1/(2\Delta s_d)$, in steps of $1/L_d$. Here, $L_d$ is the periodic grid length, which sets the frequency spacing, and $\Delta s_d$ is the sampling interval, which sets the Nyquist limit.

Binning addresses the practical memory cost of the transform. A sparse expression map on a grid of $10^6$ locations produces $10^6$ Fourier coefficients. Storing the full complex transform in double precision requires $10^6\times16=16\times10^6$ bytes, or 16\,MB per gene. For one sample with 10,000 genes, this amounts to 160\,GB. Retaining double-precision power at every frequency still requires 80\,GB. With 64 power bins per gene, the profiles for all 10,000 genes require only 5.12\,MB. This compression makes transcriptome-wide comparisons in large datasets such as MOSTA\autocite{chen2022spatiotemporal} and HESTA\autocite{pan2026spatiotemporal} practical on personal computers.

SONIC is conceptually related to graph-based harmonic analyses, which also represent patterns as combinations of spectral modes\autocite{chang2024graph,maher2024harmonic}. However, their frequencies and modes depend on the sample-specific graph, so coefficients are not directly comparable across samples. Physical Fourier frequencies provide a shared axis in inverse-distance units.

\subsection{Spectral similarity and the bivariate \(R\)-statistic}
\label{supnote:5.similarity}

We quantify spatial pattern similarity in two ways. The bivariate \(R\)-statistic measures colocalization between maps at shared locations, while spectral distance compares their frequency content and can be used across samples.

For real standardized maps $\mathbf z_x,\mathbf z_y\in\mathbb R^n$ on the same periodic grid, let $\widehat{\mathbf z}_x=\mathcal F(\mathbf z_x)$ and $\widehat{\mathbf z}_y=\mathcal F(\mathbf z_y)$ use the unnormalized transform above. With a real symmetric stationary kernel whose eigenvalues are $\mu_k$,
\[
\begin{aligned}
    R&=\mathbf z_x^{\mathsf T}\mathbf K\mathbf z_y
      =\frac1n\operatorname{Re}\sum_{k=1}^n
        \mu_k\,\overline{\widehat z_{x,k}}\,\widehat z_{y,k}\\
     &=\frac1n\sum_{k=1}^n\mu_k
       \sqrt{P_x(k)P_y(k)}
       \cos\!\bigl(\varphi_y(k)-\varphi_x(k)\bigr),
\end{aligned}
\]
where $P_x(k)=|\widehat z_{x,k}|^2$ and $\varphi_x(k)$ are the unbinned power and phase of the first map, with analogous definitions for the second. Thus, \(R\) depends on both spectral magnitude and relative phase, and the kernel weights $\mu_k$ set the contribution of each frequency.

Similarly, for nonnegative kernel weights, the kernel-weighted distance
\[
    d_{\mathbf K}^{\,2}(\mathbf z_x,\mathbf z_y)
    =(\mathbf z_x-\mathbf z_y)^{\mathsf T}
      \mathbf K(\mathbf z_x-\mathbf z_y)
\]
relates to the distance between square-root power spectra as follows:
\[
\begin{aligned}
    d_{\mathbf K}^{\,2}(\mathbf z_x,\mathbf z_y)
    &=\frac1n\sum_{k=1}^n\mu_k
      \left(\sqrt{P_x(k)}-\sqrt{P_y(k)}\right)^2\\
    &\quad+\frac2n\sum_{k=1}^n\mu_k
      \sqrt{P_x(k)P_y(k)}
      \left[1-\cos\!\bigl(\varphi_y(k)-\varphi_x(k)\bigr)\right].
\end{aligned}
\]
The first term compares the square-root power spectra, whereas the second captures their phase difference.

\subsection{Spectral covariate correction and cell-type-specific models}
\label{supnote:5.covariates}

Because the spectral representation omits precise location, it is less interpretable than the original spatial map, where variation can be traced to individual cells or regions of interest. Nevertheless, the Fourier transform naturally connects the spatial and spectral domains, since convolution in space becomes multiplication in frequency (and addition in log power). Here we illustrate the connection in the context of cell-type-specific models, and explain how SONIC's spectral covariate adjustment removes unwanted variation from the representation.

We compare three models of how cell-type proportions can affect expression. Let $m_i(s)$ denote expected expression on a common sampling scale and $\pi_{ic}(s)$ the proportion of cell type $c$ in sample $i$, both varying with location $s$. We treat $\pi_{ic}(s)$ as known and seek to remove expression variation associated with cell-type composition.

In the first additive spatial model, each cell type has a constant baseline expression $b_{ic}$,
\[
    m_i(s)=\sum_{c=1}^{C}b_{ic}\pi_{ic}(s)+r_i(s),
\]
and the residual expression $r_i(s)$ is the target signal that may reflect local environmental influences. Regressing expression on cell-type proportions removes the fitted composition effect. The Fourier transform preserves this additive relationship:
\begin{equation}\label{eq:ct_add}
    \widehat m_i(\boldsymbol\omega)
    =\sum_{c=1}^{C}b_{ic}\widehat\pi_{ic}(\boldsymbol\omega)
      +\widehat r_i(\boldsymbol\omega).
\end{equation}

Cell-type-specific models from Supplementary Note~\ref{supnote:2.ctsv} allow expression within each type to vary with location:
\[
    m_i(s)=\sum_{c=1}^{C}\pi_{ic}(s)\mu_{ic}(s),
\]
where $\mu_{ic}(s)$ is the mean expression within type $c$ we wish to infer. Each expression profile is multiplied by its local cell-type proportion. On an $n_i$-location periodic grid, these spatial products become convolutions of Fourier coefficients:
\begin{equation}\label{eq:ct_mult}
    \widehat m_i
    =\frac1{n_i}\sum_{c=1}^{C}
      \widehat\pi_{ic}*\widehat\mu_{ic},
\end{equation}
where $*$ denotes circular convolution and the factor $1/n_i$ follows the same unnormalized transform convention. The power of the mixture also contains cross-terms between cell types.

\eqref{eq:ct_add} and \eqref{eq:ct_mult} suggest that we may model covariate effects directly in the spectral domain. 
For example, assume multiplicative effects in frequency,
\begin{equation}\label{eq:ct_conv}
    \widehat m_i(\boldsymbol\omega)
    =\sum_{c=1}^{C}
      \widehat\pi_{ic}(\boldsymbol\omega)
      \widehat h_{ic}(\boldsymbol\omega).
\end{equation}
Under the same unnormalized transform convention, this corresponds to spatial convolution,
\[
    m_i(s)=\sum_{c=1}^{C}(\pi_{ic}*h_{ic})(s),
\]
where $h_{ic}$ describes the spatial reach of each cell type's contribution, potentially reflecting local cell--cell communication. The power of this sum contains cross-terms between cell types.

Motivated by this spatial--spectral connection, SONIC uses a multiplicative model for covariate effects on power, fitted by additive regression in log power. For positive powers, omitting the numerical offset,
\[
    \log P_i(b)
    =\alpha_i+\sum_{c=1}^{C}\beta_{ic}\log C_{ic}(b)+e_i(b),
\]
where \(C_{ic}(b)\) is the binned power of the cell-type proportion map \(\pi_{ic}\). We retain the exponentiated residual, \(P_i^{\mathrm{adj}}(b)=\exp(\widehat e_i(b))\), as the adjusted profile. For an unbinned fit, this correction corresponds to convolution with a spatial filter \(h_i\), whose Fourier coefficients are
\[
    \widehat h_i(\boldsymbol\omega)
    =\exp(-\widehat\alpha_i/2)
      \prod_{c=1}^{C}C_{ic}(\boldsymbol\omega)^{-\widehat\beta_{ic}/2}.
\]
Applying this filter to the centered expression field \(f_i\) produces the adjusted power
\[
    n_i^{-2}|\widehat{h_i*f_i}(\boldsymbol\omega)|^2
    =P_i^{\mathrm{adj}}(\boldsymbol\omega),
\]
where \(C_{ic}(\boldsymbol\omega)\) is the unbinned covariate power. Each cell-type term describes the spatial scales of that population's distribution. For \(\widehat\beta_{ic}>0\), its contribution to the filter downweights frequencies where the population's distribution has greater spectral power, adjusting for expression structure associated with its regional enrichment or clustering.

\clearpage
\printbibliography[title={Supplementary References}]

\end{refsection}

\end{document}